\documentclass[aps,10pt,pra,superscriptaddress,showpacs,floatfix,longbibliography,notitlepage]{revtex4-2}
\usepackage{dsfont}
\usepackage{cprotect}
\usepackage{graphicx}
\usepackage{amsmath}
\usepackage{comment}
\usepackage{slashed}
\usepackage{amssymb}

\usepackage{varioref}       
\usepackage{hyperref}      
\usepackage{cleveref}
\usepackage{soul}

\crefname{table}{Tab.}{Tabs.}
\Crefname{table}{Table}{Tables}
\crefname{figure}{Fig.}{Figs.}
\Crefname{figure}{Figure}{Figures}
\crefname{section}{Sec.}{Secs.}
\Crefname{section}{Section}{Sections}
\crefname{equation}{Eq.}{Eqs.}
\Crefname{equation}{Equation}{Equations}
\crefname{appendix}{App.}{Apps.}
\Crefname{appendix}{Appendix}{Appendices}

\usepackage{physics}
\usepackage{bm}
\usepackage{subcaption}
\hypersetup{colorlinks=true,linkcolor=cyan,citecolor=magenta,filecolor=magenta,urlcolor=cyan,runcolor=cyan}
\usepackage{xcolor}
\newcommand{\beq}{\begin{equation}}
\newcommand{\eeq}{\end{equation}}
\newcommand{\beqnn}{\begin{equation*}}
\newcommand{\eeqnn}{\end{equation*}}
\newcommand{\bea}{\begin{eqnarray}}
\newcommand{\eea}{\end{eqnarray}}
\newcommand{\beann}{\begin{eqnarray*}}
\newcommand{\eeann}{\end{eqnarray*}}
\newcommand{\bes} {\begin{subequations}}
\newcommand{\ees} {\end{subequations}}
\newcommand{\e}{\text{e}}

\renewcommand{\vec}{\bm}
\newcommand{\dtau}{\mathrm{d}\tau}

\newcommand{\lap}[1]{\mathcal{L}\{ #1 \}}
\newcommand{\iq}{\vec{i}_q}
\newcommand{\Siq}{S_{\vec{i}_q}}
\newcommand{\Sip}{S_{\vec{i}_p}}
\newcommand{\Sir}{S_{\vec{i}_r}}
\newcommand{\avg}[1]{ \langle #1 \rangle}
\newcommand{\diag}[1]{\text{diag}(#1)}

\newcommand{\wtA}{\widetilde{A}}
\newcommand{\wtB}{\widetilde{B}}
\newcommand{\wtlam}{\widetilde{\Lambda}}
\newcommand{\wtP}{\widetilde{P}}
\newcommand{\wtD}{\widetilde{D}}
\newcommand{\Hdiag}{H_{\text{diag}}}
\newcommand{\Hoffdiag}{H_{\text{offdiag}}}
\newcommand*{\estimates}{\mathrel{\widehat=}}
\newcommand{\cb}[1]{\hat{b}^{\dagger}_{#1}}
\newcommand{\db}[1]{\hat{b}_{#1}}

\newcommand{\no}[1]{\hat{n}_{#1}}

\usepackage{amsthm}
\usepackage{thmtools,thm-restate}
\newtheorem{theorem}{Theorem}

\newtheorem{corollary}{Corollary}

\newtheorem{definition}{Definition}

\usepackage{bbm}
\usepackage[normalem]{ulem}

\begin{document}
\title{Advanced measurement techniques in  quantum Monte Carlo: \\The permutation matrix representation approach}
\author{Nic Ezzell}
\affiliation{Department of Physics and Astronomy and Center for Quantum Information Science \& Technology, University of Southern California, Los Angeles, California 90089, USA}
\affiliation{Information Sciences Institute, University of Southern California, Marina del Rey, California 90292, USA}
\author{Itay Hen}
\email{itayhen@isi.edu}
\affiliation{Department of Physics and Astronomy and Center for Quantum Information Science \& Technology, University of Southern California, Los Angeles, California 90089, USA}
\affiliation{Information Sciences Institute, University of Southern California, Marina del Rey, California 90292, USA}
\begin{abstract}
    \noindent In a typical finite temperature quantum Monte Carlo (QMC) simulation, estimators for simple static observables such as specific heat and magnetization are known.  With a great deal of system-specific manual labor, one can sometimes also derive more complicated non-local or even dynamic observable estimators. { In contrast, we show that arbitrary static observables can be estimated within the permutation matrix representation (PMR) flavor for any Hamiltonian.} We then generalize these results to general imaginary-time correlation functions and non-trivial integrated susceptibilities thereof. We demonstrate the practical versatility of our method by estimating various non-local, random observables for the transverse-field Ising model on a square lattice and a toy random model.  
\end{abstract}
\maketitle

\section{Introduction }\label{Xsec1-1}\label{Xsec1}

Since their inception~\cite{suzuki1976relationship, hirsch1982monte}, quantum Monte Carlo (QMC) techniques have become indispensable in the study of many-body quantum systems~\cite{landau2015GuideMonteCarlo, suzuki2012quantum, de1985monte, assaad2008world}. In finite-temperature QMC---the setting of the present work---the overarching goal is to compute thermal expectation values of observables. The central idea is to express the partition function as a sum over efficiently computable weights that may subsequently be sampled using Markov chain Monte Carlo~\cite{mareschal2021early}. Modern QMC frameworks achieving this include continuous-time worldline approaches~\cite{beard1996simulations,prokof1998exact, evertz2003loop, kawashima2004recent}, diagrammatic determinantal methods~\cite{rubtsov2005continuous, werner2006continuous, gull2011continuous, rombouts1999quantum, burovski2006fermi, iazzi2015efficient}, and the stochastic series expansion (SSE)~\cite{sandvik1992GeneralizationHandscombQuantum, sandvik1999StochasticSeriesExpansion, sandvik2019stochastic, melko2013stochastic}.

{
In this work, we derive general estimators for the permutation matrix representation (PMR) flavor of QMC, or PMR-QMC. The PMR-QMC algorithm is built upon two conceptual steps. First, one expresses the quantum Hamiltonian as a sum of diagonal matrices multiplied by permutation matrices~\cite{gupta2020PermutationMatrixRepresentation}. For an appropriately chosen set of permutations--a notion that we clarify and broaden in this work--any product of non-identity permutations necessarily has no fixed points, enforcing a no-branching condition. Second, one performs an off-diagonal series expansion of the partition function~\cite{hen2018OffdiagonalSeriesExpansion,albash2017OffdiagonalExpansionQuantum}. Although this expansion begins analogously to SSE, with a Taylor series and an operator-string representation, it proceeds via a re-summation in which infinitely many diagonal contributions are consolidated into a divided difference of the exponential (DDE). This re-summation yields several important advantages.

First, it enables the quantum partition function to be viewed as an expansion about the classical contribution to the Hamiltonian, thereby allowing PMR-QMC to interpolate smoothly between classical and quantum regimes as the off-diagonal strength is tuned. Second, it dramatically accelerates the equilibration of glassy systems, which equilibrate substantially faster in PMR-QMC than in path-integral QMC or SSE~\cite{albash2017OffdiagonalExpansionQuantum,gupta2020CalculatingDividedDifferences}. This speed-up arises both from the reduction to the classical partition function in appropriate limits and from the stability and efficiency with which the DDE can be updated during simulation~\cite{gupta2020CalculatingDividedDifferences}. Third, the DDE possesses a rich collection of structural identities~\cite{zeng2025inequalities} that prove instrumental in deriving non-trivial PMR-QMC estimators~\cite{ezzell2025universal}.}
Finally, the construction of QMC algorithms within PMR-QMC has been shown to be highly automatable~\cite{barash2024QuantumMonteCarlo}.

In particular, there now exist deterministic and canonical procedures for writing arbitrary spin-1/2 Hamiltonians~\cite{barash2024QuantumMonteCarlo}, Bose--Hubbard models on general lattices~\cite{akaturk2024quantum}, arbitrary higher-spin Hamiltonians~\cite{babakhani2025quantum}, and fermionic Hamiltonians~\cite{babakhani2025quantum} in PMR form. Once expressed in this representation, one can automatically identify the so-called \emph{fundamental cycles} from which ergodic and detailed-balance-preserving QMC update rules are constructed. {Although recent developments in SSE have enabled automated searches for directed-loop updates~\cite{weber2022quantum}, PMR-QMC also automates the decomposition of an arbitrary Hamiltonian itself. Consequently, PMR-QMC can be used in a genuinely black-box manner~\cite{ezzell2025universal}.} This automated framework has been successfully applied to a broad spectrum of models~\cite{barash2024QuantumMonteCarlo, akaturk2024quantum, babakhani2025quantum}, ranging from standard XY and\vadjust{\vfill\pagebreak} Bose--Hubbard models on rectangular lattices to topological systems and even highly non-local random models. {Of particular relevance, Ref.~\cite{ezzell2025universal} employs a single spin-1/2 code~\cite{ezzell2025code} to study the critical behavior of the square-lattice TFIM, the square-lattice XXZ model, and a random geometrically non-local model, estimating Hamiltonian-based observables such as the energy and fidelity susceptibility. In contrast, SSE implementations--even those incorporating partial automation~\cite{weber2022quantum}--are inherently tailored to specific spin-1/2 models~\cite{bauer2011ALPSProjectRelease,melko2013stochastic,merali2024stochastic} due to the model- and geometry-dependent nature of the Hamiltonian decomposition.}

In this work, we ask a simple question: {Can PMR-QMC also estimate the thermal expectation value of \emph{arbitrary} operators? In principle, we find the answer is \emph{yes}. We show that one can derive unbiased PMR-QMC estimators for arbitrary static operators and subsequently generalize these estimators to dynamic quantities of physical interest. These include estimators for imaginary-time correlation functions and for their integrated susceptibilities, which play key roles in probing spectral features~\cite{hen2012excitation, blume1997excited, blume1998excited} and in diagnosing quantum criticality~\cite{albuquerque2010QuantumCriticalScaling, ezzell2025universal, schwandt2009QuantumMonteCarlo, wang2015FidelitySusceptibilityMade}.} As a demonstration, we apply our approach to estimate a wide variety of non-trivial observables in the transverse-field Ising model on the square lattice {as well as in a toy model with 100 spins}. These include sums of random, non-local Pauli strings and associated dynamic response functions. To our knowledge, no existing QMC framework is capable of evaluating such general observables.

To support broad adoption of our method, we provide open-source code~\cite{ezzell2025code}. This implementation builds on the well-tested spin-1/2 PMR-QMC code~\cite{barash2024PmrQmcCode} developed in Ref.~\cite{barash2024QuantumMonteCarlo} and therefore currently supports arbitrary spin-1/2 systems. Nevertheless, the estimators derived in this work are completely general. The logic underlying them applies without modification to any Hamiltonian written in PMR form, and our implementation can be readily ported into existing or future PMR-QMC codes for other classes of models~\cite{babakhani2025quantum, akaturk2024quantum}, including the recently released higher-spin code~\cite{barash2025highspincode}.

{

\section{Technical overview of estimators we derive}\label{Xsec2-2}\label{Xsec2}
\label{sec:overview-of-estimators}

Given any finite-dimensional Hamiltonian $H$, we derive an unbiased QMC estimator for the thermal expectation value of an arbitrary operator $A$. Our reasoning relies on the structural features of the permutation matrix representation (PMR). Although PMR was originally proposed in Ref.~\cite{gupta2020PermutationMatrixRepresentation}, we provide in this work a more rigorous and broadly applicable formulation and further develop its foundational properties in {\cref{sec:pmr-revisited}}. The essential observation is that any square matrix may be expressed as a ``linear combination'' of permutation matrices. Quantum Hamiltonians can therefore be written as $H = \sum_j D_j P_j$, where each $P_j$ is a matrix representation of an element of a permutation group possessing \emph{regular natural action} (defined formally in {\cref{sec:pmr-revisited}}), and where the associated $D_j$ are diagonal matrices.\footnote{This bears a conceptual resemblance to the Birkhoff--von Neumann decomposition~\cite{dufosse2018further}.}

This decomposition underlies the well-developed PMR-QMC framework, which applies to arbitrary Hamiltonians~\cite{gupta2020PermutationMatrixRepresentation, albash2017OffdiagonalExpansionQuantum, barash2024QuantumMonteCarlo, babakhani2025quantum, gupta2020elucidating, ezzell2025universal} and which we briefly review in {\cref{sec:pmr-qmc}}. Crucially, since any square matrix admits such a decomposition, the same logic allows us to express a general operator as $A = \sum_k \widetilde{D}_k P_k$, regardless of Hermiticity. By decomposing both $H$ and $A$ in the same PMR group basis, we obtain a unified framework for reasoning---in purely group-theoretic terms---about the thermal expectation value $\langle A \rangle = \Tr[A e^{-\beta H}] / \Tr[e^{-\beta H}]$. Combined with the off-diagonal series expansion~\cite{albash2017OffdiagonalExpansionQuantum,hen2018OffdiagonalSeriesExpansion}, this approach allows us to construct unbiased estimators for arbitrary operators $A$, as detailed in {\cref{sec:static-operator-estimators}}.

Before addressing the completely general case, we first derive estimators for several important subclasses of operators for which more efficient formulas exist. These include arbitrary diagonal operators such as $\Hdiag$ ({\cref{subsec:estimating-pure-diag-operators}}), powers of the Hamiltonian $H^k$ ({\cref{subsec:estimating-f(H)-things}}), and the off-diagonal component of the Hamiltonian, $\Hoffdiag$ ({\cref{subsec:estimating-Hoffdiag}}). Throughout, we refer to observables involving $\Hdiag$, $H$, $\Hoffdiag$, or their combinations as \emph{Hamiltonian-based observables}. Once estimators for the core primitive quantities have been obtained, one may combine them to evaluate derived observables such as the specific heat~\footnote{This familiar expression follows from $\langle H \rangle = -\partial_\beta \ln \mathcal{Z}$ together with the variable change $\beta = 1/T$ in natural units.},
\begin{equation}
 \label{eq:specific-heat}
 C_v \equiv \frac{\partial \avg{H}}{\partial T} = \beta^2 (\avg{H^2} - \avg{H}^2),
\end{equation}
using standard jackknife binning techniques~\cite{berg2004introduction, barash2024QuantumMonteCarlo}, which are implemented directly in our code~\cite{ezzell2025code}.

We then extend these ideas to progressively more general classes of operators ({\cref{subsec:estimating-dl-pl}} and \cref{subsec:estimating-al-pl}), culminating in the notion of a \emph{canonical operator}, for which an unbiased estimator exists. For some Hamiltonians, such as those with transverse fields, the canonical form may be written down straightforwardly; for others, it requires more careful handling, and failure to use the canonical form can introduce bias ({\cref{subsec:illustrative-examples}}). We show that any operator can, in principle, be transformed into canonical form ({\cref{subsec:estimating-arbitrary-operators}}). The underlying justification rests on a group-theoretic criterion determining which permutations $P_j$ contribute to $\langle A \rangle$ for a given Hamiltonian. We further demonstrate that the most general estimator may encounter a rare-event sampling issue; however, this difficulty can be mitigated by an improved estimator ({\cref{subsec:improved-estimator}}) that exploits the fact that the PMR group may always be chosen to be Abelian.

Following a short interlude on estimator complexity and total simulation cost ({\cref{subsec:complexity-interlude}}), we generalize our static estimators to dynamic ones ({\cref{sec:dynamic-operator-estimators}}).} Given any operator $A$, dynamic observables are defined in terms of the imaginary-time evolution
\begin{equation}
 \label{eq:imag-time-evolved}
 A(\tau) \equiv e^{\tau H} A e^{-\tau H}.
\end{equation}
An important example is the imaginary-time correlator
\begin{equation}
 \label{eq:AB-imag-time-correlator}
 \avg{A(\tau) B},
\end{equation}
whose estimator we derive in {\cref{subsec:estimating-atau-b}}. We also obtain exact estimators for the associated integrated quantities,
\begin{equation}
 \label{eq:AB-EintFint}
 \int_0^\beta \avg{A(\tau) B} \dtau
\qquad\text{and}\qquad
 \int_0^{\beta/2} \tau \avg{A(\tau) B} \dtau,
\end{equation}
which are discussed in {\cref{subsec:estimating-esus}} and {\cref{subsec:estimating-fsus}}. As in the static case, dynamic estimators for arbitrary operators $A$ and $B$ exist once these operators are expressed in canonical form. Importantly, our formulas do not rely on numerical quadrature: the integrals of $\avg{A(\tau) B}$ can be evaluated analytically, yielding exact closed-form expressions that themselves take the form of PMR-QMC estimators~\cite{ezzell2025universal}.

{
These results strictly generalize and formalize the estimators for the energy susceptibility (ES) and fidelity susceptibility (FS) developed in Ref.~\cite{ezzell2025universal} for studying quantum criticality. More explicitly, for $H(\lambda) = H_0 + \lambda H_1$, the ES and FS may be written~\cite{ezzell2025universal, albuquerque2010QuantumCriticalScaling,schwandt2009QuantumMonteCarlo} as
\begin{equation}
 \label{eq:Esus}
 \chi_E^{H_1} \equiv \int_0^\beta (\avg{H_1(\tau) H_1} - \avg{H_1}^2) \dtau = \int_0^\beta \avg{H_1(\tau) H_1} \dtau - \beta \avg{H_1}^2
\end{equation}
and
\begin{equation}
 \label{eq:Fsus}
 \chi_F^{H_1} \equiv \int_0^{\beta/2} \tau (\avg{H_1(\tau) H_1} - \avg{H_1}^2) \dtau = \int_0^{\beta/2} \tau \avg{H_1(\tau) H_1} \dtau - \frac{\beta^2}{8} \avg{H_1}^2,
\end{equation}
which correspond to the particular choice $A = B = H_1$ considered here. We anticipate that our more general dynamic estimators will find broad application, including in the analysis of spectral properties~\cite{hen2012excitation, blume1997excited, blume1998excited}, which we leave for future investigation.
}

\section{The permutation matrix representation revisited}\label{Xsec3-3}\label{Xsec3}
\label{sec:pmr-revisited}

The permutation matrix representation (PMR) is a special decomposition of square matrices in terms of permutations. As we discuss, the PMR formalism was originally conceived for its usefulness in QMC~\cite{gupta2020PermutationMatrixRepresentation}, but it has since found applications in mitigating the sign problem~\cite{gupta2020elucidating,hen2021determining}, in Dyson series expansions~\cite{kalev2021integral}, quantum algorithms for Hamiltonian simulation~\cite{kalev2021quantum, chen2021quantum}, and recently in generalizing the Feynman path integral to discrete systems~\cite{kalev2024feynman}. Despite this, the PMR decomposition as originally conceived lacks a precise and rigorous general definition. To develop a general theory of measurement within PMR-QMC, we first propose such a precise, rigorous, general PMR form definition.

{
To understand this definition, some understanding of elementary group theory is necessary; this material is contained in numerous introductory textbooks or lecture notes~\cite{gallian2021contemporary,wemyss2013introduction,donaldson2024introduction}. We briefly summarize key facts and introduce notation. Some notation is nonstandard in order to prevent confusion with QMC notation later. A \emph{permutation} is a re-ordering of $d$ elements. For example, the $3$-cycle $(012)$ describes the permutation $0 \rightarrow 1, 1 \rightarrow 2, 2 \rightarrow 0$. The set of all permutations over $d$ elements is the \emph{symmetric group} $S_d$ which has $d!$ elements. A \emph{subgroup} $G < S_d$ is called a \emph{permutation group}. The number of elements in $G$, denoted $|G|$ is called the order of the group. Abstract elements of $G$ are denoted with $\sigma, \pi \in G$. The identity element, denoted $e_G$ or $1$, has trivial cycle form $().$

Permutations have a \emph{natural group action} on $\Omega_d = \{0, \ldots, d-1\}$. That is, given a group element $\sigma \in G$, we can view $\sigma$ as a function such that $\sigma(x) = y$ for $x, y \in \Omega_d$. This group action is called \emph{regular} or \emph{simply transitive} provided the action is fixed-point free and transitive. A group action is fixed-point free if $\sigma\cdot x = x$ if and only if $\sigma = e_G$, i.e., no non-identity element of $G$ fixes a point of $\Omega_d$. The action is transitive provided for all $x, y \in \Omega_d$, there exists a $\sigma \in G$ such that $\sigma(x) = y$. When the action is regular, the $\sigma$ sending $x$ to $y$ is unique. Consequently, permutation groups with regular natural action have order $d$, i.e., $|G| = d$.

For any $d$, there always exists a cyclic group $G < S_d$ with regular natural action, e.g., $\avg{(0\ldots d-1)}_g$. This group is defined as follows. We denote repeated composition of a permutation by raising it to a power, i.e., $\sigma^2(x) = \sigma(\sigma(x)).$ Then the cyclic group $\avg{(0\ldots d-1)}_g = \{(0\ldots d-1), (0\ldots d-1)^2, \ldots, (0\ldots d-1)^d\}$ has $d$ elements because $(0\ldots d-1)^d = ()$. Clearly, this group is Abelian (all elements commute with each other). Later we also use $\sigma^{-1}$ to denote the inverse of a permutation, i.e., $\sigma^{-1}\sigma = \sigma \sigma^{-1} = 1.$

Finally, given a permutation $\sigma \in S_d$, there exists a mapping to a square $d \times d$ permutation matrix $P_\sigma$. This is called a permutation matrix representation (PMR) of $\sigma$, and it is not unique. We take the column convention in this work---that is, if $\sigma(x) = y$, we define $P_\sigma$ such that $P_\sigma \ket{x} = \ket{y}$. By construction, $P_{\sigma^{-1}} = P_\sigma^T.$ Of course, this itself is not uniquely determined up to how we define our basis representations of $\ket{x}, \ket{y}$, but this freedom does not matter for claims in this section.
}

\begin{definition}[Permutation matrix representation (PMR) form]
\label{def:pmr-form}
 We say a square $d \times d$ matrix $A$ is given in PMR form provided,
 {
 \begin{equation}
 A = \sum_{\sigma \in G < S_d} D_\sigma P_\sigma,
 \label{Xeqn7-7}
\end{equation}
 where $D_\sigma$ are $d \times d$ diagonal matrices and $P_\sigma$ are $d \times d$ permutation matrices. The $P_\sigma$ are \emph{permutation matrix representations} of the abstract permutations $\sigma$ that form a subgroup of the symmetric group $G < S_d$. As part of our definition, we require that $G$ has regular natural action. We call the matrix representation of $G$, $\{P_\sigma\}$, a ``PMR-basis''. If $G$ is Abelian, we say $A$ is in Abelian PMR form. If $G = \avg{(0\ldots d-1)}_g$, we say $A$ is in canonical PMR form. Finally, we use the shorthand $P_\sigma \in A$ if $D_\sigma \neq 0$.}

\label{Xenun6-1}
\end{definition}

{
To better understand this decomposition, suppose we wish to extract the $(x,y)$ element of $A$, $A_{x,y}$. Since $G$ has regular natural action, there is a unique $P_{\sigma'}$ such that $P_{\sigma'} \ket{y} = \ket{x}$. Hence, $A_{x,y} = \avg{x | D_{\sigma'} P_{\sigma'} | y} = \avg{x | D_{\sigma'} | x}$, i.e., $A_{x,y} = (D_{\sigma'})_{x,x}.$ Since this was for an arbitrary $(x, y)$, every matrix element of $A$ is given by the diagonal entry of a single $D_\sigma$ matrix.

This observation is also consistent with a simple counting argument. Namely, each $D_\sigma$ has $d$ entries and $|G| = d$ since it has regular natural action. Thus, we can uniquely specify all $d^2$ entries of a generic $A$. For $A$ with $k$ non-zero entries, this also indicates that at most $k$ unique elements of $|G|$ are needed (with all other $D_\sigma = 0$). If $A$ is sparse so that $k \ll d^2$, then $A$ can thus be expressed as a sum of only a few terms.

}

\begin{theorem}[Existence of PMR form]
 \label{thrm:pmr-form-exists}
 Given a square matrix $A$, it is always possible to find an Abelian PMR decomposition, i.e., a PMR decomposition with $G$ Abelian.
\label{Xenun1-1}
\end{theorem}

\begin{proof}
{
 The proof is essentially by construction since $G = \avg{\pi}_g$ for $\pi$ any $d$-cycle has regular natural action. To help concretely illustrate the properties of such groups, we provide an explicit proof.

 Recall that $\Omega_d = \{0,\ldots, d-1\}$. Since $\pi$ is a $d-$cycle with no fixed points, $\pi^k (x)$ is surjective in $\Omega_d$ for $k$ an integer. This shows that for every $x, y$ in $\Omega_d$, there exists a $k$ such that $\pi^k(x) = y$ and hence, $G$ has \emph{transitive} natural action.

 Next, we show $G$ has free natural action. Suppose $1 \neq \sigma \in \avg{\pi}_g$ has a fixed point, $x$. By commutativity, $\pi^k \sigma (x) = \sigma \pi^k (x) = \pi^k (x)$. But since $\pi^k$ is surjective in $\Omega_d$, we find $\sigma(y) = y$ for all $y \in \Omega_d$, a contradiction.

 This proves that $G$ is simply transitive or regular, and hence, for every $x, y \in \Omega_d$, there is a unique $\sigma \in G$ such that $\sigma(x) = y$. To be explicit again, suppose both $\sigma_1 (x) = y$ and $\sigma_2 (x) = y$ for $\sigma_1 \neq \sigma_2$. Since this implies $\sigma_2^{-1} \sigma_1(x) = x$, then $\sigma_2^{-1} \sigma_1$ has a fixed point, and hence must be the identity. But then $\sigma_2 = \sigma_1$, a contradiction.
 }
\end{proof}

This proof is essentially the same as given in Ref.~\cite{gupta2020PermutationMatrixRepresentation} but in the language of group actions. Unlike Ref.~\cite{gupta2020PermutationMatrixRepresentation}, however, our definition of the PMR form requires specifying a group $G$ with the desired properties explicitly rather than implicitly. This insistence on specifying the PMR basis as part of the definition greatly clarifies the PMR formalism and prevents subtle technical pitfalls. Most importantly, our definition guarantees the following property, essential for all current applications of PMR~\cite{gupta2020PermutationMatrixRepresentation, gupta2020elucidating, hen2021determining, akaturk2024quantum, barash2024QuantumMonteCarlo, chen2021quantum, kalev2021integral, kalev2021quantum, kalev2024feynman}.

\begin{corollary}[Products of PMR permutations have no fixed points]
\label{corr:Siq-must-be-identity}
 If a product of permutations from a PMR basis satisfies $ P_{\sigma_q}P_{\sigma_{q-1}} \ldots P_{\sigma_1} \ket{x} = \ket{x}$ for even a single $\ket{x}$, then this product simplifies to the identity, $\mathds{1}$.
\label{Xenun4-1}
\end{corollary}

\begin{proof}
 By construction all $1 \neq \pi \in G$ have no fixed points, so the result follows by group closure.
\end{proof}

Within our suggested formalism, this proof is trivial, but this is a nontrivial convenience of our definition. In earlier work~\cite{gupta2020PermutationMatrixRepresentation}, a PMR form was defined without specifying a group $G$ or its properties, which is not enough to imply {\cref{corr:Siq-must-be-identity}} as previously implicitly assumed. To understand this claim, let $P_1$ be the PMR of the cycle $(01)(234)$. One can verify that $P_1^{-1}$ is equivalent to $(01)(243)$ also has no fixed points, and hence, $H = D_0 + D_1 P_1 + D_2 P_1^{-1}$ is a legitimate PMR Hamiltonian according to prior PMR definitions~\cite{gupta2020PermutationMatrixRepresentation}. On the other hand, because $(243) \in \avg{(01)(234)}_g$, products of permutations in this $H$ can have fixed points, which violates the key property of {\cref{corr:Siq-must-be-identity}}. In summary, an important difference in our new definition is that the no fixed points property is imposed on a group $G$ and not on non-trivial permutations that could comprise $H$.

Given a choice of a PMR group basis $G$, the diagonals are then uniquely determined. We show this by construction with the following novel property.

\begin{theorem}[Computing $D_\sigma's$ generally]
\label{thrm:computing-Dj}
Let $\diag{X} = (X_{11}, \ldots, X_{dd}\}$ be the vector of diagonal entries of a matrix.
 \begin{equation}
 \diag{D_\sigma} = \diag{A P_\sigma^T} = \text{diag}(P_\sigma A^T)
 \label{Xeqn8-8}
\end{equation}
\end{theorem}

\begin{proof}
 We show the first equality by direct computation,
 \begin{align*}
 (A P_\sigma^T)_{kk} &= \sum_{\pi, l} ( D_\pi P_\pi)_{k, l} (P_\sigma^T)_{l, k}
 = \sum_{\pi} (D_\pi)_{k, k} \sum_l (P_\pi)_{k,l} (P_\sigma^T)_{l, k} \\
 &= \sum_\pi (D_\pi)_{k, k} (P_\pi P_\sigma^T)_{k, k} \\
 &= (D_\sigma)_{k,k}.
 \end{align*}
The final line follows for two reasons. First, $P_\sigma$ has a unique inverse $P_\sigma^T$, and the resulting identity matrix has a 1 for every $k$. Second, products of permutations have no fixed points unless they are identity by {\cref{corr:Siq-must-be-identity}}. The second equality in the theorem then follows since a square matrix and its transpose have the same diagonal elements.
\end{proof}
This shows that $\{P_\sigma \}$ is akin to a generalized basis for which the weights, $\{D_\sigma \}$ are given by the above ``inner-product'' like operation---hence why we refer to it as a PMR basis. Each diagonal can be computed in $O(d)$ time and space using this formula since we can efficiently store and manipulate permutations as permuted vectors of length $d$. For many systems and PMR bases of interest in practice, though, the $D_j$'s can often be determined immediately by inspection~\cite{gupta2020PermutationMatrixRepresentation} or computed efficiently using sparse representations of Pauli operators and modular linear algebra~\cite{barash2024QuantumMonteCarlo, babakhani2025quantum}.

So far, we have only assumed $A$ is a square matrix, but within PMR-QMC, we are interested primarily in Hermitian matrices. As suggested by Ref.~\cite{gupta2020PermutationMatrixRepresentation}, Hermiticity implies a constraint on the PMR form of $A$.
\begin{theorem}[Hermiticity in the PMR]
\label{thrm:hermiticity}
 $A$ is Hermitian $(A = A^{\dagger})$ if and only if $D_\sigma^* = P_\sigma D_{\sigma^{-1}} P_{\sigma}^T$ for every $\sigma \in G$.
\label{Xenun3-3}
\end{theorem}

\begin{proof}
 ($\Leftarrow$) Suppose $D^*_\sigma = P_\sigma D_{\sigma^{-1}} P_\sigma^T, \ \forall \sigma \in G$. Then $A^{\dagger} = \sum_\sigma P_\sigma^T D_\sigma^* = \sum_\pi D_\pi P_\pi = A$ for $\pi = \sigma^{-1}.$ ($\Rightarrow$) Suppose $A^{\dagger} = A$. By the PMR construction, $\avg{k | A | l} = \avg{k | D_\sigma P_\sigma | l} = (D_\sigma)_{k,k}$ for the unique $\sigma$ such that $P_\sigma \ket{l} = \ket{k}.$ Similarly, there is a unique $P_\pi $ such that $\avg{k | A^{\dagger} | l} = \avg{k | P_\pi^T D_\pi^* | l} = (D_\pi^*)_{l,l}$. Since $P_\pi \ket{k} = \ket{l}$, then by uniqueness, $\pi = \sigma^{-1}.$ Combined with the assumption $A = A^{\dagger}$, we find $(D_\sigma^*)_{k,k} = (D_{\sigma^{-1}})_{l,l}$ or more conveniently, $(D_\sigma^*)_{k,k} = (P_\sigma D_{\sigma^{-1}} P_\sigma^T)_{k,k}$ and hence $D_\sigma^* = P_\sigma D_{\sigma^-1} P_\sigma^T$.
\end{proof}
This claim is actually different from the one made in Ref.~\cite{gupta2020PermutationMatrixRepresentation}, which incorrectly suggested $D_\sigma^* = D_{\sigma^{-1}}$ instead of our now corrected claim $D_\sigma^* = P_\sigma D_{\sigma^{-1}} P_\sigma^T$. We refute the old claim and verify our claim in an explicit example in {\cref{subsec:pmr-examples}}.
\begin{corollary}[Alternate way to compute $D_\sigma's$]
\label{corr:computing-Dj}
 When $A$ is Hermitian so that $A^T = A^*$, $\diag{D_\sigma} = \diag{P_\sigma A^*}$.
\label{Xenun5-2}
\end{corollary}

We remark that none of the properties we proved above rely on $G$ being Abelian except for our general existence proof~\footnote{Note that $S_3$ with 6 elements is the smallest non-Abelian group, so general existence for $d < 6$ actually does require $G$ to be Abelian.}. Hence, for applications like PMR-QMC that rely on {\cref{corr:Siq-must-be-identity}}, one may choose any $G$ that has regular group action, Abelian or not. For spin models~\cite{barash2024QuantumMonteCarlo, babakhani2025quantum} and bosonic models~\cite{akaturk2024quantum}, however, practical PMR bases are also Abelian. Using a Jordan--Wigner transformation, one can also study Fermions with an Abelian PMR basis, but a non-Abelian PMR form seems more effective~\cite{babakhani2025fermionic}.

\subsection{Uniqueness of the PMR and the local, canonical PMR basis}\label{Xsec4-3.1}\label{Xsec4}
{ Given a suitable PMR basis $G$, we can write any square matrix $A$ in the form}
\begin{equation}
 \label{eq:A-pmr-spin-1/2}
 A = \sum_{P \in G} D_P P , \ \ \ \diag{D_P} = \diag{A P^T},
\end{equation}
where we have introduced several convenient abuses of notation---namely, viewing $G$ and its representation interchangeably in writing $P \in G$ and using $P$ both as a permutation and a subscript in $D_P$. We shall find this notation useful throughout, particularly in {\cref{subsec:estimating-arbitrary-operators}}. Given $G$, this decomposition is unique (i.e., the only valid assignment of $D_P$'s comes from {\cref{thrm:computing-Dj}}), but the choice of $G$ itself is not unique. { In this section, we discuss the non-uniqueness of the PMR basis in more detail.}

{ Naively, any permutation group with regular group action is a valid PMR group basis, but this greatly over counts. To make this point, let $S_4$ denote the symmetric group over $\Omega_4 = \{0,1,2,3\}$.} By {\cref{thrm:pmr-form-exists}}, { we know $\avg{(0123)}_g$ and $\avg{(0213)}_g$ are valid PMR bases, but up to the freedom to order a basis $\{\ket{0}, \ket{1}, \ket{2}, \ket{3}\}$, these are clearly the same. This concept is formalized by the so-called \emph{conjugacy classes} or \emph{cycle type}. For $S_4$, the only subgroups with regular natural action up to cycle type are the cyclic group $C_4 = \avg{(0123)}_g$ and the Klein group generated by disjoint transpositions, $V_4 = \{(02)(13), (01)(23)\}.$ The structure of the Klein group actually becomes more familiar via its matrix representation, $\{\mathds{1}, X_1, X_2, X_1 X_2\} \cong K_4$, which motivates the following definition.}

\begin{definition}[Local, canonical PMR form]
 \label{def:local-pmr-form}
 Given $n$ particles with local dimensions $d_i$, we say $G$ is the local, canonical (PMR) basis if every $P \in G$ is of the form,
 \begin{equation}
 P = \bigotimes_{i=1}^n P_{i}^{k_i},
 \label{Xeqn10-10}
\end{equation}
 for $P_i$ the permutation representation of { the $d_i$-cycle $(01\ldots d_i -1 )$ } and $k_i \in \{0, \ldots, d_i-1\}$. Any $A$ written in such a basis is said to be written in the local, canonical PMR form. { By definition, there is only one local, canonical PMR basis for any physical system of $n$ particles.}
\label{Xenun7-2}
\end{definition}
 Implicit in this definition is that such constructions are legitimate PMR forms. This follows simply by the mixed-product property of the Kronecker product, i.e., $(D_A P_A) \otimes (D_B P_B) = (D_A \otimes D_B) (P_A \otimes P_B).$ In other words, given two local PMR forms for operators $A$ and $B$, we immediately have the local PMR form for $A \otimes B$ by this property. By successive applications, we can thus build up $n$-particle PMR forms from local forms. Whenever each local form is also canonical (generated by a $d_i$-cycle), then the total is a local, canonical basis.

{ For $S_4$, we saw that the global basis $C_4$ and the local basis $K_4$ are the only possibilities. For $S_6$, however, we can view $S_3$ as a proper subgroup whose group action on itself embedded into $\Omega_6$ is regular---an example of a non-Abelian PMR group basis for $d = 6.$ This example reveals that non-Abelian PMR bases exist, and furthermore, that a complete characterization of PMR bases is akin to a certain classification of all finite groups. Yet in practice, all PMR-QMC calculations thus far have been carried out with the local, canonical basis defined above~\cite{albash2017OffdiagonalExpansionQuantum,gupta2020PermutationMatrixRepresentation,barash2024QuantumMonteCarlo, babakhani2025quantum,akaturk2024quantum}, so we consider further characterization of uniqueness a question in group theory outside the scope of this work.}

{
 \subsection{Illustrative finite dimensional examples}
    \label{subsec:pmr-examples}
Before moving on to infinite dimensional generalizations, we provide a few examples of the PMR in action for simple spin systems. We also discuss the choice of PMR bases most relevant to the study of physical models. }

\subsubsection{A single spin-1/2 particle (one qubit)}\label{Xsec5-3.2.1}\label{Xsec5}

Consider an arbitrary single qubit operator decomposed in terms of the spin-1/2 Pauli operators $\{\mathds{1}, X, Y, Z\}$,
\begin{equation}
 A = a_0 \mathds{1} + a_1 X + a_2 Y + a_3 Z = \sum_{i=0}^3 a_i \sigma^{(i)}.
\label{Xeqn11-11}
\end{equation}
The only $2$-cycle is $(01)$, so $\avg{(), (01)}$ is the unique PMR basis with corresponding matrix representation $\{P_0 = \mathds{1}, P_1 = X\}$ in the standard computational basis. By direct computation of $\diag{A P^T}$ or by inspection and the relation $Y = -i Z X$, we find
\begin{equation}
 A = D_0 P_0 + D_1 P_1, \ \ D_0 = a_0 \mathds{1} + a_3 Z, \ \ D_1 = a_1 \mathds{1} - i a_2 Z.
\label{Xeqn12-12}
\end{equation}
Since we have not required $A$ is Hermitian, then the $a_i's$ can be complex in general. Yet, if $A$ is Hermitian, then each $a_i$ must be real. In this case, we find $D_0^* = D_0$ and $D_1^* = X D_1 X$
since $ZXZ = -Z$ as expected from {\cref{thrm:hermiticity}}. In addition, we find $D_1^* \neq D_1$ which debunks the claim $D_P^* = D_{P^T}$ in Ref.~\cite{gupta2020PermutationMatrixRepresentation}.

We remark that despite the definitional issues of the PMR form as originally conceived~\cite{gupta2020PermutationMatrixRepresentation}, all present examples of PMR in numerical studies~\cite{gupta2020PermutationMatrixRepresentation, barash2024QuantumMonteCarlo, akaturk2024quantum, babakhani2025quantum} have utilized a local, canonical PMR basis for which {\cref{corr:Siq-must-be-identity}} holds by our general proof. In other words, their theoretical claims and empirical results are not challenged by our work.

\subsubsection{Two spin-1/2 particles (two qubits)}\label{Xsec6-3.2.2}\label{Xsec6}

For two qubits, the canonical local basis is $G^{(2)}_X = \{\mathds{1}, X_1, X_2, X_1X_2\}$ (with cycle representation given by $\{(), (02)(13), (01)(23), (03)(12) \}$, but this is no longer a unique choice. For example, $\avg{(1234)}_g$ is a valid, distinct alternative PMR basis. As an explicit demonstration of this, let $\widetilde{P}$ be the matrix representation of $(1234)$,
\begin{equation}
 \widetilde{P} =
 \begin{bmatrix}
 0 & 0 & 0 & 1 \\
 1 & 0 & 0 & 0 \\
 0 & 1 & 0 & 0 \\
 0 & 0 & 1 & 0 \\
 \end{bmatrix}.
\label{Xeqn13-13}
\end{equation}
By inspection, $\wtP \notin G^{(2)}_X$, and by closure of the group, $\wtP$ cannot be expressed as a product of Pauli X strings. As with any square matrix, though, we can express it in this local canonical PMR basis,
\begin{equation}
 \tilde{P} = D_2 X_2 + D_3 X_1 X_2, \ \ \diag{D_2} = (0,1,0,1), \ \ \diag{D_3} = (1,0,1,0).
\label{Xeqn14-14}
\end{equation}
Since $\wtP$ is not Hermitian, we observe an unsurprising violation of the Hermiticity condition, $D_2^* \neq X_2 D_2 X_2 = D_3$.

\subsubsection{Any number of spin-1/2 particles ($n$ qubits)}\label{Xsec7-3.2.3}\label{Xsec7}
\label{subsec:spin1/2-pmr-example}
Now consider an $n$ qubit operator which can always be decomposed into the basis of $n$-qubit Pauli strings, $\{I, X, Y, Z\}^{\otimes n}$, which form an orthogonal basis with respect to the Hilbert-Schmidt inner product. That is, an arbitrary $2^n \times 2^n$ square matrix $A$ can be written,
\begin{equation}
 \label{eq:A-spin-1/2}
 A = \sum_{Q \in \{I, X, Y, Z\}^{\otimes n}} a_Q Q, \ \ \ a_Q = \frac{1}{2^n}\Tr[ Q A],
\end{equation}
{ where we use $Q$ as a Pauli string to avoid clashing with permutation matrix notation throughout.} In this representation, Hermiticity is ensured provided each $a_Q$ is real. { Though this is not the only possible choice of a valid matrix basis, this is by far the most common since the Pauli strings have nice properties such as locality (i.e., Hamiltonians for the transverse-field Ising model, XY model, Heisenberg model, and so on are written in this basis).}

{ In analogy, the corresponding local canonical PMR basis is the Abelian group formed by all Pauli-X strings, $G^{(n)}_X = \{X^{b_1} X^{b_2}\cdot \ldots \cdot X^{b_n} : (b_1,\ldots,b_n) \in \{0,1\}^n\}$. Though this is not unique in the sense that $\avg{(01\ldots d-1)}$ is a valid, distinct basis, for example, it is the unique local basis. Furthermore, given a matrix as a sum of Pauli strings} (i.e., {\eqref{eq:A-spin-1/2}}), { there is an efficient, mod-2 arithmetic linear algebra method}~\cite{barash2024QuantumMonteCarlo} to compute the corresponding diagonal matrices $D_P$ for each $P \in G^{(n)}_X$ that is faster than using the general formula $\diag{A P^T}$.

 \subsubsection{A local, canonical PMR basis for $n$ $d$-level systems}

{ For a qubit, we have seen that $(01)$ or $X$ generates the unique PMR basis $\{\mathds{1}, X\}.$ For $n$-qubits, the unique local, canonical basis is the $n$-fold tensor product of such operators, i.e., the set of all Pauli-X strings. In analogy, for a $d$-level system with basis states $\ket{0}, \ket{1}, \ldots, \ket{d-1}$, $(01\ldots d-1)$ forms the unique canonical PMR basis $\avg{(01\ldots d-1)}_g$ up to trivial permutations of the basis elements. Yet, not all PMR bases must be formed by a $d-$cycle, i.e., for $d = 4$ above, we saw $G^{(2)}_X = \{\mathds{1}, X_1, X_2, X_1X_2 \}$ also works. As a technical remark, though $G^{(2)}_X$ is a local canonical basis for two qubits, it is neither local nor canonical for a general 4-level system.}

{ Let $X_d$ denote the matrix representation of $(01\ldots d-1)$,
which is known as the \emph{shift matrix}. The shift matrix is one of the finite generators of the Heisenberg-Weyl group as discussed in Sec.~3.7 of Ref.} \cite{wilde2013quantum} or used in \citet{tripathi2024qudit}). { With this defined, a local, canonical PMR basis for $n$ such $d-$level systems is given by permutations of the form $X_d^{k_1} \otimes \ldots \otimes X_d^{k_n}$ for $k_i \in \{0, \ldots, d-1\}$. In the physical case of spin $s$ particles of local dimension $d = 2s + 1$ in particular, the translation of a given operator in a standard spin basis into a local, canonical PMR form has been described in Ref.}~\cite{babakhani2025quantum}. As with the spin-1/2 case~\cite{barash2024QuantumMonteCarlo}, { there is an efficient mod-$d$ linear algebra algorithm to construct the $D_P$'s}~\cite{babakhani2025quantum}.

\subsection{Extensions to bosons and fermions}\label{Xsec8-3.3}\label{Xsec8}

\subsubsection{The Bose--Hubbard model: an extension to countably infinite dimensions}\label{Xsec9-3.3.1}\label{Xsec9}
\label{subsec:bose-hubbard}
Our rigorous discussion of the PMR form is based on the idea of representing a square $d \times d$ matrix. Yet, the essential details can readily be generalized to some infinite dimensional systems, as explored for the Bose-Hubbard model in Ref.~\cite{gupta2020PermutationMatrixRepresentation} and applied successfully in practical PMR-QMC simulations in Ref.~\cite{akaturk2024quantum}. In second quantization, the Bose-Hubbard Hamiltonian on $L$ lattice sites can be written,
\begin{equation}
 H = -t \sum_{\avg{i,j}} \cb{i} \db{j} + \frac{U}{2} \sum_{i=1}^L \no{i} (\no{i} - 1) - \mu \sum_{i=1}^L \no{i},
\label{Xeqn16-16}
\end{equation}
for $\avg{i,j}$ a summation over neighboring lattice sites. The basis for which one can define permutations over is most conveniently the second quantized occupation number basis, $\ket{\bf{n}} = \ket{n_1,n_2,\ldots,n_L}$ for each $n_k$ a non-negative integer denoting the number of bosons on each lattice site.

Within this basis, the corresponding operator terms in $H$ can be explained. Firstly, $\cb{i}$ and $\db{i}$ are creation and annihilation operators, respectively, and they satisfy
\begin{align}
 \cb{i} \ket{\bf{n}} &= \sqrt{(n_i + 1)} \ket{n_1, \ldots, n_{i-1}, n_{i} + 1, n_{i+1}, \ldots, n_L} \\
 \db{i} \ket{\bf{n}} &= \sqrt{n_i} \ket{n_1, \ldots, n_{i-1}, n_i - 1, n_{i+1}, \ldots, n_L}
\end{align}
and the commutation relations
\begin{equation}
 [\cb{i}, \cb{j}] = [\db{i}, \db{j}] = 0, \ \ \ [\cb{i}, \db{j}] = \delta_{i,j}.
\label{Xeqn17-19}
\end{equation}
The operator $\no{i} = \cb{i} \db{i}$ is called the number operator since it satisfies
\begin{equation}
 \no{i} \ket{\bf{n}} = n_i \ket{\bf{n}}.
\label{Xeqn18-20}
\end{equation}
Evidently, terms involving only the number operator are diagonal in this basis, and correspondingly,
\begin{equation}
 D_0 = \frac{U}{2} \sum_{i=1}^L \no{i} (\no{i} - 1) - \mu \sum_{i=1}^L \no{i}.
\label{Xeqn19-21}
\end{equation}
The remaining $-t \ \cb{i} \db{j}$ terms have action,
\begin{equation}
 -t \ \cb{i} \db{j}\ket{\bf{n}} = -t \sqrt{(n_i + 1) n_j} \ket{\mathbf{n}^{(i,j)}} \propto \ket{\mathbf{n}^{(i,j)}},
\label{Xeqn20-22}
\end{equation}
{ with $\ket{\bf{n}^{(i,j)}}$ denoting the state $\ket{\bf{n}}$ with one additional boson at site $i$ and one fewer at site $j$} and the proportionality assumes $n_j > 0$ (we will return to this subtlety soon). Defining permutations $P_{i,j}$ and associated diagonal operators,
\begin{align}
 \label{eq:bose-hubbard-pij}
 P_{i,j} \ket{\bf n} &= \ket{\mathbf{n}^{(i,j)}} \\
 \label{eq:bose-hubbard-dij}
 D_{i,j} &= -t \sum_{\bf n} \sqrt{n_i (n_j + 1)} \ketbra{\bf n},
\end{align}
we can write
\begin{equation}
 H = D_0 + \sum_{\avg{i,j}} D_{i, j}P_{i,j},
\label{Xeqn21-25}
\end{equation}
which is essentially a PMR form of $H$.

Within the standard number basis, however, these permutations do not commute, so this might appear to be a non-Abelian PMR form. To see this, we observe $P_{2,1} P_{1,2} \ket{0,1} = \ket{0,1}$ yet $P_{2,1} \ket{0,1} = 0$ since $n_0 = 0$ and one cannot remove a non-existent boson from the first lattice site. Thus, $[P_{2,1}, P_{1,2}] \neq 0$ in the standard number basis. Yet, there is a simple way to view this as an Abelian permutation group by introducing ``artificial basis elements'' in which lattice sites are allowed to have negative bosons. In doing so, $P_{2,1} \ket{0,1} = \ket{-1, 2}$, so $P_{1,2} P_{2,1} \ket{0,1} = \ket{0,1}$ and permutations commute. The set of all possible permutations, still defined via {\eqref{eq:bose-hubbard-pij}} on this extended basis, is now a finite Abelian group with the desired properties such as no fixed points.

Having extended the permutation definition to this artificial basis, we also extend the corresponding diagonals,
\begin{equation}
 D_{i,j} \longrightarrow -t \sum_{\bf n} \sqrt{n_i (n_j + 1)} \ketbra{\bf n}{\bf n} + \sum_{- \mathbf{n}} 0 \ketbra{- \mathbf{n}}{- \mathbf{n}},
\label{Xeqn22-26}
\end{equation}
where the second sum is over all artificial basis elements with negative bosons on at least one lattice site. This definition leads to consistent evaluation of matrix elements in PMR-QMC where we evaluate matrix elements of the form (see {\cref{sec:pmr-qmc}}),
\begin{equation}
 \left \langle \mathbf{n} \left| \prod_{l=1}^q D_{i_l,j_l} P_{i_l,j_l} \right| \mathbf{n} \right \rangle.
\label{Xeqn23-27}
\end{equation}
Since basis elements are always evaluated with $D_{i,j}P_{i,j}$ pairs, then the extension to the negative boson basis is consistent with the standard basis evaluation. Namely, if this product is $0$ in the standard basis without an extended definition of the permutations and diagonal, then it will also be $0$ in the extended basis with the expanded permutations and diagonals.

\subsubsection{Fermionic systems: usage of non-Abelian, extended PMR form}\label{Xsec10-3.3.2}\label{Xsec10}
\label{subsec:fermionic-example}
A fermionic Hamiltonian written in second quantization can readily be transformed into a spin-1/2 system via the Jordan--Wigner transformation~\cite{nielsen2005fermionic}. In doing so, one can re-cast a fermionic model into a spin-1/2 model, for which the Abelian PMR form is well understood~\cite{barash2024QuantumMonteCarlo}, as we explained in {\cref{subsec:spin1/2-pmr-example}}. This transformation is also discussed in more detail in Ref.~\cite{babakhani2025quantum}, but in practice, one can design a more efficient PMR-QMC scheme directly within second quantization~\cite{babakhani2025fermionic}. In this scheme, one can define the permutations in terms of creation and annihilation operators, similar to the discussion of Bose-Hubbard in {\cref{subsec:bose-hubbard}}. Similar to the Bose-Hubbard model, some permutations actually annihilate a state entirely, returning $0$. Unlike the Bose-Hubbard model, however, it does not appear possible to embed the permutations into an Abelian group on an extended basis in which states are not annihilated by permutations.

Naively, we might say that it is simply a non-Abelian PMR form, which therefore satisfies all necessary practical properties like {\cref{corr:Siq-must-be-identity}}. However, the possibility to annihilate a state means that the permutations are really more than a permutation in this special boundary case, so one might call this ``an extended, fermionic PMR form.'' For concrete details, we refer readers to Ref.~\cite{babakhani2025fermionic}. Nevertheless, this subtle technicality does not affect PMR-QMC estimation much. In practice, one simply keeps track of a variable, $s \in \{0, -1, 1\}$, that depends on the order in which permutations are applied to a given basis state. This presents no issue to our derivations, as most of our estimators are dependent on the order permutations appear as well. As such, all the estimators we derive in this work---with the exception of {\eqref{eq:improved-estimator-for-canonical-ops}} which uses commutativity---carry over to this unusual fermionic PMR form with the small addition of the $s$ variable.

\section{Review of divided differences}\label{Xsec11-4}\label{Xsec11}
\label{sec:div-diff}

We review the technical details of the divided difference, an essential object in the derivation of the PMR-QMC partition function and our estimators. Our summary is heavily inspired by a similar review in Ref.~\cite{ezzell2025universal} and earlier observations made in Refs.~\cite{albash2017OffdiagonalExpansionQuantum, hen2018OffdiagonalSeriesExpansion}. The divided difference of any holomorphic function $f(x)$ can be defined over the multiset $[x_0, \ldots, x_q]$ using a contour integral~\cite{mccurdy1984accurate, deboor2005DividedDifferences},
\begin{equation}
 \label{eq:contour-int-dd}
 f[x_0, \ldots, x_q] \equiv \frac{1}{2\pi i} \oint\limits_{\Gamma} \frac{f(x)}{\prod_{i=0}^q(x - x_i)} \mathrm{d}x,
\end{equation}
for $\Gamma$ a positively oriented contour enclosing all the $x_i$'s. Several elementary properties utilized in PMR-QMC~\cite{albash2017OffdiagonalExpansionQuantum, gupta2020PermutationMatrixRepresentation} follow directly from this integral representation, or by invoking Cauchy's residue theorem. For example, $f[x_0, \ldots, x_q]$ is invariant to permutations of arguments, the definition reduces to Taylor expansion weights when arguments are repeated,
\begin{equation}
 f[x_0, \ldots, x_q] = f^{(q)}(x) / q!, \ \ \ x_0 = x_1 \ldots = x_q = x,
\label{Xeqn25-29}
\end{equation}
and whenever each $x_i$ is distinct, we find
\begin{equation}
 f[x_0, \ldots, x_q] = \sum_{i=0}^q \frac{f(x_i)}{ \prod_{k \neq i} (x_i - x_k)},
\label{Xeqn26-30}
\end{equation}
the starting definition in Refs.~\cite{albash2017OffdiagonalExpansionQuantum, hen2018OffdiagonalSeriesExpansion, gupta2020PermutationMatrixRepresentation}. Each of these divided difference definitions can be shown to satisfy the Leibniz rule~\cite{deboor2005DividedDifferences},
\begin{align}
 \label{eq:leibniz-rule}
 (f \cdot g)[x_0, \ldots, x_q] &= \sum_{j=0}^q f[x_0, \ldots, x_j] g[x_j, \ldots, x_q] \nonumber\\  &= \sum_{j=0}^q g[x_0, \ldots, x_j] f[x_j, \ldots, x_q]
\end{align}
which is particularly important in our derivations.

Yet another useful way to view the divided difference for our work is to derive its power series expansion,
\begin{align}
 \label{eq:dd-as-series-expansion}
 f[x_0, \ldots, x_q] = \sum_{m=0}^\infty \frac{f^{(q+m)}(0)}{(q+m)!} \sum_{\sum k_j = m} \prod_{j=0}^q x_j^{k_j},
\end{align}
where the notation $\sum_{\sum k_j = m}$ is a shorthand introduced in Refs.~\cite{hen2018OffdiagonalSeriesExpansion, albash2017OffdiagonalExpansionQuantum} which represents a sum over all \emph{weak integer partitions} of $m$ into $q+1$ parts. More explicitly, it is an enumeration over all vectors $\bm{k}$ in the set $\{\bm{k} = (k_0, \ldots, k_q) : k_i \in \mathds{N}_0, \sum_{j=0}^q k_j = m\}$, where $\mathds{N}_0 = \mathds{N} \cup \{0\}$ is the natural numbers including zero. To derive the series expansion, we first Taylor expand $f(x)$ inside the contour integral,
\begin{align}
 f[x_0, \ldots, x_q] &=\sum_{n=0}^\infty \frac{f^{(n)}(0)}{n!} \left( \frac{1}{2 \pi i } \oint_{\Gamma} \frac{x^n}{\prod_{i=0}^q (x - x_i)} \right)
 \nonumber\\& = \sum_{n=0}^\infty \frac{f^{(n)}(0)}{n!} [x_0, \ldots, x_q]^n ,
\end{align}
where we have introduced the shorthand $[x_0, \ldots, x_q]^k \equiv p_k[x_0, \ldots, x_q]$ for $p_k(x) = x^k$ as introduced in Refs.~\cite{hen2018OffdiagonalSeriesExpansion, albash2017OffdiagonalExpansionQuantum}. The divided difference of a polynomial has a closed form expression most easily written with the change of variables $n \rightarrow q + m$,
\begin{equation}
 \label{eq:poly-div-diff}
 [x_0, \ldots, x_q]^{q+m} =
 \begin{cases}
 0 & m < 0 \\
 1 & m = 0 \\
 \sum_{\sum_{k_j} = m} \prod_{j=0}^q x_j^{k_j} & m > 0,
 \end{cases}
\end{equation}
as noted in the same notation in Refs.~\cite{hen2018OffdiagonalSeriesExpansion, albash2017OffdiagonalExpansionQuantum} and derived with a different notation in Ref.~\cite{deboor2005DividedDifferences}. Together, these two observations yield {\eqref{eq:dd-as-series-expansion}}.

As with previous PMR works~\cite{albash2017OffdiagonalExpansionQuantum, gupta2020PermutationMatrixRepresentation}, we are also especially interested in the divided difference of the exponential (DDE) where we utilize the shorthand notation $e^{t [x_0, \ldots, x_q]} \equiv f[x_0, \ldots, x_q]$ for $f(x) = e^{t x}$. Replacing $f^{(q+m)}(0)$ with $t^{q+m}$ gives the power series expansion of the DDE, as derived in the appendices of Refs.~\cite{albash2017OffdiagonalExpansionQuantum, hen2018OffdiagonalSeriesExpansion}. Replacing the variable $x \rightarrow \alpha x$ in {\eqref{eq:contour-int-dd}}, we find the rescaling relation,
\begin{equation}
 \label{eq:rescaling-relation}
 \alpha^q e^{t [\alpha x_0, \ldots, \alpha x_q]} = e^{\alpha t [x_0, \ldots, x_q]},
\end{equation}
which is useful in numerical schemes~\cite{gupta2020CalculatingDividedDifferences, barash2022calculating} and in computing the Laplace transform. In particular, combining {\eqref{eq:rescaling-relation}} with {(\protect\cref{Xeqn11-11})} of Ref.~\cite{kunz1965inverse} with $P_m(x) = 1$, we find
\begin{equation}
 \label{eq:laplace-of-dd}
 \lap{e^{\alpha t [x_0, \ldots, x_q]}} = \frac{\alpha^q}{ \prod_{j=0}^q (s - \alpha x_j)},
\end{equation}
where $\mathcal{L}$ denotes the Laplace transform from $t \rightarrow s$. One can alternatively derive {\eqref{eq:laplace-of-dd}} by directly performing the integration to the contour integral definition in {\eqref{eq:contour-int-dd}}, e.g. by Taylor expanding $e^{\alpha t x}$ and re-summing term-by-term, legitimate by the uniform convergence of the exponential. As shown in Refs.~\cite{zeng2025inequalities, ezzell2025universal}, the Laplace transform is a powerful tool for deriving integral relations of the DDE, and we will make a similar use of it in this work.

\section{Permutation matrix representation quantum Monte Carlo (PMR-QMC)}\label{Xsec12-5}\label{Xsec12}
\label{sec:pmr-qmc}

The permutation matrix representation quantum Monte Carlo (PMR-QMC) algorithm, recently introduced in Ref.~\cite{gupta2020PermutationMatrixRepresentation}, is a universal parameter-free Trotter error-free quantum Monte Carlo algorithm for simulating general quantum and classical many-body models within a single unifying framework. The algorithm builds on a power series expansion of the quantum partition function in its off-diagonal terms~\cite{albash2017OffdiagonalExpansionQuantum, hen2018OffdiagonalSeriesExpansion} in a way that the quantum `imaginary-time' dimension consists of products of elements of a permutation group, allowing for the study of essentially arbitrarily defined systems on the same footing~\cite{albash2017OffdiagonalExpansionQuantum,gupta2020elucidating,halverson2020efficient}. Notably, Ref.~\cite{barash2024PmrQmcCode} developed an automated, deterministic algorithm to generate PMR-QMC update rules that satisfy detailed balance and are ergodic for arbitrary spin-1/2 Hamiltonians, and Ref.~\cite{akaturk2024quantum} did the same for Bose-Hubbard models on arbitrary graphs. This has since been generalized to higher spin systems and hence also to arbitrary bosonic and fermionic systems and mixtures thereof~\cite{babakhani2025quantum}.

\subsection{The off-diagonal series expansion}\label{Xsec13-5.1}\label{Xsec13}
\label{subsec:ODE}
We derive the `off-diagonal series expansion' of $\Tr[f(H)]$, following closely a similar derivation for $\Tr[e^{-\beta H}]$ given in Refs.~\cite{hen2018OffdiagonalSeriesExpansion, albash2017OffdiagonalExpansionQuantum, gupta2020PermutationMatrixRepresentation}. In our derivation, { we assume $H$ is written in PMR form, $\sum_{\sigma \in G} D_\sigma P_\sigma$ for $G$ a valid PMR basis. However, to make a more direct connection to notation used in all prior PMR works, we specifically write $H = D_0 + \sum_j D_j P_j$, which can be viewed as making the assignment $P_k \equiv P_{\sigma_k}$ for $\{\sigma_k\}_{k=0}^{|G|-1}$ a chosen fixed ordering of the group elements with $\sigma_0 = 1$. We further denote $\{\ket{z}\}$ an orthonormal basis in which $D_0$ is diagonal (the ``computational basis'') and assume $f$ is analytic.}

In this case, a direct Taylor expansion of $f(H)$ about $0$ yields
\begin{align}
 \Tr[f(H)] &= \sum_z \sum_{n=0}^\infty \frac{f^{(n)}(0)}{n!} \avg{z | ( D_0 + \sum_j D_j P_j )^n | z} \\
 &= \sum_z \sum_{n=0}^\infty \sum_{\{C_{\bm{i}_n}\}} \frac{f^{(n)}(0)}{n!} \avg{ z | C_{\bm{i}_n} | z},
\end{align}
where in the second line we sum over all operator sequences consisting of $n$ products of $D_0$ and $D_j P_j$ terms, which we denote $\{C_{\bm{i}_n}\}$. The multi-index $\bm{i}_n \equiv (i_1, \ldots, i_n)$ denotes the ordered sequence, i.e. $\bm{i}_3 = (3, 0, 1)$ indicates the sequence $C_{\bm{i}_3} = (D_1 P_1) D_0 (D_3 P_3)$, read from right to left. More generally, each $i_k \in \{0, \ldots, M\}$ denotes a single term from the PMR form $H = \sum_{j=0}^M D_j P_j$.

For convenience, we can separate the contributions from diagonal operators $D_j$ from off-diagonal permutations, $P_j$ which yields the following complicated expression,
\begin{equation}
 \label{eq:initial-tr-expansion}
 \Tr[f(H)] = \sum_z \sum_{q=0}^{\infty} \sum_{\Siq} D_{(z, \Siq)} \avg{ z | \Siq | z} \left( \sum_{n=q}^{\infty} \frac{f^{(n)}(0)}{n!} \sum_{\sum_i k_i = n - q} E^{k_0}_{z_0} \cdot \ldots \cdot E^{k_q}_{z_q} \right),
\end{equation}
justified in prior works~\cite{albash2017OffdiagonalExpansionQuantum, hen2018OffdiagonalSeriesExpansion}. We now summarize the notation, which will be used throughout. Firstly, $\Siq = P_{\iq} \cdot \ldots \cdot P_{\bm{i}_1}$ denotes a product over $q$ permutations, each taken from $\{P_j\}_{j=1}^M$, i.e., the multiset indices are now $\bm{i}_j \in \{1, \ldots, M\}$. Next, we denote $\ket{z_0} \equiv \ket{z}$ and $\ket{z_k} \equiv P_{\bm{i}_k} \cdot \ldots \cdot P_{\bm{i}_1} \ket{z}$. This allows us to define the ``diagonal energies'' as $E_{z_k} \equiv \avg{z_k | H | z_k} = \avg{z_k | D_0 | z_k}$ (recall $\ket{z_k}$ is a basis for $D_0$ not of $H$) and the off-diagonal ``hopping strengths'' $D_{(z, \Siq)} \equiv \prod_{k=1}^q \avg{z_k | D_{\bm{i}_k} | z_k}$. The sum over $\sum_{\sum k_i}$ is again the set of weak partitions of $n - q$ into $q+1$ integers, as described in the series expansion of the DDE.

Having defined the notation, we now observe that the diagonal contribution in parentheses is exactly the series expansion of $f[E_{z_0}, \ldots, E_{z_q}]$ in {\eqref{eq:dd-as-series-expansion}} via the change of variables $m \equiv n - q$. Furthermore, by the no fixed points property of the PMR, we know $\avg{z | \Siq | z} = 1$ if and only if $\Siq = \mathds{1}$. Put together, we find the non-trivial simplification,
\begin{equation}
 \label{eq:off-diagonal-series-expansion}
 \Tr[f(H)] = \sum_z \sum_{q=0}^{\infty} \sum_{\Siq = \mathds{1}} D_{(z, \Siq)} f[E_{z_0}, \ldots, E_{z_q}],
\end{equation}
where $\sum_{\Siq = \mathds{1}}$ is the sum over all products of $q$ permutations chosen from $\{P_j\}_{j=1}^M$ that evaluate to identity. This form is what we refer to as the off-diagonal series expansion of $\Tr[f(H)]$, and as the name suggests, it is a perturbative series in $q$, the size of the off-diagonal ``quantum dimension.'' Put differently, when $q = 0$, this is the expansion of $\Tr[f(D_0)]$ (which yields the classical partition function for $e^{-\beta D_0}$), and the remaining terms correct for the off-diagonal, non-commuting (or quantum) contribution.

\subsection{The PMR-QMC algorithm}\label{Xsec14-5.2}\label{Xsec14}
\label{subsec:PMR-QMC}

In the special case where $f(H) = e^{-\beta H}$, with $\beta \equiv 1 / T$ the inverse temperature, {(\cref{eq:off-diagonal-series-expansion})}
becomes an expansion for the partition function of the Hamiltonian $H$,
\begin{equation}
 \label{eq:partition-function}
 \mathcal{Z} = \Tr[\e^{-\beta H}]=\sum_{(z,\iq)} D_{(z,S_{{{\bf i}}_q})} \e^{-\beta [E_{z_0},\ldots,E_{z_q}]}
\end{equation}
where the summation above is shorthand for ${\sum_z \sum_{q=0}^{\infty} \sum_{S_{{{\bf i}}_q}=\mathds{1}}}$, namely, a sum over all `classical' states $z$ and all products of off-diagonal permutation matrices that evaluate to the identity operator~\cite{gupta2020PermutationMatrixRepresentation}. Given the above expression for $\mathcal{Z}$, we are now in a position to associate a QMC algorithm with the above expansion~\cite{gupta2020PermutationMatrixRepresentation,barash2024QuantumMonteCarlo, akaturk2024quantum, babakhani2025quantum}.

We define a QMC configuration ${\cal C }$ as any pair \hbox{${\cal C}=\{|z\rangle, S_{{{\bf i}}_q}\}$} [or ${(z,\iq )}$ for short] of
a basis state and a product $S_{{{\bf i}}_q}$ of permutation operators that evaluate to the identity element $P_0=\mathds{1}$ with associated generalized Boltzmann weight,
\begin{equation}
 \label{eq:gbw}
 w_{\mathcal{C}} \equiv w_{(z,{{\bf i}})} \equiv D_{(z,S_{{{\bf i}}_q})} \e^{-\beta [E_{z_0},\ldots,E_{z_q}]}.
\end{equation}
The configuration and associated weight can be conveniently visualized as a closed walk on a hypercube of classical basis states~\cite{gupta2020PermutationMatrixRepresentation, hen2021determining}. In general, this weight can be complex-valued through $D_{(z, \Siq)}$ since each $D_j$ can have complex entries. However, for every configuration $(z,\Siq)$ there is a conjugate configuration $(z,\Siq^{\dagger})$, which produces the conjugate weight $w_{(z,\iq^*)} = w^*_{(z,\iq)}$. Explicitly, for every closed walk $S_{\iq} =P_{i_q} \ldots P_{i_2} P_{i_1}$ there is a conjugate walk in the reverse direction, whose operator sequence is $S^\dagger_{{\bf i}_q} = P_{i_1}^T P_{i_2}^T \ldots P_{i_q}^T$. The imaginary parts of the complex-valued summands therefore do not contribute to the partition function and may be disregarded altogether. We may therefore take
\beq\label{eq:W}
W_{(z,\iq)} = \Re[w_{(z,{{\bf i}})}] = \Re[D_{(z,\iq)}] \e^{-\beta [E_{z_0},\ldots,E_{z_q}]} \,
\eeq
as the summands, or weights, of the expansion, averaging for every $(z,\iq)$ the configuration and its conjugate. Of course, the weights may in the general case be negative (when this happens, the system is said to possess a sign problem~\cite{hen2021determining}). Other choices, such as the absolute value of $w_{(z, \bm{i}_q)}$, are possible~\cite{babakhani2025quantum}.

We now describe a QMC algorithm based on sampling partition function configurations of {\eqref{eq:partition-function}} with associated weights {\eqref{eq:W}}. The Markov process begins with initial configuration \hbox{$\mathcal{C}_0=\{|z\rangle, S_0=\mathds{1}\}$} where $|z\rangle$ is a randomly generated initial classical (equivalently, diagonal) state and $S_0$ is the empty operator sequence. The weight of this initial configuration is
\beq
W_{\mathcal{C}_0} = w_{{\cal C}_0} = \e^{-\beta [E_z]}=\e^{-\beta E_z} ,
\eeq
i.e., the classical Boltzmann weight of the initial random state $|z\rangle$.

Next, we define a set of QMC updates to sample the configuration space, $(z, \Siq).$ A set of general, local updates was first proposed in Ref.~\cite{gupta2020PermutationMatrixRepresentation} that have successfully been applied to a variety of spin systems~\cite{gupta2020PermutationMatrixRepresentation, gupta2020elucidating}, superconducting circuit Hamiltonians~\cite{halverson2020efficient}, and Bose-Hubbard models~\cite{akaturk2024quantum}. In a major advancement, Ref.~\cite{barash2024QuantumMonteCarlo} showed that QMC moves that are ergodic and satisfy detailed balance can be found deterministically and automatically for arbitrary spin-1/2 Hamiltonians. This has recently been extended to arbitrary high spin systems~\cite{babakhani2025quantum} and the Bose-Hubbard model~\cite{babakhani2025quantum}. With appropriate modifications (i.e., see {\cref{subsec:fermionic-example}}), these ideas can also be extended to fermionic systems~\cite{babakhani2025fermionic}. As outlined in Sec. IV.D of Ref.~\cite{barash2024QuantumMonteCarlo}, these moves consist of (i) simple (local) swap, (ii) pair insertion and deletion, (iii) block swap, (iv) classical updates, (v) fundamental cycle completion, (vi) composite updates, and (vii) worm updates.

We leave a detailed discussion of these updates to the relevant references~\cite{barash2024QuantumMonteCarlo, akaturk2024quantum, babakhani2025quantum, babakhani2025fermionic}, but for a basic understanding, we discuss the simple (local) swap and classical moves. To explain simple swap, let $m \in \{1, \ldots, q-1\}$. By {\cref{corr:Siq-must-be-identity}}, if $\Siq = P_{\iq} \ldots P_{\bm{i}_m} P_{\bm{i}_{m+1}} \ldots P_{\bm{i}_1} = \mathds{1}$, then $\Siq' \equiv P_{\iq} \ldots P_{\bm{i}_{m+1}} P_{\bm{i}_{m}} \ldots P_{\bm{i}_1} = \mathds{1}$ as well. Yet, these sequences generally have different PMR-QMC weights, $W_{\mathcal{C}}$ and $W_{\mathcal{C'}}$, respectively. For example, the classical basis state $\ket{z_m} \equiv P_{\bm{i}_m} \ldots P_{\bm{i}_1} \ket{z} \rightarrow \ket{z'_m} \equiv P_{\bm{i}_{m+1}}, P_{\bm{i}_{m-1}}, \ldots, P_{\bm{i}_1} \ket{z}$. The acceptance probability for this update that satisfies detailed balance is simply
\begin{equation}
 \label{eq:acceptance-prob}
 p = \min\left( 1, \frac{W_{\mathcal{C'}}}{W_{\mathcal{C}}} \right).
\end{equation}
The classical update is simply to update $\ket{z} \rightarrow \ket{z'}$, i.e., by a local spin flip for spin systems, while leaving $\Siq$ unchanged. The acceptance probability to satisfy detailed balance is also {\eqref{eq:acceptance-prob}} in this case. This is an expensive move for quantum simulations since it requires updating all classical energies ${E_{z}, \ldots, E_{z_q}}$ in the divided difference multiset, but for classical simulations, it is both simple and actually the only move with non-zero probability. This emphasizes the nice feature that PMR-QMC naturally reduces to classical QMC when $H = D_0 = \Hdiag$\footnote{This was one of the original motivations for the off-diagonal series expansion~\cite{hen2018OffdiagonalSeriesExpansion} and associated QMC~\cite{albash2017OffdiagonalExpansionQuantum} that eventually became the early version of PMR-QMC~\cite{gupta2020PermutationMatrixRepresentation}.}. Other moves are more complicated, and in the case of the (v), require defining the notion of a fundamental cycle~\cite{barash2024QuantumMonteCarlo}.

A complete description of a full PMR-QMC algorithm also includes a discussion of how to estimate observables, which we spend the rest of this paper discussing in great detail.

\section{Estimation of static observables}\label{Xsec15-6}\label{Xsec15}
\label{sec:static-operator-estimators}
Given a static observable $A$, we call any function $A_{\mathcal{C}}$ such that
\begin{equation}
 \label{eq:estimator-form}
 \avg{A} \equiv \frac{\Tr[A e^{-\beta H}]}{\Tr[e^{-\beta H}]} = \frac{\sum_{(z, \iq)} w_{(z, \iq)} A_{(z, \iq)} }{ \sum_{(z, \iq)} w_{(z, \iq)}} \equiv \frac{\sum_{\mathcal{C}} w_{\mathcal{C}} A_{\mathcal{C}} }{ \sum_{\mathcal{C}} w_{\mathcal{C}}}
\end{equation}
a PMR-QMC estimator of $A$. Generally, estimators are not unique, so we can more precisely write
\begin{equation}
 A_{\mathcal{C}} \estimates \avg{A},
\label{Xeqn37-47}
\end{equation}
to be read ``$A_{\mathcal{C}}$ is an (unbiased)~\footnote{Deriving biased estimators is outside the scope of the present work.} estimator of $\avg{A}$'' when it is important to make the non-uniqueness clear. We use both notations as convenient. As described in {\cref{sec:pmr-qmc}}, $\mathcal{C} \equiv (z, \iq)$ specifies an instantaneous PMR-QMC configuration and $w_{\mathcal{C}}$ is the generalized Boltzmann weight defined in {\eqref{eq:gbw}}. Throughout the course of the PMR-QMC simulation, one can estimate $A$ by occasionally computing $A_{\mathcal{C}}$ and averaging over such realizations.

As with weight in {\eqref{eq:W}}, one can always define a real-valued estimator via,
\begin{equation}
 \label{eq:real-estimator-rule}
 \avg{A} = \frac{\sum_{\mathcal{{C}}} W_\mathcal{C} (\Re[A_\mathcal{C} w_\mathcal{C}] / W_{\mathcal{C}})}{\sum_{\mathcal{C}} W_\mathcal{C}},
\end{equation}
where $W_\mathcal{C} \equiv \Re[w_\mathcal{C}]$ as before~\cite{barash2024QuantumMonteCarlo}. Comparing with {\eqref{eq:estimator-form}}, we see that $\Re[A_\mathcal{C} w_\mathcal{C}] / W_{\mathcal{C}}$ is the instantaneous quantity, i.e., the real-valued equivalent of $A_\mathcal{C}$. This subtlety does not play a role in our derivations; therefore, we henceforth derive expressions for potentially complex estimators $A_{\mathcal{C}}$ throughout for simplicity. However, we note that in our actual implementation~\cite{ezzell2025code}, we use {\eqref{eq:real-estimator-rule}}.

In the rest of this section, we derive explicit, closed form, and exact estimators for various static observables. A high-level summary of the different observables we consider is contained in {\cref{sec:overview-of-estimators}}. We ultimately show ({\cref{subsec:estimating-arbitrary-operators}}) that thermal expectation values of arbitrary static operators can be estimated in principle. Along the way, we derive several simpler estimators (i.e., for $\avg{H^k}$) that are more efficient when applicable. A summary of static operator estimators and complexities is provided in {\cref{tab:static-complexity}}, though complexity itself is not formally defined until {\cref{subsec:complexity-interlude}}.

\begin{table}
\caption{A summary of static observable estimators we derive in this work and their computational complexity in terms of the PMR-QMC off-diagonal expansion order, $q$. { As elsewhere throughout this work, $\Lambda$ is an arbitrary diagonal observable, $\Hdiag = D_0$ is the pure diagonal part of the Hamiltonian, $\Hoffdiag = H - \Hdiag$, and $D_l P_l$ are diagonal-permutation pairs contained in the PMR decomposition of $H$}.}{%
\begin{tabular}{lll}
\toprule
{Static observable} & {Estimator} & {Estimator complexity} \\
\colrule
$\avg{\Lambda^k}$ & {\eqref{eq:diag-to-k-estimator}} & $O(1)$ \\
$\avg{\Hdiag^k}$ & {\eqref{eq:d0-k-estimator}} & $O(1)$ \\
$\avg{H}$ & {\eqref{eq:H-estimator}} & $O(1)$ \\
$\avg{H^2}$ & {\eqref{eq:H2-estimator}} & $O(1)$ \\
$\avg{H^k}$ & {\eqref{eq:H-n-estimator}} & $O(k)$ \\
$\avg{\Hoffdiag}$ & {\eqref{eq:hoffdiag-estimator}} & $O(1)$ \\
$\avg{\Hoffdiag^2}$ & {\eqref{eq:hoffdiag2-estimator}} & $O(1)$ \\
$\avg{D_l P_l}$ & {\eqref{eq:dl-pl-estimator-delta}} & $O(1)$ \\
$\avg{\Lambda_l D_l P_l}$ & {\eqref{eq:al-dl-pl-estimator}} & $O(1)$ \\
$\avg{\sum_{l=0}^{K-1} \Lambda_l D_l P_l}$ & {\eqref{eq:al-dl-pl-estimator}} & $O(K)$ \\
$\avg{\Lambda_1 D_{\bm{l}_1} P_{\bm{l}_1} \cdot \ldots \cdot \Lambda_l D_{\bm{l}_L} P_{\bm{l}_L} \Lambda_{L+1}}$ & {\eqref{eq:A-canonical-estimator}} & $O(L)$ \\
\botrule
\end{tabular}}
\label{tab:static-complexity}
\end{table}

\subsection{Estimating purely diagonal operators}\label{Xsec16-6.1}\label{Xsec16}
\label{subsec:estimating-pure-diag-operators}
Suppose $\Lambda$ is an arbitrary diagonal operator with matrix elements $\Lambda(z) \equiv \avg{z | \lambda | z}.$ By writing out the trace and performing an off-diagonal series expansion of $e^{-\beta H}$, we find,
\begin{align}
 \Tr[\Lambda e^{-\beta H}] = \sum_{z} \Lambda(z) \avg{z | e^{-\beta H} | z} = \sum_{z} \sum_{\Siq = \mathds{1}} w_{(z, \Siq)} \Lambda(z).
\end{align}
By inspection of this expression and the form of a general estimator in {\eqref{eq:estimator-form}}, we immediately conclude
\begin{equation}
 \label{eq:diag-estimator}
 (\Lambda)_{\mathcal{C}} \equiv \Lambda(z) \estimates \avg{\Lambda}.
\end{equation}
We note that if $H = D_0$ were a classical Hamiltonian, then $\Lambda(z)$ is simply a straightforward classical MC estimator---consistent with the general logic and spirit of the off-diagonal series expansion. This becomes especially apparent when $\Lambda = D_0$ for which $E_{z} \equiv \avg{z | D_0 | z}$ and we can write
\begin{equation}
 \label{eq:d0-estimator}
 (D_0)_{\mathcal{C}} \equiv E_z \estimates \avg{D_0} \equiv \avg{\Hdiag}.
\end{equation}
Note that the second equality is by definition---$\Hdiag$ was used in {\cref{sec:numerical-demo}} for clarity since the PMR decomposition for which $D_0 \equiv \Hdiag$ had not yet been introduced.

As an additional observation, {\eqref{eq:diag-estimator}} also directly gives us the estimator
\begin{equation}
 \label{eq:diag-to-k-estimator}
 (\Lambda^k)_{\mathcal{C}} \equiv (\Lambda(z))^k \estimates \avg{\Lambda^k}
\end{equation}
since any power of $\Lambda$ is itself a diagonal operator. Hence, we can also write,
\begin{equation}
 \label{eq:d0-k-estimator}
 (D_0^k)_{\mathcal{C}} \equiv (E_z)^k \estimates \avg{(D_0)^k}.
\end{equation}

\subsection{Estimating functions of the Hamiltonian}\label{Xsec17-6.2}\label{Xsec17}
\label{subsec:estimating-f(H)-things}
Let $g(H)$ be an arbitrary analytic function of the Hamiltonian. Choosing $f(H) = g(H) e^{-\beta H}$ in the off-diagonal series expansion of $f(H)$ in {\eqref{eq:off-diagonal-series-expansion}}, we find
\begin{align*}
 \Tr[g(H) e^{-\beta H}] &= \sum_{(z, \iq)} D_{(z, \iq)} \sum_{j=0}^q g[E_{z_j}, \ldots, E_{z_q}] e^{-\beta [E_{z_0}, \ldots, E_{z_j}]}
\end{align*}
by direct usage of the Leibniz rule for divided differences (see {\eqref{eq:leibniz-rule}}). To coax this expression into a bona fide PMR estimator, we simply multiply by $e^{-\beta [E_z, \ldots, E_{z_q}]} / e^{-\beta [E_z, \ldots, E_{z_q}]}$ upon which we find,
\begin{equation}
 \avg{g(H)} = \frac{ \sum_{(z, \iq)} w_{(z, \iq)} \left( \sum_{j=0}^q g[E_{z_j}, \ldots, E_{z_q}] \frac{ e^{-\beta [E_{z_0}, \ldots, E_{z_j}]} } { e^{-\beta [E_z, \ldots, E_{z_q}]} } \right) }{ \sum_{(z, \iq)} w_{(z, \iq)} },
\label{Xeqn43-54}
\end{equation}
where we can identify
\begin{equation}
 \label{eq:g(H)-estimator}
 (g(H))_{\mathcal{C}} \equiv \sum_{j=0}^q g[E_{z_j}, \ldots, E_{z_q}] \frac{ e^{-\beta [E_{z_0}, \ldots, E_{z_j}]} } { e^{-\beta [E_z, \ldots, E_{z_q}]} } \estimates \avg{g(H)}
\end{equation}
as the quantity to compute and collect during QMC simulation in order to estimate $\avg{g(H)}$.

In the special case $g(H) = H$, we find
\begin{equation}
 (H)_\mathcal{C} \equiv E_{z_q} + \frac{e^{-\beta [E_{z_0}, \ldots, E_{z_{q-1}}}]}{e^{-\beta [E_{z_0}, \ldots, E_{z_q}]}} \estimates \avg{H}.
\label{Xeqn45-56}
\end{equation}
In the spirit of the off-diagonal series expansion, this expression has a purely diagonal/classical contribution, $E_{z_0}$, and an off-diagonal correction, the ratio of DDEs. Formally, this derivation assumes $q \geq 1$ for the ratio of DDEs to appear, and more precisely we can write,
\begin{equation}
 \label{eq:H-estimator}
 (H)_{\mathcal{C}} = E_{z_q} + \mathbf{1}_{q\geq 1} \frac{e^{-\beta [E_{z_0}, \ldots, E_{z_{q-1}}}]}{e^{-\beta [E_{z_0}, \ldots, E_{z_q}]}} = \begin{cases}
 E_{z_q} & q = 0 \\
 E_{z_q} + \frac{e^{-\beta [E_{z_0}, \ldots, E_{z_{q-1}}}]}{e^{-\beta [E_{z_0}, \ldots, E_{z_q}]}} & q > 0,
 \end{cases}
\end{equation}
where $\mathbf{1}_{q \geq 1}$ is the indicator function that is 0 when $q < 1$ and 1 when $q \geq 1$. Henceforth, we will continue this convention of assuming $q$ is large enough to support all terms in our derivations but providing concrete corrections to specific estimators.

Similarly, we can immediately write the estimator for any integer power of $H$,
\begin{equation}
 \label{eq:H-n-estimator}
 (H^k)_{\mathcal{C}} \equiv \sum_{j=0}^{\max\{k, q\}} [E_{z_j}, \ldots, E_{z_q}]^n \frac{e^{-\beta [E_{z_0}, \ldots, E_{z_j}]}}{e^{-\beta [E_{z_0}, \ldots, E_{z_q}]}} \estimates \avg{H^k},
\end{equation}
where the explicit expression $[E_{z_0}, \ldots, E_{z_j}]^n$ can be deduced from {\eqref{eq:poly-div-diff}}. For example,
\begin{align}
 \label{eq:H2-estimator}
 (H^2)_{\mathcal{C}} &\equiv E_{z_q}^2 + \frac{\mathbf{1}_{q \geq 1}(E_{z_q} + E_{z_{q-1}}) e^{-\beta [E_{z_0}, \ldots, E_{z_{q-1}}]} + \mathbf{1}_{q \geq 2}e^{-\beta [E_{z_{0}}, \ldots, E_{z_{q-2}}]} }{ e^{-\beta [E_{z_0}, \ldots, E_{z_q}]} }
\end{align}
and
\begin{multline}
 \label{eq:H3-estimator}
 (H^3)_{\mathcal{C}} \equiv E_{z_q}^3 + \mathbf{1}_{q\geq 1} \frac{(E^2_{z_q} + E_{z_q} E_{z_{q-1}} + E^2_{z_{q-1}}) e^{-\beta [E_{z_0}, \ldots, E_{z_{q-1}}]}}{e^{-\beta [E_{z_0}, \ldots, E_{z_q}]}} \\
 + \mathbf{1}_{q\geq 2} \frac{(E_{z_q} + E_{z_{q-1}} + E_{z_{q-2}}) e^{-\beta [E_{z_0}, \ldots, E_{z_{q-2}}]}}{e^{-\beta [E_{z_0}, \ldots, E_{z_q}]}}
 + \mathbf{1}_{q\geq 3} \frac{e^{-\beta [E_{z_0}, \ldots, E_{z_{q-3}}]} }{ e^{-\beta [E_{z_0}, \ldots, E_{z_q}]} },
\end{multline}
and so on. As mentioned in our numerical demonstration ({\cref{sec:numerical-demo}}), one can straightforwardly use our estimators $(H^2)_{\mathcal{C}}$ and $(H)_{\mathcal{C}}$ to estimate specific heat via {\eqref{eq:specific-heat}} with a simple jackknife binning analysis~\cite{berg2004introduction}. Finally, we remark that our estimators for $\avg{H}$ and $\avg{H^2}$ agree with those derived using a physics-inspired derivation~\cite{barash2024QuantumMonteCarlo} despite using an approach only utilizing general divided difference properties, highlighting the versatility and generality of the PMR-QMC approach.

\subsection{Estimating the pure off-diagonal portion of the Hamiltonian}\label{Xsec18-6.3}\label{Xsec18}
\label{subsec:estimating-Hoffdiag}

By definition $\Hoffdiag = H - \Hdiag$, so by linearity of expectation, we immediately find $\avg{\Hoffdiag} = \avg{H} - \avg{\Hdiag}$. Again by linearity, we can simply write $(\Hoffdiag)_{\mathcal{C}} = (H)_{\mathcal{C}} - (\Hdiag)_{\mathcal{C}}$ which gives
\begin{equation}
 \label{eq:hoffdiag-estimator}
 (\Hoffdiag)_{\mathcal{C}} \equiv \frac{e^{-\beta [E_{z_0}, \ldots, E_{z_{q-1}}}]}{e^{-\beta [E_{z_0}, \ldots, E_{z_q}]}} \estimates \avg{\Hoffdiag}
\end{equation}
by combining {\eqref{eq:d0-estimator}}, {\eqref{eq:H-estimator}}, and the $E_{z_0} = E_{z_q}$ periodicity induced by only $\Siq = \mathds{1}$ terms contributing~\footnote{This periodicity of $\ket{z}$ comes from the trace, and it is the PMR analogue of $\beta$ periodicity in world line methods.}. Similarly, since $\Hoffdiag^2 = (H - \Hdiag)^2,$ then by the cyclicity of the trace, we can write
\begin{align}
 \avg{\Hoffdiag^2}& = \avg{H^2 - H \Hdiag - \Hdiag H + \Hdiag^2} \nonumber\\&= \avg{H^2} - 2 \avg{\Hdiag H} + \avg{\Hdiag^2},
\label{Xeqn49-62}
\end{align}
by linearity and cyclicity of the trace. Expanding the numerator of the middle term,
\begin{equation}
 \Tr[\Hdiag H] = \sum_z E_z \avg{z | H e^{-\beta H} | z},
\label{Xeqn50-63}
\end{equation}
we find
\begin{equation}
 E_z (H)_{\mathcal{C}} \equiv E_{z_0} \left(E_{z_0} + \mathbf{1}_{q\geq 1} \frac{e^{-\beta [E_{z_1}, \ldots, E_{z_q}}]}{e^{-\beta [E_{z_0}, \ldots, E_{z_q}]}} \right) \estimates \avg{\Hdiag H},
\label{Xeqn51-64}
\end{equation}
by applying the Leibniz rule off-diagonal expansion of $\avg{z | H e^{-\beta H} | z}$ as in {\cref{subsec:estimating-f(H)-things}}. Together with {\eqref{eq:d0-k-estimator}}, {\eqref{eq:H2-estimator}}, and $E_{z_0} = E_{z_q}$, we find
\begin{equation}
 \label{eq:hoffdiag2-estimator}
 (\Hoffdiag^2)_{\mathcal{C}} \equiv (H)_{\mathcal{C}} + E_{z_q}(E_{z_q} - 2 (H)_{\mathcal{C}}) \estimates \avg{\Hoffdiag^2}.
\end{equation}
Higher powers can be continually derived in this way by using {\eqref{eq:d0-k-estimator}} and {\eqref{eq:H-n-estimator}}.

\subsection{Estimating terms that comprise the Hamiltonian}\label{Xsec19-6.4}\label{Xsec19}
\label{subsec:estimating-dl-pl}
In this section, we derive estimators for any term of the form $D_l P_l$, i.e., any term contained within the PMR decomposition $H = \sum_j D_j P_j$. As a shorthand, we write $D_j P_j \in H$ which also implies $P_j \in G$, the PMR group in which we have decomposed $H$. Further, we assume $\diag{D_j} \neq \mathbf{0}$ since any such trivial term where $\diag{D_j} = \mathbf{0}$ clearly satisfies $\avg{D_j P_j} = \avg{0} = 0$.

By {\cref{subsec:estimating-pure-diag-operators}}, we have already solved the $l = 0$, $P_0 = \mathds{1}$ case with $E_{z_0} \estimates \avg{D_0},$ so we assume $l \neq 0$. Firstly, a straightforward off-diagonal series expansion yields,
\begin{align}
 \label{eq:dl-pl-first-step}
 \Tr[D_l P_l e^{-\beta H}] &= \sum_z \avg{z | D_l | z} \avg{z | P_l e^{-\beta H} | z} = \sum_{z} \sum_{\Sip} D_l(z) w_{(z, \Sip)} \avg{z | P_l \Sip | z},
\end{align}
where $D_l(z) \equiv \avg{z | D_l | z}$. By the PMR properties of the permutation, $\avg{z | P_l \Sip | z} = 1$ if and only if $P_l \Sip = \mathds{1}$ and is $0$ otherwise. The key insight first observed in Refs.~\cite{albash2017OffdiagonalExpansionQuantum, gupta2020CalculatingDividedDifferences}, is that we can combine $P_l \Sip \equiv \delta^{(q)}_{P_l} \Siq$ where
\begin{equation}
 \delta^{(q)}_{P_l} \equiv
 \begin{cases}
 1 & P_{\iq} = P_l \\
 0 & P_{\iq} \neq P_l
 \end{cases}.
\label{Xeqn53-67}
\end{equation}
This enforces that $\Siq$ ends with the permutation $P_l$, and we can treat $D_l P_l$ as the $q{\text{th}}$ off-diagonal contribution in a series expansion involving $\Siq$ for $q = p + 1$ instead of $\Sip$. In the end, we can write,
\begin{align}
 \Tr[D_l P_l e^{-\beta H}] = \sum_z \sum_{\Siq = \mathds{1}} w_{(z, \Siq)} \left( \delta^{(q)}_{P_l} \frac{ e^{-\beta [E_{z_0}, \ldots, E_{z_{q-1}}]} }{ e^{-\beta [E_{z_0}, \ldots, E_{z_q}]} } \right),
\end{align}
where we have replaced the actual weight from the off-diagonal at order $p$,
\begin{equation}
 w_{(z, \Sip)} = \prod_{j=1}^p D_{\bm{i}_j}(z_j) e^{-\beta [E_{z_0}, \ldots, E_{z_p}]},
\label{Xeqn54-69}
\end{equation}
with a `fictitious' total off-diagonal weight at order $q = p + 1$,
\begin{equation}
 w_{(z, \Siq)} = \prod_{j=1}^{p+1} D_{\bm{i}_j}(z_j) e^{-\beta [E_{z_0}, \ldots, E_{z_{p+1}}]} \equiv \prod_{j=1}^{q} D_{\bm{i}_j}(z_j) e^{-\beta [E_{z_0}, \ldots, E_{z_{q}}]},
\label{Xeqn55-70}
\end{equation}
using the relation
\begin{align}
 \label{eq:siq-weight-relation}
 w_{(z, \Siq)} = w_{(z, \Sip)} D_l(z_q) \frac{ e^{-\beta [E_{z_0}, \ldots, E_{z_q}]} }{ e^{-\beta [E_{z_0}, \ldots, E_{z_{q-1}}]} }.
\end{align}
In total, this gives us an estimator,
\begin{equation}
 \label{eq:dl-pl-estimator-delta}
 (D_l P_l)_{\mathcal{C}} \equiv \delta^{(q)}_{P_l} \frac{ e^{-\beta [E_{z_0}, \ldots, E_{z_{q-1}}]} }{ e^{-\beta [E_{z_0}, \ldots, E_{z_q}]} } \estimates \avg{D_l P_l}.
\end{equation}
Comparing this expression with the estimator for $\avg{\Hoffdiag}$ in {\eqref{eq:hoffdiag-estimator}}, we see that $(D_l P_l)_{\mathcal{C}}$ is the off-diagonal contribution from the $P_l$ permutation and $\sum_{l>0}(D_l P_l)_{\mathcal{C}} = (\Hoffdiag)_{\mathcal{C}}$ as expected from linearity.

By the linearity of expectation, we have thus also derived an estimator for any sum of terms in $H$,
\begin{equation}
 \sum_{l \in S_A} (D_l P_l)_{\mathcal{C}} \estimates \avg{ \sum_{l \in S_A} D_j P_j }, \ \ \ S_A \subset \{0, \ldots, M\}.
\label{Xeqn57-73}
\end{equation}
Of course, if $S_A = \{0\}$, $S_A = \{0, 1, \ldots, M\},$ or $S_A = \{1, \ldots, M\}$, it is more efficient to simply use the direct estimators for $\avg{D_0},$ $\avg{H},$ or $\avg{\Hoffdiag}$ as in {\eqref{eq:d0-estimator}}, \eqref{eq:hoffdiag-estimator}, \eqref{eq:H-estimator}, respectively.

\subsection{A generalization of estimating terms that comprise the Hamiltonian: Canonical estimators}\label{Xsec20-6.5}\label{Xsec20}
\label{subsec:estimating-al-pl}
Having established that we can easily estimate $\avg{D_l P_l}$ for $D_l P_l \in H$, we now discuss the slight generalization $\avg{\wtlam_l P_l}$ for $\wtlam_l$ a general diagonal matrix. As before, if $l = 0$, we find by {\cref{subsec:estimating-pure-diag-operators}} that $\wtlam_0(z) \estimates \avg{\wtlam_0},$ so we assume $l \neq 0$. By the off-diagonal series expansion,
\begin{equation}
 \label{eq:al-pl-expansion}
 \Tr[\wtlam_l P_l e^{-\beta H}] = \sum_{z}\sum_{\Sip} \wtlam_l(z) w_{(z, \Sip)} \avg{z | P_l \Sip | z},
\end{equation}
and proceeding as before by introducing $P_l \Sip = \delta^{(q)}_{P_l} \Siq$ and using the conversion in \eqref{eq:siq-weight-relation}, we find
\begin{equation}
 \label{eq:al-pl-estimator}
 (\wtlam_l P_l)_{\mathcal{C}} \equiv \delta^{(q)}_{P_l} \frac{\wtlam_l(z)}{D_l(z)} \frac{ e^{-\beta [E_{z_0}, \ldots, E_{z_{q-1}}]} }{ e^{-\beta [E_{z_0}, \ldots, E_{z_q}]} },
\end{equation}
which---unlike {\eqref{eq:dl-pl-estimator-delta}}---now contains a ratio of matrix elements $\wtlam_l(z) / D_l(z)$. This small difference can lead to incorrect estimation for some $\wtlam_l P_l$, and as one might expect intuitively, issues may arise whenever $D_l(z) = 0$. Indeed, this ``division by zeros'' may (but does not always) lead to PMR-QMC estimation inaccuracy, as discussed in {\cref{subsec:illustrative-examples}}.

A sufficient condition to sidestep this problem completely is $D_l(z) = 0 \Rightarrow \wtlam_l(z) = 0$. As a special case, this includes diagonals that are never zero, $D_l(z) \neq 0$ for any $\ket{z}$. A necessary and sufficient condition to satisfy this implication is $\wtlam_l = \Lambda_l D_l$ for $\Lambda_l$ an arbitrary diagonal matrix, as zeros of $D_l$ are now passed on to $\wtlam_l$. Intuitively, the corresponding estimator,
\begin{equation}
 \label{eq:al-dl-pl-estimator}
 (\Lambda_l D_l P_l)_{\mathcal{C}} \equiv \delta^{(q)}_{P_l} \Lambda_l(z) \frac{ e^{-\beta [E_{z_0}, \ldots, E_{z_{q-1}}]} }{ e^{-\beta [E_{z_0}, \ldots, E_{z_q}]} }.
\end{equation}
does not divide by anything that can possibly be zero. By a similar derivation, we can readily estimate the product of two such operators,
\begin{equation}
 (\Lambda_k D_k P_k \Lambda_l D_l P_l)_{\mathcal{C}} \equiv \delta_{P_k}^{(q)} \delta_{P_l}^{(q-1)} \Lambda_k(z_q) \Lambda_l(z_{q-1}) \frac{e^{-\beta [E_{z_0}, \ldots, E_{z_{q-2}}]}}{e^{-\beta [E_{z_0}, \ldots, E_{z_q}]}}.
\label{Xeqn61-77}
\end{equation}
More generally, the PMR-QMC estimator for a finite product
\begin{equation}
 \label{eq:product-canonical-operator}
 A \equiv \Lambda_1 D_{\bm{l}_1} P_{\bm{l}_1} \Lambda_2 D_{\bm{l}_2} P_{\bm{l}_2} \cdot \ldots \cdot \Lambda_l D_{\bm{l}_L} P_{\bm{l}_L} \Lambda_{L+1} = \prod_{s=1}^L (\Lambda_l D_{\bm{l}_s} P_{\bm{l}_s}) \Lambda_{L+1},
\end{equation}
can be written
\begin{align}
 \label{eq:A-canonical-estimator}
 (A)_{\mathcal{C}} &\equiv \mathcal{A}(\mathcal{C}) \frac{e^{-\beta [E_{z_0}, \ldots, E_{z_{q-L}}]}}{e^{-\beta [E_{z_0}, \ldots, E_{z_q}]}} \\
 \mathcal{A}(\mathcal{C}) &\equiv \Lambda_{L+1}(z_{q-L}) \prod_{s=1}^L \Lambda_s (z_{q-s+1}) \delta^{(q-s+1)}_{P_{\bm{l}_{s}}}.
\end{align}
{
For brevity, we call operators that can be written in the form of $A$ in {\eqref{eq:product-canonical-operator}} \emph{canonical}, and their corresponding \emph{canonical estimators} can always be evaluated in practice without possible issues of ``division by zero.'' By linearity, we can also estimate linear combinations of such canonical operators, which we also refer to as canonical. Conversely, an operator not given in canonical form is said to be in \emph{uncanonical form}.
}

{
\subsection{Illustrative examples}\label{Xsec21-6.6}\label{Xsec21}
\label{subsec:illustrative-examples}

Before further generalization and abstraction, it is instructive to consider several examples, especially to explore the consequences of estimating observables in \emph{canonical} or \emph{uncanonical} forms. We focus on spin-1/2 models decomposed in the Pauli-X basis for simplicity.

\subsubsection{A two qubit spin-1/2 transverse-field Ising model}\label{Xsec22-6.6.1}\label{Xsec22}

Consider a two qubit transverse-field Ising model (TFIM) cast into PMR form,
\begin{equation}
 H = -Z_1 Z_2 + h (X_1 + X_2) \longrightarrow D_0 + D_1 P_1 + D_2 P_2
\label{Xeqn63-81}
\end{equation}
for $D_0 = - Z_1 Z_2, D_1 = D_2 = h \mathds{1}$ and $P_1 = X_1, P_2 = X_2$. From {\cref{subsec:estimating-pure-diag-operators}}, \cref{subsec:estimating-f(H)-things}, \cref{subsec:estimating-Hoffdiag}, we know it is straightforward to estimate $\avg{H}$, $\avg{D_0} = \avg{-Z_1Z_2}$, $\avg{\Hoffdiag} = \avg{h (X_1 + X_2)}$ and powers thereof. Beyond $\avg{D_0}$, we can actually estimate any diagonal matrix, i.e., any linear combination of $\mathds{1}, Z_1, Z_2,$ and $Z_1 Z_2$. From {\cref{subsec:estimating-dl-pl}}, we can also estimate $\avg{D_1 P_1} = \avg{h X_1}$ and $\avg{D_2 P_2} = \avg{h X_2}$. These results follow by translating the general framework to this specific example.

We now consider the seemingly small generalization of \emph{canonical operators} in {\cref{subsec:estimating-al-pl}}, specifically {\eqref{eq:product-canonical-operator}} and demonstrate that, for this example, we can readily estimate any observable. First, we observe that $X_k = \Lambda D_k P_k$ for $\Lambda = \mathds{1} / h$ and hence, we also directly estimate $\avg{X_1}$ and $\avg{X_2}$ without the pre-factor. More importantly, we can estimate $\avg{Y_k}$ since $Y_k = \Lambda_k D_k P_k$ for $\Lambda_k = -i Z_k / h $. Furthermore, $X_1 X_2 = (\Lambda D_1 P_1)(\Lambda D_2 P_2)$ and similarly $X_1 Y_2 = (\Lambda D_1 P_2)(\Lambda_2 D_2 P_2)$, $Y_1 X_2 = (\Lambda_1 D_1 P_1)(\Lambda D_2 P_2)$.

All together, we have shown it is possible to estimate any Pauli string, and hence, any linear combination of them. Since the Pauli strings form a basis, we have thus demonstrated that it is possible to estimate arbitrary operators. We have also seen the value in considering $\Lambda D P$ products in practice---specifically how it allows us to write the canonical form for operators that are not explicitly in $H$, i.e., $X_1 X_2 \notin H$ yet it has a canonical form and estimator.

\subsubsection{An arbitrary spin-1/2 TFIM}\label{Xsec23-6.6.2}\label{Xsec23}
The simple two qubit example result that all operators can be put into canonical form extends to an arbitrary spin-1/2 TFIM. First, we write
\begin{equation}
 H_Z + h \sum_{i=1}^n X_i \longrightarrow D_0 + \sum_{i=1}^n D_i P_i
\label{Xeqn64-82}
\end{equation}
for $H_Z = D_0$ any diagonal matrix and $D_i = h \mathds{1}, P_i = X_i$. On the one hand, this follows straightforwardly from the logic of the two qubit example. More generally, this follows from two simple observations.

First, the set of all possible products of permutations in $H$, $\avg{P_1, \ldots, P_n}_g$ is simply the PMR group of all Pauli-X strings, $G_X$. Second, each diagonal $D_i$ is full rank, and hence, $\Lambda D_i$ can be made to equal any desired diagonal. Together, we can thus write any operator of the form $\tilde{D} \tilde{P}$ for $\tilde{P} \in G_X$ and $\tilde{D}$ any diagonal into canonical form. Since $G_X$ is a PMR basis, we can thus write any operator as a sum of such $\tilde{D} \tilde{P}$ terms, and hence, as a canonical operator for the TFIM written in the Pauli-X PMR basis.

\subsubsection{An XX chain}\label{Xsec24-6.6.3}\label{Xsec24}

Next we consider estimation of general observables for the 3-spin $XX$ chain,
\begin{equation}
 H = X_1 X_2 + X_2 X_3 \longrightarrow D_1 P_1 + D_2 P_2
\label{Xeqn65-83}
\end{equation}
for $D_1 = D_2 = \mathds{1}$ and $P_1 = X_1 X_2$, $P_2 = X_2 X_3$. From our experience with the TFIM example, it is instructive to consider $\avg{P_1, P_2}_g = \{\mathds{1}, X_1 X_2, X_2 X_3, X_1 X_3\}$. Unlike the TFIM, this is a strict subgroup of the all Pauli-X strings on 3 spins, $\avg{P_1, P_2}_g < G_X$. This suggests that it is not possible to put every operator into a canonical form.

As an explicit example, consider $X_1$. Clearly, $X_1 \notin \avg{P_1, P_2}_g$, and in this case, this implies it cannot be written as a product of the form $(\Lambda_l D_l P_l)\ldots (\Lambda_k D_k P_k)$. However, as one can explicitly check for this example, $\avg{X_1} = 0$, and we will show this follows exactly because $X_1 \notin \avg{P_1, P_2}_g$. More generally, all observables that cannot be put into canonical form for this model have trivial expectation value.

\subsubsection{Biased estimation example: a zero along the diagonal}\label{Xsec25-6.6.4}\label{Xsec25}
\label{subsubsec:a-model-with-zeros}

Our prior examples all contains $D P$ terms for which each $D$ is full rank. We now consider an example where one of the PMR diagonals has zeros. Namely,
\begin{equation}
 H = Z_1 X_2 + X_2 + Y_1 X_2 \longrightarrow D_1 P_1 + D_2 P_2
\label{Xeqn66-84}
\end{equation}
for $P_1 = X_2, P_2 = X_1 X_2$ and $D_1 = \mathds{1} + Z_1, D_2 = -i Z_1.$ We now consider estimating $\avg{X_2}$. Since $D_1 = \diag{2, 2, 0, 0}$ has two zeros, $D_1 P_1 \neq X_2$, and it cannot be made canonical by some trivial rescaling, $\Lambda D_1 P_1$. Furthermore, we cannot conclude $\avg{X_2} = 0$ generally either since $X_2 \in \avg{P_1, P_2}_g$, and in fact, we find $\avg{X_2} = -\tanh{\beta}$ by direct evaluation.

For specific values of $\beta$, we can furthermore use the uncanonical estimator, {\eqref{eq:al-pl-estimator}}, to predict the value of $\avg{X_2}$ using our Monte Carlo code. At $\beta = 2$, our code predicts the value $-0.83(1)$ which clearly disagrees with the exact value $\tanh{2} \approx 0.964$. This example shows explicitly that the use of a uncanonical estimator truly can lead to biased and incorrect estimation. Specifically, the direct uncanonical estimator is incorrect because it is possible to construct configurations for which the PMR-QMC weight is zero whereas the corresponding direct off-diagonal expansion summand is non-zero. This is how formal division by zero actually corrupts or biases the estimator in practice.   

For example, this happens when $\ket{z} = \ket{10}$ and $\Siq = P_1 P_2 P_1 P_2$. More explicitly, the final $P_1$ means $w_{(z, \Siq)} \propto \avg{10 | D_1 | 10} = 0$. Yet, the corresponding direct off-diagonal expansion summand is (up to the divided difference contribution),
\begin{equation}
 \avg{10 | P_1 (D_2 P_2) (D_1 P_1) (D_2 P_2) | 10} = 2.
\label{Xeqn67-85}
\end{equation}
Fortuitously, this explicit explanation of the problem also suggests a solution: we can write $X_2$ in a canonical form,
\begin{equation}
 \label{eq:canonical-x2-example}
 X_2 = \frac12 \mathds{1} (D_1 X_2 + (D_2 X_1 X_2)(D_1 X_2)(D_2 X_1 X_2)).
\end{equation}
Unlike the TFIM or XX examples, this decomposition requires a sum of two canonical terms, making it less obvious. Nevertheless, it is straightforward to verify, and the underlying principle of this decomposition can be made more apparent. Namely, each term has a permutation portion that evaluates to $X_2$. Evaluating the thermal expectation value of each term reveals $\frac12 \avg{D_1 X_2}_{\beta = 2} \approx -0.833$ and $\frac12 \avg{(D_2 X_1 X_2)(D_1 X_2)(D_2 X_1 X_2)}_{\beta = 2} \approx -0.131$. The former value is exactly what our naive QMC predicts for $\avg{X_1}_{\beta = 2}$ using the uncanonical estimator, and we now understand why. The naive estimator is biased exactly because $\frac12 D_1 X_2$ has two zeros that $X_2$ itself does not.

Finally, we remark that zeros in $D_1$ alone are not enough to predict the failure of the uncanonical estimator. Indeed, $\avg{X_2}$ can be estimated using {\eqref{eq:al-pl-estimator}} in $H' = Z_1 X_2 + Z_2$ without bias. We can understand this luck by first simplifying a direct off-diagonal expansion, 
\begin{align}
    \Tr[X_2 e^{-\beta H'}] &= \sum_z \sum_{p = 0}^\infty (\avg{z | D_1 |z})^{2p+1} (-\beta)^{2p + 1}.
\end{align}
In our simplifications, we utilized $D_0 = 0 \implies E_{z_j} = 0$ for all $z$ and $j$, the equality $e^{-\beta [x_0, \ldots, x_q]} = (-\beta)^q e^{-\beta x}$ when $x_0 = \ldots x_q = x$ (see {\cref{sec:div-diff}}), and the fact $X_2 D_1 X_2 = D_1$ by direct computation or on general Hermiticity grounds (see {\cref{thrm:hermiticity}}). Evidently, configurations for which $\avg{z | D_1 | z} = 0$---i.e., when $\ket{z} = \ket{10}$ or $\ket{z} = \ket{11}$---do not contribute to this sum. Comparing this the associated PMR-QMC weight (see {\cref{eq:gbw}}) and $X_2$ estimator (see {\cref{eq:al-pl-estimator}}), we find
\begin{align}
    w_{(z, q)} &=  (-\beta)^q (\avg{z | D_1 | z})^q \\
    (X_2)_{(z, q)} &= - \delta^{(q)}_{X_2} \frac{1}{\avg{z | D_1 | z}} \frac{1}{\beta} \estimates \avg{X_2},
\end{align}
which also satisfies $w_{(z, q)} = 0$ exactly when $\ket{z}$ is $\ket{10}$ or $\ket{11}$. Thus, the direct off-diagonal series expansion and modified PMR-QMC weighted sum expansions agree term-by-term. As an additional clarifying note, $w_{(z, q)} = 0$ whenever $(X_2)_{(z, q)}$ diverges as expected, but this alone is not the problem---rather it simply suggests there could be an issue.

\subsection{Estimating arbitrary operators}\label{Xsec26-6.7}\label{Xsec26}
\label{subsec:estimating-arbitrary-operators}

We now build upon the intuition developed from the above examples to show that it is possible to write arbitrary operators (not just Hermitian observables) into canonical form, and hence, to estimate arbitrary operators with the unbiased estimator in {\eqref{eq:A-canonical-estimator}}. We then consider special cases we observed from our examples. After that, we provide a group-theoretic characterization of which matrix elements of a general observable contribute to its thermal expectation value. Finally, we briefly discuss the complexity of finding a canonical form in general.


\subsubsection{General argument}\label{Xsec27-6.7.1}\label{Xsec27}

We provide a general argument that any operator's thermal expectation value can be written in canonical form. The intuition for this argument is simply that canonical operators naturally appear as terms in the Taylor expansion of $e^{-\beta H}$ in PMR form, and indeed, the PMR-QMC partition function itself is a simplification of these canonical terms. As such, we can write non-trivial matrix elements $\ketbra{x}{y}$ into canonical form and hence estimate any operator.

More formally, any operator can be written $A = \sum_{x,y} A_{x,y} \ketbra{x}{y}$ for $\ket{x}, \ket{y}$ vectors in the $\{\ket{z}\}$ basis, so we can certainly estimate $\avg{A}$ provided we can estimate $\avg{ A_{x,y} \ketbra{x}{y} } = A_{x, y} \avg{y | e^{-\beta H} | x}$ for each $x, y$. By Hermiticity of $H$, $\avg{y | e^{-\beta H} | x} = \avg{ x | e^{-\beta H} | y}$, which conveniently makes questions of estimating $A_{x,y}$ related directly to $(e^{-\beta H})_{x,y}$. Since $\avg{ \ketbra{x}{y} } = 0$ if and only if $(e^{-\beta H})_{x,y} = 0$, we can restrict attention only to estimating $\avg{ \ketbra{x}{y} } \neq 0$.

By Taylor expanding, it follows that there exists some integer $n$ for which $\avg{ x | H^n | y} = \avg{ x | (\sum_j D_j P_j)^n | y} \neq 0$, which in turn means there is at least one product of $D P$'s for which $\avg{ x | \prod_{l=1}^n (D_{f(l)} P_{f(l)}) | y} \neq 0$ ($f$ maps each $D_{f(l)} P_{f(l)}$ to a non-trivial $D_j P_j$ in $H$). This is the essential observation, as we can now write the matrix element in canonical form~\footnote{We are indebted to a dedicated and astute anonymous reviewer for suggesting this argument.},
\begin{align}
 \ketbra{x}{y}& = \Lambda_y (D_{f(l)} P_{f(l)} \cdot \ldots \cdot D_{f(1)} P_{f(1)}) \Lambda_x, \nonumber\\ &\Lambda_y \equiv \ketbra{y}{y}, \ \ \ \Lambda_x \equiv \frac{1}{\prod_{l=1}^n D_{f(l)}(x_l)} \ketbra{x}{x},
\label{Xeqn69-87}
\end{align}
as one may verify by an off-diagonal expansion argument, i.e., successively evaluating $D_{f(l)} P_{f(l)} \ket{x_{l-1}}$ for $\ket{x_{l-1}} = \prod_{k=1}^{l-1} P_{f(k)} \ket{x}$ and $\ket{x_0} = \ket{x}$ as usual. As a remark, the diagonal product in the denominator is not zero or else the matrix element itself would be.

The above shows that we can correctly estimate expectation values of arbitrary operators. Specifically, we can put all matrix elements that contribute non-trivially to the thermal expectation value into canonical form. Explicitly, we have proven the existence of a mapping
\begin{equation}
 \label{eq:canonical-mapping}
 \mathcal{M}(\ketbra{x}{y}) =
 \begin{cases}
 0 & \Tr[\ketbra{x}{y} e^{-\beta H}] = 0 \\
 \Lambda_y \prod_{l=1}^{n_{x,y}} D_{f_{x,y}(l)} P_{f_{x,y}(l)} \Lambda_x & \text{else}
 \end{cases},
\end{equation}
where $n_{x,y}$ and $f_{x,y}$ reflect that the specific product of $D P$ terms depends on $x, y$. In reality, this is a family of mappings since there are infinitely many valid mappings (i.e., for any valid product, one can iterate it twice, thrice, and so on), but for simplicity, we refer to it as a singular mapping $\mathcal{M}$ for some specific choice of $D P$s. We then extend this mapping linearly,
\begin{equation}
 \mathcal{M}(A) = \sum_{x, y} A_{x, y} \mathcal{M}(\ketbra{x}{y})
\label{Xeqn71-89}
\end{equation}
for which we can now precisely write
\begin{equation}
 \avg{A} = \avg{\mathcal{M}(A)},
\label{Xeqn72-90}
\end{equation}
and every matrix element of $\mathcal{M}(A)$ has an unbiased, canonical estimator by construction.

\subsubsection{Revisiting our examples as special cases of general construction}\label{Xsec28-6.7.2}\label{Xsec28}
We reconsider some aspects of our examples that can be viewed as special cases of the general construction above. First, some operators can be written into a simple canonical form wherein a single product matches all matrix elements simultaneously, i.e., $Y_1 = (-\frac{i}{h} Z_1h \mathds{1} X_1) $ for the TFIM. Others, require a sum of terms as in {\eqref{eq:canonical-x2-example}} for which we observe
\begin{equation}
 X_2 = \begin{pmatrix}
0 & 1 & 0 & 0 \\
1 & 0 & 0 & 0 \\
0 & 0 & 0 & 1 \\
0 & 0 & 1 & 0
\end{pmatrix}, \ \ \ \frac12 D_1 P_1 = \begin{pmatrix}
0 & 1 & 0 & 0 \\
1 & 0 & 0 & 0 \\
0 & 0 & 0 & 0 \\
0 & 0 & 0 & 0
\end{pmatrix}, \ \ \  \frac12 (D_2 P_2)(D_1 P_1)(D_2 P_2) = \begin{pmatrix}
0 & 0 & 0 & 0 \\
0 & 0 & 0 & 0 \\
0 & 0 & 0 & 1 \\
0 & 0 & 1 & 0
\end{pmatrix},
\label{Xeqn73-91}
\end{equation}
is a matching of the upper left and upper right matrix elements of $X_2$, respectively. In the worst case, one might need to match every non-zero matrix element of an observable with its own canonical form.

Another aspect of our construction is that we do not need to (and actually cannot) write a canonical form for matrix elements for which $\Tr[ \ketbra{x}{y} e^{-\beta H}] = 0$. We observed this in the XX example (\cref{Xsec24-6.6.3}), for which it was not possible to write $X_1$ into canonical form, but we did not need to because $\avg{X_1} = 0$. Of note, one can directly verify that all non-zero entries of $X_1$ land on zero entries of $X_1 X_2 + X_2 X_3$ and also of $X_1 X_3$. We dedicate further general analysis of this special case in the next section.

\subsubsection{Classification of permutations in estimation}\label{Xsec29-6.7.3}\label{Xsec29}
It is instructive to ``classify permutations in estimation'', and we will see that doing so determines in a principled way which matrix elements satisfy $\Tr[\ketbra{x}{y} e^{-\beta H}] = 0$. Let $G$ be a fixed PMR basis and write
\begin{equation}
 H = \sum_{P \in S_H} D_P P, \ \ \ A = \sum_{P \in S_A} \wtD_P P
\label{Xeqn74-92}
\end{equation}
for $S_H, S_A$ the ``PMR support'' of $H$ and an arbitrary operator $A$, i.e., $S_H = \{P \in G: D_P \neq 0\}$ and $S_A = \{P \in G: \wtD_P \neq 0\}.$ Estimation of $\avg{A}$ is possible by linearity through estimation of each $\avg{\wtD_P P}$.

Consider a specific $\wtP \in S_A$. If $\wtP \in S_H$ as well, then $\wtD_{\wtP} \wtP$ generally has non-zero expectation value, and it can be estimated accurately provided it is cast into canonical form. An explicit non-trivial example of this is provided in {\eqref{eq:canonical-x2-example}}. If instead $\wtP \notin S_H$, two interesting possibilities arise. Let $\avg{S_H}_g$ denote the group generated by $S_H$ by repeated multiplications. In the first case, $\wtP \in \avg{S_H}_g$, and we can again find a canonical estimator that is generally non-zero (i.e., as in the TFIM where $X_1 X_2 \in \avg{X_1, X_2}_g$ and $X_1 X_2 = (\frac{1}{h} \mathds{1})(h \mathds{1} X_1) (\frac{1}{h} \mathds{1})(h \mathds{1} X_2)$ is canonical).

In the second case, $\wtP \notin \avg{S_H}_g$, and we find that $\avg{\wtD_{\wtP} \wtP} = 0$ (i.e., as in the XX chain where $\avg{X_1} = 0$). This can be shown readily by an off-diagonal series expansion,
\begin{equation}
 \Tr[\wtD_{\wtP} \wtP e^{-\beta H}] = \sum_{z, \Siq} \wtD_{\wtP}(z) w_{(z, \Siq)} \avg{ z | \wtP \Siq | z}.
\label{Xeqn75-93}
\end{equation}
As usual, this only has possibly non-zero contributions when $\wtP \Siq = \mathds{1}$ by the PMR group structure, but in this case, $\Siq \in \avg{S_H}_g$ can never equal to $\wtP^T.$ Hence, this matrix element will always be zero, and hence, $\avg{\wtD_{\wtP} \wtP} = 0$ for any $\wtD_{\wtP}$. Thus, it is sufficient to estimate $A$ via only those $\wtP \in S_A \cap \avg{S_H}_g$, i.e.,
\begin{equation}
 \label{eq:projected-equivalence}
 \avg{A} = \avg{A_H}, \ \ \ A_H \equiv \sum_{P \in S_A \cap \avg{S_H}_g} \wtD_P P,
\end{equation}
where $A_H$ can be viewed as a type of ``projected operator''. Specifically, it is restriction of $A$ onto the support of PMR permutations contained in $\avg{S_H}_g,$ which are the only terms which can have nonzero thermal average.

Importantly, we can also directly relate the $\ketbra{x}{y}$ with trivial thermal expectation value, $\Tr[\ketbra{x}{y} e^{-\beta H}] = 0,$ to $\avg{S_H}_g$ due the PMR assumptions on $G$. Specifically, $\Tr[\ketbra{x}{y} P] = \avg{y | P | x}$, and by the regular natural action of $G$, the $P^{(x\rightarrow y)} \in G$ such that $P^{(x\rightarrow y)} \ket{x} = \ket{y}$ is unique. Consequently,
\begin{equation}
 \label{eq:non-trivial-matrix-elements-new}
 P^{(x\rightarrow y)} \in G - \avg{S_H}_g \implies \Tr[\ketbra{x}{y} e^{-\beta H}] = 0,
\end{equation}
since any $\Siq$ that appears in the off-diagonal expansion of $e^{-\beta H}$ necessarily satisifes $\Siq \in \avg{S_H}_g$. Of course, this observation is consistent with the permutations we can exclude in thermal estimation, {\eqref{eq:projected-equivalence}}. The left-hand-side of this implication can equivalently be replaced with $\avg{y | P | x} = 0 \ \forall P \in \avg{S_H}_g$, which, albeit enumerative, provides a group-theoretic means to find all matrix elements $\ketbra{x}{y}$ with trivial thermal expectation value that avoids computing $\avg{y | e^{-\beta H} | x}$ directly.

\subsubsection{A few comments on an algorithm to find a canonical estimator for an observable}\label{Xsec30-6.7.4}\label{Xsec30}

Given an arbitrary operator $A$, we have shown that it is possible to construct an unbiased PMR-QMC estimator for Hamiltonian $H$ from the following steps. (1) Construct the group of permutations generated by the non-trivial permutations contained in $H$, $\avg{S_H}_g$. (2) Form the truncated operator $A_H$ whose PMR form contains only permutations in $\avg{S_H}_g,$ which corresponds to including only matrix elements that contribute to the thermal expectation value, i.e., $\avg{A} = \avg{A_H}$. (3) Construct a canonical representation of each non-zero matrix element of $A_H$, say $\ketbra{x}{y}$, by finding a string of $D P$'s for which $\avg{x | \prod_{l=1} D_{f(l)} P_{f(l)} |y} \neq 0$ by enumeration.

A couple of remarks on the practicality of this algorithm are in order. Firstly, even if one succeeds, the resulting estimator is not generally practical. To see this, let $d$ be the dimension of the system (so each $P$ in a $d\times d$ matrix) and suppose $A_H$ contains $k$ permutations. In this case, $A_H$ contains $k d$ total non-zero matrix elements, so a naive unbiased estimator without simplifications results in separately estimating $k d$ estimators and combining the result. For quantum systems, $d$ scales exponentially with the number of particles (i.e. $d = 2^n$ for $n$ qubits), so this is not practical. Of course, things can be much easier in specific practical examples, as for the TFIM, XX, and {\eqref{eq:canonical-x2-example}} cases where practical observables only required summing up a few canonical terms.

Secondly, this general algorithm to construct canonical estimators is also not practical. For simplicity, suppose we wish to construct a canonical estimator for the singular non-trivial matrix element $\ketbra{x}{y}$. A sense of complexity can be estimated from the smallest $q$ such that some $\Siq \ket{y} = \ket{x}$ since any $\Sip$ for $p < q$ will thus give $\avg{x | \Sip | y} = 0$. Yet, in the worst case, $q$ might be $O(d)$ for $d$ the dimension of the system, making enumeration of $\Sip$ up to $\Siq$ impractical. To see this, suppose $d$ is even and $H = P_{\sigma} + P_{\sigma^{-1}}$ for $\sigma = (0\ldots d-1)$ and $\sigma^{-1} = (0\ldots d-1)^{d - 1}.$ Since $\avg{S_H}_g = \avg{(0\ldots d-1)}_g$ forms an entire PMR group, all matrix elements can contribute to thermal averages. In the worst case, we might have an observable with a matrix element contained in $\tau = (0\ldots d)^{d/2}$, which requires at least $d / 2$ permutations to construct, i.e., $P_\sigma^{d/2}$.

Yet for practical model like the TFIM, the largest power needed for practical observables can be much smaller. For example, we can clearly write any $n$-body Pauli string using at most $n$ local $X_i$ permutations. For the TFIM specifically, in fact, constructing canonical estimators for Pauli observables requires only simple Pauli algebra of a single qubit and hence can be done by and hand without an algorithm.

}

{
\subsection{Improved estimator for rare sampling}\label{Xsec31-6.8}\label{Xsec31}
\label{subsec:improved-estimator}

We have shown that the thermal expectation value of any operator can be estimated with an unbiased estimator of the form {\eqref{eq:A-canonical-estimator}}. In general, this estimator can have a rare sampling issue that simpler estimators for diagonal operators or Hamiltonian based operators do not have. To see this, first observe that observables like $H$ have an estimator, {\eqref{eq:H-estimator}}, which is generally non-zero for any given PMR-QMC configuration. In contrast, {\eqref{eq:A-canonical-estimator}} involves a product of $L$ delta-function conditions that becomes increasingly likely to be equal to zero as $L$ increases. For this reason, {\eqref{eq:A-canonical-estimator}} may have a \emph{rare sampling problem}. Using the same logic in the derivation of the PMR-QMC partition function, however, we can derive \emph{improved estimators} that mitigate this problem.

First, we observe the following simplification of a canonical operator
\begin{equation}
 \prod_{s=1}^L D_{\bm{l}_s}P_{\bm{l}_s} = \widetilde{P} \widetilde{D}, \ \ \ \widetilde{P} = \prod_{s=1}^L P_{\bm{l}_s} , \ \ \ \widetilde{D} = \prod_{s=1}^L (P_{\bm{l}_s} \ldots P_{\bm{l}_1})^T D_{\bm{l}_s} (P_{\bm{l}_s} \ldots P_{\bm{l}_1}),
\label{Xeqn82-96}
\end{equation}
from off-diagonal manipulations. As a simple justification, observe that for $\wtP = P_2 P_1$, one can write
\begin{align}
 \wtP \wtD \ket{z} &= (P_2 P_1) (P_2 P_1)^T D_2 (P_2 P_1) (P_1)^T D_1 P_1 \ket{z} \\
 &= D_2 P_2 D_1 P_1 \ket{z},
\end{align}
which generalizes to any length product of $D_l P_l$'s. As a remark, our notation suggests that each $P_{\bm{l}_s} \neq \mathds{1}$. Yet, the above logic still holds even if some intermediate $P$'s are identity (i.e., choose $P_2 = \mathds{1}$ above). More generally, all conclusion in this section can straightforwardly be generalized to the case where intermediate permutations can be identity and each $D_{\bm{l}_s} \rightarrow \Lambda_s D_{\bm{l}_s}.$ For simplicity of presentation, we shall assume each $\Lambda_s = \mathds{1}$ and each $P_{\bm{l}_s} \neq \mathds{1}$.

With these simplifying assumptions, we next observe
\begin{equation}
 \avg{z | \wtD | z} = \avg{z |P_1^T P_2^T D_2 P_2 P_1| z} \avg{z | P_1^T D_1 P_1 | z} = \avg{z_2 | D_2 | z_2} \avg{z_1 | D_1 | z_1},
\label{Xeqn83-99}
\end{equation}
which generalizes to $\avg{z | \wtD | z} = D_{(z, \wtP)}$ for any $L$. Recall $D_{(z, \Siq)}$ denotes the hopping-term in the off-diagonal series expansion of the partition function and partly comprises the PMR-QMC weight. Using the off-diagonal series expansion logic used to derive {\eqref{eq:al-pl-estimator}} and \eqref{eq:al-dl-pl-estimator}, we can write
\begin{align}
 \Tr[\wtP \wtD e^{-\beta H}] &= \sum_{z, \Sip} w_{(z, S_{\bm{i}_{p+L}})} \frac{ e^{-\beta [E_{z_0}, \ldots, E_{z_{p}}]} }{ e^{-\beta [E_{z_0}, \ldots, E_{z_{p+L}}]} } \avg{z |\wtP \Sip | z}.
\end{align}
If we introduce $\delta$ functions to enforce $\wtP \Sip = \Siq$, we again arrive at {\eqref{eq:A-canonical-estimator}} which has a rare sampling problem.

Yet when the PMR group is Abelian---which is always possible and routinely chosen in practical simulations~\cite{barash2024QuantumMonteCarlo,akaturk2024quantum,gupta2020PermutationMatrixRepresentation,babakhani2025quantum}---we can relax the stringent ordering enforced by the $\delta$ functions. Specifically, we can employ the following general estimator,
\begin{equation}
 \label{eq:improved-estimator-for-canonical-ops}
 \left(\prod_{s=1}^L D_{\bm{l}_s}P_{\bm{l}_s} \right)_{\mathcal{C}} = \frac{1}{L!} \prod_{s=1}^L \mathbf{1}(P_{\bm{l}_s} \in \{ P_{\iq}, P_{\bm{i}_{q-1}}, \ldots, P_{\bm{i}_{q-L}} \}) \frac{ e^{-\beta [E_{z_0}, \ldots, E_{z_{q-L}}]} }{ e^{-\beta [E_{z_0}, \ldots, E_{z_{q}}]} },
\end{equation}
where the indicator function is merely formal notation to encode the milder requirement that each $P_{\bm{l}_s}$ must be contained in the final $L$ permutation of $\Siq$. The essential logic of this estimator can again be understood for $L = 2$,
\begin{equation}
 \Tr[\wtP \wtD e^{-\beta H}] = \Tr[\frac{(P_2 P_1 + P_1 P_2)}{2!} \wtD e^{-\beta H}],
\label{Xeqn85-102}
\end{equation}
and the linearity of the trace. More generally, we are taking advantage of the commutative relation,
\begin{equation}
 P_1 P_2 \ldots P_L = \frac{1}{L!} \sum_{\sigma \in S_L} P_{\sigma(1)} \cdot \ldots \cdot P_{\sigma(L)},
\label{Xeqn86-103}
\end{equation}
where the sum is over all permutations of $L$ elements. As a practical remark {\eqref{eq:improved-estimator-for-canonical-ops}} is the general estimator we use in our open source code~\cite{ezzell2025code} to estimate terms such as $X_3 X_9$ in the transverse-field Ising model numerical experiments in {\cref{sec:numerical-demo}}.

}

\section{Computational complexity of estimators and simulation time}\label{Xsec32-7}\label{Xsec32}
\label{subsec:complexity-interlude}

Before transitioning from static operator estimators to their dynamic counterparts, we define a notion of estimator computational complexity relevant to PMR-QMC simulation. We then state and justify the complexity of the various static observable estimators we have derived. { We subsequently discuss factors that contribute to total simulation time. While total simulation time partially depends on observable estimation complexity, it is typically more reliant on model and phase dependent properties, which we discuss.} This section can be skipped on first reading or for those uninterested in computational costs, but henceforth, we will simply state estimator complexities with justification following from the logic expounded here.

{
\subsection{Costs of atomic operations and meaning of estimator complexity}\label{Xsec33-7.1}\label{Xsec33}
The basic details of a PMR-QMC simulation are summarized in {\cref{subsec:PMR-QMC}} and discussed in depth in several prior works~\cite{albash2017OffdiagonalExpansionQuantum,gupta2020PermutationMatrixRepresentation,barash2024QuantumMonteCarlo}. The primary computational cost in PMR-QMC is computing the DDE, $e^{-\beta [E_{z_0}, \ldots, E_{z_q}]}$ given a multiset of diagonal energies $[E_{z_0}, \ldots, E_{z_q}]$. Though the stable computation of the DDE has been an active area of research~\cite{mccurdy1984accurate}, PMR-QMC takes advantage of recent numerical advancements~\cite{gupta2020CalculatingDividedDifferences}. Specifically, Ref.~\cite{gupta2020CalculatingDividedDifferences} provides a stable $O(q)$ algorithm to compute the DDE with the addition (or removal) of an input, i.e., $e^{-\beta [E_{z_0}, \ldots, E_{z_q}, E_{z_{q+1}}]}$ that works for a wide range of energies and inverse temperatures $\beta$.}

{ Of course, to compute the DDE at all, we must have access to $E_{z_j}$ values. In real simulations, we pre-compute all diagonal matrix elements $\avg{z | D_P | z}$ for every $P$ for which $D_P \neq 0$ and store them in a dictionary. Thus, we can query any diagonal energy, $E_{z_j} = \avg{z_j | D_{\mathds{1}} | z_j}$ in $O(1)$ time. In total, this means computing the PMR-QMC weight $W_{(z, \Siq)} = \prod_{k=1}^q \avg{z_k | D_{\bm{i}_k} | z_k} e^{-\beta [E_{z_0}, \ldots, E_{z_q}]}$ requires $O(q^2)$ effort if computed from scratch. However, many PMR-QMC moves only update the relevant multiset by the removal of one diagonal energy and the insertion of another (see, e.g., Sec. IV.D QMC Updates in Ref.}~\cite{barash2024QuantumMonteCarlo}).{ As such, the only non-trivial cost in a typical PMR-QMC update is an $O(q)$ DDE deletion/insertion cost. }

{ To evaluate observables, we compute the relevant estimator $O_{(z, \Siq)}$. This is done after a PMR-QMC update has already completed, and hence, we have already computed $W_{(z, \Siq)}$. In fact, for more efficient manipulations, we store all partial computations $D_{(z, \Sip)} = \prod_{k=1}^p \avg{z_k | D_{\bm{i}_k} | z_k}$ and $e^{-\beta [E_{z_0}, \ldots, E_{z_p}]}$ for all $p = 0, \ldots, q$ during simulation and hence have $O(1)$ access to each of these values for the purposes of estimation. Hence, we say an estimator has $O(q^k)$ complexity, for example, if computing $O_{(z, \Siq)}$ requires $O(q^k)$ \emph{additional effort} using these stored quantities. As we shall discuss, this additional effort can result from repeated $O(1)$ estimations or the cost of adding or removing an input from the PMR-QMC weight DDE.}

\vspace*{-1em}

\subsection{ Estimator complexities of static observables }\label{Xsec34-7.2}\label{Xsec34}

With the complexity of atomic PMR-QMC operations now clarified, we can discuss the complexity of various static observables. The results are summarized in {\cref{tab:static-complexity}}, and for most static observables, the complexity is $O(1)$. For clarity, we discuss them in order of appearance. For $\avg{\Lambda^k}$ and $\avg{\Hdiag^k}$, one need only use $O(1)$ time to query $E_{z}$ (or $\Lambda(z)$) and exponentiate it, yielding $O(1)$ complexity. For $\avg{H}$, one simply combines $E_{z_q}$ with a ratio DDEs for which we have $O(1)$ access by assumption, also yielding $O(1)$. For $\avg{H^2}$ estimated by {\eqref{eq:H2-estimator}}, we perform a small, constant number of arithmetic operations on $O(1)$ access values, again giving $O(1)$. For $\avg{H^k}$, however, we must combine $O(k)$ such $O(1)$ terms where $k$ is not constant and potentially large, resulting in $O(k)$ complexity. Since $\avg{\Hdiag^k}$ and $\avg{H}$ can be estimated in $O(1)$ time, $\avg{\Hoffdiag}$ and $\avg{\Hoffdiag^2}$ are also estimated in $O(1)$ time by {\eqref{eq:hoffdiag-estimator}} and \eqref{eq:hoffdiag2-estimator}. This covers all the Hamiltonian-based static PMR-QMC observables.

For general canonical observables, the arguments proceed similarly. For example, $\avg{D_l P_l}$ as estimated by {\eqref{eq:dl-pl-estimator-delta}} is $O(1)$ since it only involves a ratio of $O(1)$ DDEs. For similar reasons, $\avg{\Lambda_l D_l P_l}$ estimated by {\eqref{eq:al-dl-pl-estimator}} is also $O(1)$. Hence, a sum of $K$ such terms, $\avg{\sum_{l=0}^{K-1} \Lambda_l D_l P_l}$ is $O(K)$. Finally, $\avg{\Lambda_1 D_{\bm{l}_1} P_{\bm{l}_1} \cdot \ldots \cdot \Lambda_l D_{\bm{l}_L} P_{\bm{l}_L} \Lambda_{L+1}}$ estimated by {\eqref{eq:A-canonical-estimator}} requires $L$ permutation equality checks (each delta function). Otherwise, there are $O(L)$ multiplications of quantities accessible in $O(1)$ time, giving a total estimation time of $O(L)$.

\subsection{ Interpreting $q$ in big $O$ expressions}\label{Xsec35-7.3}\label{Xsec35}

{ In our discussion of PMR-QMC move costs and in the estimator complexities for dynamic observables (see {\cref{tab:dynamic-complexity}}), we encounter big O estimates of the form $O(q^k)$ for some integer $k$. As discussed in {\cref{sec:pmr-qmc}} and in prior works, especially Ref.~\cite{albash2017OffdiagonalExpansionQuantum}, $q$ is the imaginary-time dimension similar to the operator length $M$ encountered in SSE. Unlike SSE, however, $q$ can dynamically update during PMR-QMC simulations depending on how large and ``how quantum'' the system is. Writing an $N$ qubit Hamiltonian as $H = \Hdiag + \Gamma \Hoffdiag$, prior in-depth empirical studies~\cite{albash2017OffdiagonalExpansionQuantum,ezzell2025universal} found that the average value of the quantum dimension during simulation $\mathds{E}[q]$ scales as $O(\Gamma^2 \beta N)$, and this is corroborated in our numerical results here ({\cref{fig:timing-estimates}}). Hence, the expression $O(q^k)$ can be interpreted as meaning, on average, an estimate of $\sim (\Gamma^2 \beta N)^k$ atomic computations are required.}

\subsection{ Total simulation time (TST)}\label{Xsec36-7.4}\label{Xsec36}

{
We have discussed the complexity of PMR-QMC updates and computing observable estimators as well as how to interpret $q$ in expressions like $O(q)$. Though all these factors contribute to total simulation time (TST)---that is, the time it takes to reliably estimate thermal expectation values with reasonable error bars---they do not tell the complete story. Like all QMC methods, TST is highly model and phase dependent in PMR-QMC, and can be hindered by problems such as classical frustration~\cite{kandel1990cluster, zhang1994cluster, coddington1994generalized, houdayer2001cluster, rakala2017cluster}, the sign problem~\cite{hen2019resolution, gupta2020elucidating, hen2021determining}, and subsequent high autocorrelation times. These problems can arise even without costly observable estimation.

We expect the TST to be determined by the complexity of QMC updates rather than that of the estimators for a wide range of temperatures~\cite{barash2024QuantumMonteCarlo,ezzell2025universal}. Since it is highly unlikely that a simple, a priori condition to ensure fast convergence exists~\cite{troyer2005ComputationalComplexityFundamental}, it is generally not possible to write down a meaningful estimate of TST for all models. Nevertheless, different facets of TST have been well studied in PMR-QMC both theoretically~\cite{hen2019resolution, hen2021determining} and empirically~\cite{albash2017OffdiagonalExpansionQuantum,gupta2020PermutationMatrixRepresentation,gupta2020elucidating,barash2024QuantumMonteCarlo,ezzell2025universal} in prior work.

\subsubsection{Comparisons to path integral QMC and stochastic series expansion QMC}\label{Xsec37-7.4.1}\label{Xsec37}
Perhaps most importantly, it is known that TST can be much shorter for PMR-QMC than TST in conventional methods like path integral QMC (PIQMC) and stochastic series expansion (SSE). Specifically, Refs.~\cite{albash2017OffdiagonalExpansionQuantum,gupta2020PermutationMatrixRepresentation} studied a random ensemble of 3-regular MAX2SAT instances encoded into an Ising model with a transverse field using PMR-QMC. For low temperatures (high $\beta)$ or near first-order phase transitions in this ensemble, it is known that QMC can suffer from exponentially long autocorrelation times (in $N$), causing slow convergence and high TST. Nevertheless, PMR-QMC was found to converge relatively quickly even for moderate $N$ and large $\beta$, especially compared to PIQMC and SSE. Figure 4 of Ref.~\cite{gupta2020PermutationMatrixRepresentation}, for example, shows PMR-QMC converges to thermal expectation values over four orders of magnitude faster than SSE---a difference of taking seconds to taking over a day of compute time.

\subsubsection{The sign problem }\label{Xsec38-7.4.2}\label{Xsec38}
Next, we discuss the sign problem, a known obstacle to fast convergence of QMC simulations~\cite{gupta2020elucidating,hen2021determining}. Put simply, the sign problem arises whenever we cannot decompose the partition function as a sum of non-negative weights, i.e., when $\text{sgn}(W_{\mathcal{C}}) = -1$ for at least one configuration $\mathcal{C}$ with nonzero weight. When this happens, we can no longer interpret $W_{\mathcal{C}} / \sum_\mathcal{C} W_{\mathcal{C}}$ as a valid probability distribution with which to perform importance sampling. Nevertheless, we can define a valid Markov chain over probabilities $|W_{\mathcal{C}}| / \sum_\mathcal{C} |W_{\mathcal{C}}|$ at the cost of noisier observable estimates that converge more slowly (see Eq. (5) in Ref in Ref.~\cite{gupta2020elucidating} or Eq. (41) in Ref.~\cite{barash2024QuantumMonteCarlo}). The empirical severity of the sign problem can thus be tracked by computing $\avg{\text{sgn}} = W_\mathcal{C} / \sum_{\mathcal{C}} | W_{\mathcal{C}}|$, and this is automatically done in our code~\cite{ezzell2025code}.

Empirical facets of the sign problem in PMR-QMC simulation have been explored for minimal qutrit models~\cite{gupta2020elucidating}, XY models~\cite{barash2024QuantumMonteCarlo}, and random $k$-local Hamiltonians~\cite{barash2024QuantumMonteCarlo}. As expected, $\avg{q}$, time per MC update, and hence TST all rise as $\avg{\text{sgn}}$ approaches zero and the sign problem becomes more severe (see Tab. I in \citet{barash2024QuantumMonteCarlo}). Furthermore, the sign problem itself can become more severe at lower temperatures or with more terms in a Hamiltonian (i.e. see Figs. 2 and 5 in Ref.~\cite{barash2024QuantumMonteCarlo}). Theoretically, Ref.~\cite{hen2021determining} provides a necessary and sufficient condition for a PMR-QMC simulation to have a sign problem in a fixed basis. We leave a detailed discussion of this condition to the relevant prior work~\cite{hen2021determining}, but we remark that this is a stronger and more predictive relation than the more well-known \emph{stoquasticity}.}

\subsubsection{ Autocorrelation time}\label{Xsec39-7.4.3}\label{Xsec40-8}

{
Autocorrelation time characterizes the number of Monte Carlo steps over which measurements remain correlated~\cite{landau2015GuideMonteCarlo}. Hence, it determines the number of effective independent samples in the estimation of statistical errors. As with prior QMC approaches, PMR-QMC utilizes binning analysis (see Appendix B of Ref.~\cite{barash2024QuantumMonteCarlo}) to estimate statistical errors. In the spirit of the substantial automation in PMR-QMC, our modern codes~\cite{barash2024PmrQmcCode,ezzell2025code,barash2025highspincode} automatically perform autocorrelation diagnostics to ensure measurement blocks are effectively uncorrelated in binning analysis. Our diagnostic checks that statistical errors from independent runs (i.e., performed in parallel simulations) are comparable to those estimated in any individual run. When our tests pass, they suggest that autocorrelation time is shorter than the block length.

Correspondingly, we do not directly estimate autocorrelation times to adjust simulation parameters. Instead, we directly test whether samples are too correlated using our diagnostic, and if so, we adjust simulation parameters until they are not. As such, our simulation parameters provide an indirect, upper-bound estimate of autocorrelation time. Detailed simulation parameters and simulation times have been provided in prior works~\cite{barash2024QuantumMonteCarlo,ezzell2025universal,babakhani2025quantum}, and we provide additional details for this work in {\cref{sec:numerical-demo}}.
}

\section{Estimation of dynamic observables}\label{Xsec40}\label{Xsec42-8.2}
\label{sec:dynamic-operator-estimators}

We extend our estimator discussion to various dynamic observables defined in {\cref{sec:overview-of-estimators}}. Each quantity to be redefined in the relevant section is defined in terms of the imaginary-time evolved operator,
\begin{equation}
 \label{eq:imag-time-op-again}
 O(\tau) \equiv e^{\tau H} O e^{-\tau H}.
\end{equation}
A direct result of our derivations is that if one can estimate the static quantities $\avg{A}$ and $\avg{B}$, then one can also estimate the dynamic quantities involving $A$ and $B$ we consider below (e.g., see our numerical results in {\cref{sec:numerical-demo}}). For simplicity, then, we perform our derivations with the assumption $A$ and $B$ are in the simple canonical form,
\begin{align}
 \label{eq:simple-A}
 A &\equiv \wtA_k D_k P_k, \\
 \label{eq:simple-B}
 B &\equiv \wtB_l D_l P_l,
\end{align}
unless stated otherwise. This allows us to avoid repeating the various subtleties we described in {\cref{subsec:estimating-arbitrary-operators}}, but handling these subtleties shows that, in principle, one can estimate these dynamic quantities for arbitrary operators. A summary of the dynamic observable estimators we derive in this work alongside their PMR-QMC complexity is given in {\cref{tab:dynamic-complexity}}. As an important remark, our integrated susceptibilities do not rely on numerical integration, as we discuss in {\cref{subsec:estimating-esus}} and \cref{subsec:estimating-fsus}.

\begin{table}
\caption{A summary of dynamic observable estimators we derive in this work and their computational complexity in terms of the PMR-QMC off-diagonal expansion order, $q$. Throughout $A = \wtA_k D_k P_k$ and $B = \wtB_l D_l P_l$ are assumed to be in canonical form. Namely, $\wtA_k$ is diagonal and $D_k P_k \in H$. Furthermore, $\wtA$ and $\wtB$ are purely diagonal.}{%
\begin{tabular}{lll}
\toprule
{Dynamic observable} & {Estimator} & {Estimator complexity} \\
\colrule
$\avg{A(\tau) B}$ & {\eqref{eq:simple-AtauB-correlator-estimator}} & $O(q^2)$ \\
$\int_0^\beta \avg{A(\tau) B} \dtau $ & {\eqref{eq:esimate-es-int-ab}} & $O(q)$ \\
$\int_0^\beta \avg{A(\tau) \wtB} \dtau $ & {\eqref{eq:esimate-es-int-abdiag}} & $O(q^2)$ \\
$\int_0^\beta \avg{A(\tau) \Hdiag} \dtau $ & {\eqref{eq:esimate-es-int-ahdiag}} & $O(1)$ (see Ref.~\cite{ezzell2025universal}) \\
$\int_0^{\beta/2} \tau \avg{A(\tau) B} \dtau $ & {\eqref{eq:esimate-fs-int-ab}} & $O(q^4)$ \\
$\int_0^{\beta/2} \tau \avg{A(\tau) \Hdiag} \dtau $ & {\eqref{eq:esimate-fs-int-ahdiag}} & $O(q^3)$ (see Ref.~\cite{ezzell2025universal}) \\
\botrule
\end{tabular}}
\label{tab:dynamic-complexity}
\end{table}

\subsection{Imaginary time correlators}\label{Xsec41-8.1}\label{Xsec43-8.3}
\label{subsec:estimating-atau-b}
The imaginary time correlator is given by
\begin{equation}
 \avg{A(\tau) B} = \avg{e^{\tau H} A e^{-\tau H} B}.
\label{Xeqn88-107}
\end{equation}
To make our derivation easier to follow, we first proceed with $A$ and $B$ purely diagonal, which we denote $\wtA$ and $\wtB$. We then generalize our results for $A$ and $B$ given by {\eqref{eq:simple-A}} and \eqref{eq:simple-B}.

By the cyclicity of trace, we can write
\begin{equation}
 \Tr[\wtA(\tau) \wtB e^{-\beta H}] = \sum_{z} \avg{z | \wtA e^{-\tau H} \wtB e^{-(\beta - \tau)H} | z},
\label{Xeqn89-108}
\end{equation}
which does not neatly resemble any of our static estimator off-diagonal expansions. Nevertheless, we can make progress by realizing that the off-diagonal expansion introduced in {\cref{subsec:ODE}} is really an expansion of $f(H) \ket{z}$ for analytic functions $f$---not specifically of $\avg{z | e^{-\beta H} | z}$. That is, we can write the expansion,
\begin{equation}
 \label{eq:tau-minus-beta-expansion}
 e^{-(\beta-\tau)H}\ket{z} = \sum_{\Sip} D_{(z, \Sip)} e^{-(\beta-\tau) [E_{z_0}, \ldots, E_{z_p}]} \Sip \ket{z}.
\end{equation}
Proceeding from right to left, then
\begin{equation}
 \label{eq:B-tau-minus-beta-expansion}
 {\wtB} e^{-(\beta-\tau)H}\ket{z} = \sum_{\Sip} {\wtB(z_{p})} D_{(z, \Sip)} e^{-(\beta-\tau) [E_{z_0}, \ldots, E_{z_p}]} \Sip \ket{z},
\end{equation}
for $\ket{z_p} \equiv \Sip \ket{z}$ as usual. We can now perform another expansion of $e^{-\tau H} \ket{z_{p}}$ to find,
\begin{equation}
 {e^{-\tau H}} \wtB e^{-(\beta-\tau)H}\ket{z} = {\sum_{\Sir}}\sum_{\Sip} {D_{(z_{p}, \Sir)}} {\wtB(z_{p})} D_{(z, \Sip)} {e^{-\tau [E_{z_{p}}, \ldots, E_{z_{p+r}}]} } e^{-(\beta-\tau) [E_{z_0}, \ldots, E_{z_p}]} {\Sir} \Sip \ket{z}.
\label{Xeqn92-111}
\end{equation}
Applying the final $\wtA$ operator to $\ket{z_{r+p}} = \Sir \Sip \ket{z}$, we find
\begin{multline}
 {\wtA} {e^{-\tau H}} \wtB e^{-(\beta-\tau)H}\ket{z} = \sum_{\Sir}\sum_{\Sip} {\wtA(z_{p+r})} {\wtB(z_{p})} {D_{(z_{p}, \Sir)}} D_{(z, \Sip)} \\
 \times e^{-\tau [E_{z_{p}}, \ldots, E_{z_{p+r}}]} e^{-(\beta-\tau) [E_{z_0}, \ldots, E_{z_p}]} {\Sir} \Sip \ket{z},
\end{multline}
which can be expressed in a PMR-QMC estimator form before summing over basis states to compute the trace.

First, we can define $q \equiv p + r$ and then substitute,
\begin{equation}
 \Sir \Sip \rightarrow \Siq,
\label{Xeqn93-113}
\end{equation}
Correspondingly, we then observe
\begin{equation}
 D_{(z_{p}, \Sir)} D_{(z, \Sip)} = D_{(z, \Siq)}.
\label{Xeqn94-114}
\end{equation}
So far, this coaxing to an estimator has proceeded in the same way as in the static case (i.e., see {\cref{subsec:estimating-dl-pl}}). The main difference is that we must now handle the double sum, which after careful inspection can be written,
\begin{equation}
 \sum_{\Sir} \sum_{\Sip} \ldots = \sum_{\Siq} \sum_{p = 0}^{q} \ldots.
\label{Xeqn95-115}
\end{equation}
The logic is actually pretty simple. Given a permutation string $\Siq$ composed or $p$ permutations from $\Sip$ and $r$ from $\Sir$, then the length of $p$ can range from $0$ to $q$ and $r$ is fixed at value $r = q - p$. Since the sum notation over $\Sip$ and $\Sir$ means we are really summing over all permutations of all possible strength lengths, this gives us the claimed equality. More precisely, we can write
\begin{equation}
 \label{eq:leibniz-works-noncommuting}
 \avg{z | \wtA {e^{-\tau H}} \wtB e^{-(\beta-\tau)H} | z} = \sum_{\Siq}\sum_{p=0}^q {\wtA(z_{q})} {\wtB(z_{p})} D_{(z, \Siq)} e^{-\tau [E_{z_{p}}, \ldots, E_{z_{q}}]} e^{-(\beta-\tau) [E_{z_0}, \ldots, E_{z_p}]} \avg{z | \Siq | z},
\end{equation}
which as usual is only nonzero when $\Siq = \mathds{1}$. Finally, multiplying and diving by $e^{-\beta [E_{z_0}, \ldots, E_{z_q}]}$ gives,
\begin{equation}
 \avg{z | \wtA {e^{-\tau H}} \wtB e^{-(\beta-\tau)H} | z} = \sum_{\Siq} w_{(z, \Siq)} \left( {\wtA(z)} \sum_{p=0}^q {\wtB(z_{p})} \frac{ e^{-\tau [E_{z_{p}}, \ldots, E_{z_{q}}]} e^{-(\beta-\tau) [E_{z_0}, \ldots, E_{z_p}]}}{e^{-\beta [E_{z_0}, \ldots, E_{z_q}]}} \right),
\label{Xeqn97-117}
\end{equation}
where we employed the $\Siq = \mathds{1}$ periodicity, $\ket{z_q} = \ket{z}$, to pull the $\wtA(z)$ out of the inner sum. Simply summing over all possible basis states $\ket{z}$ to compute the trace reveals,
\begin{equation}
 \label{eq:diagonal-correlator-estimator}
 \left( {\wtA(z)} \sum_{p=0}^q {\wtB(z_{p})} \frac{ e^{-\tau [E_{z_{p}}, \ldots, E_{z_{q}}]} e^{-(\beta-\tau) [E_{z_0}, \ldots, E_{z_p}]}}{e^{-\beta [E_{z_0}, \ldots, E_{z_q}]}} \right) \estimates \avg{\wtA(\tau) \wtB}.
\end{equation}

The derivation of this estimator is completely rigorous but a bit tedious. By inspection, however, its structure resembles the Leibniz rule, {\eqref{eq:leibniz-rule}}, rather closely. Recall that we actually used the Leibniz rule to derive a formal estimator for $\avg{z | g(H) e^{-\beta H} | z}$ for any analytic function $g$ in {\cref{subsec:estimating-f(H)-things}}. For $g(H) e^{-\beta H}$, there is nothing particularly subtle about using the Leibniz rule, as one can view $f(H) = g(H) e^{-\beta H}$ as a singular function of $H$. Put differently, $[g(H), e^{-\beta H}] = 0$, so properties such as invariance to permutations and so on that are expected of the Leibniz rule carry over straightforwardly to the commuting case. Here, things do not commute, as $[e^{-\tau H}, \wtB] \neq 0$, for example. Nevertheless, {\eqref{eq:leibniz-works-noncommuting}} shows that, at least practically, carrying out the Leibniz rule logic without worrying about rigor in this non-commuting case still gives the correct answer in this application.

Proceeding with this line of reasoning, we can quickly derive an estimator for $\avg{A(\tau) B}$ when $A,B$ are given by {\eqref{eq:simple-A}} and \eqref{eq:simple-B},
\begin{equation}
 \label{eq:simple-AtauB-correlator-estimator}
 \left( {\wtA_k(z)} {\delta_{P_k}^{(q)}} \sum_{p=1}^{q-1} {\delta_{P_l}^{(p)}} {\wtB_l(z_{p})} \frac{ e^{-\tau [E_{z_{p}}, \ldots, {E_{z_{q-1}}}]} e^{-(\beta-\tau) [E_{z_0}, \ldots, {E_{z_{p-1}}}]}}{e^{-\beta [E_{z_0}, \ldots, E_{z_q}]}} \right) \estimates \avg{A(\tau) B}.
\end{equation}
The differences between this expression and and the pure diagonal case in {\eqref{eq:diagonal-correlator-estimator}} are the $\delta$ functions and the removal of one argument from both DDEs in the numerator. On the one hand, these differences arise from careful usage of the Leibniz rule heuristic, but we can also understand them by appealing to our prior derivation.

Notably, consider the jump from {\eqref{eq:tau-minus-beta-expansion}}, {\eqref{eq:B-tau-minus-beta-expansion}} but with $\wtB$ replaced with $B = \wtB_l D_l P_l$. The $P_l$ contained in $B$ now becomes the $p{\text{th}}$ permutation in $\Siq$---which explains both $\delta_{P_l}^{(p)}$ and the removal of $E_{z_p}$ from $e^{-(\beta - \tau) [\ldots]}$. At the same time, no division by $D_l(z_p)$ is necessary since $B$ is in canonical form. Similar arguments for $A = \wtA_k D_k P_k$ explain the remaining changes.

We note that both {\eqref{eq:diagonal-correlator-estimator}} and {\eqref{eq:simple-AtauB-correlator-estimator}} are $O(q^2)$ estimators. This follows since building up both $e^{-\tau [\ldots]}$ and $e^{-(\beta - \tau) [\ldots]}$ requires $O(q)$ effort in the worst case. By the sum, we must do this $q$ times, thus $O(q^2)$ represents a conservative, worst case analysis. In practice, we remark that {\eqref{eq:simple-AtauB-correlator-estimator}} involves a simple $O(1)$ containment check with $\delta_{P_k}^{(q)}$. When this fails, one need not evaluate the estimator at all, confirming that $O(q^2)$ is indeed a worst case bound.

\subsection{A generalized energy susceptibility integral}\label{Xsec42}\label{Xsec44-9}
\label{subsec:estimating-esus}
The finite temperature energy susceptibility (ES)~\cite{albuquerque2010QuantumCriticalScaling, ezzell2025universal, schwandt2009QuantumMonteCarlo, wang2015FidelitySusceptibilityMade} is defined in {\eqref{eq:Esus}}. For this reason, we denote the integrated susceptibility
\begin{equation}
 \int_0^\beta \avg{A(\tau) B} \dtau,
\label{Xeqn100-120}
\end{equation}
the ``generalized ES integral,'' though it also useful in estimating various spectral properties~\cite{hen2012excitation, blume1997excited, blume1998excited}. One approach to estimate this quantity is to simply estimate $\avg{A(\tau) B}$ for a grid of $\tau$ using {\eqref{eq:simple-AtauB-correlator-estimator}} and perform numerical integration. { Specifically, this would require evaluating the quantity $f(\tau) = e^{-\tau [E_{z_{p}}, \ldots, {E_{z_{q-1}}}]} e^{-(\beta-\tau) [E_{z_0}, \ldots, {E_{z_{p-1}}}]}$ on $N$ grid points $\tau_1, \ldots, \tau_N$ for a total cost of $O(N q^2)$ operations. The required $N$ depends on the desired accuracy, and will grow with $\beta$. For example, composite Simposon has an error scaling as $O(\beta^2 / N)$, so this approach becomes quite expensive in the low temperature $\beta \gg 1$ regime. 

Alternatively, we know from linearity that
\begin{equation}
 \int_0^\beta \left( {\wtA_k(z)} {\delta_{P_k}^{(q)}} \sum_{p=1}^{q-1} {\delta_{P_l}^{(p)}} {\wtB_l(z_{p})} \frac{ e^{-\tau [E_{z_{p}}, \ldots, {E_{z_{q-1}}}]} e^{-(\beta-\tau) [E_{z_0}, \ldots, {E_{z_{p-1}}}]}}{e^{-\beta [E_{z_0}, \ldots, E_{z_q}]}} \right) \dtau  \estimates \int_0^\beta \avg{A(\tau) B} \dtau.
\label{Xeqn101-121}
\end{equation}
In Ref.~\cite{zeng2025inequalities} the present authors showed,
\begin{equation}
 \int_0^\beta e^{-\tau [x_{j+1}, \ldots, x_q]} e^{-(\beta-\tau) [x_0, \ldots, x_j]} \dtau = - e^{-\beta [x_0, \ldots, x_q]},
\label{Xeqn102-122}
\end{equation}
with a proof also given in {\cref{app:DDE-integral-proofs}} for completeness. Hence, a direct estimator without the need for costly and noisy numerical integration is given by,
\begin{equation}
 \label{eq:esimate-es-int-ab}
 - {\delta_{P_k}^{(q)}} {\wtA_k(z)} \frac{ e^{-\beta [E_{z_0}, \ldots, {E_{z_{q-1}}}]} }{e^{-\beta [E_{z_0}, \ldots, E_{z_q}]}}\sum_{p=1}^{q-1} {\delta_{P_l}^{(p)}} {\wtB_l(z_{p})} \estimates \int_0^\beta \avg{A(\tau) B} \dtau,
\end{equation}
where we have pulled the DDE (recall that DDE is a short-hand for ``divided difference of the exponential'') ratio outside the sum since the numerator no longer depends on $p$. This significantly generalizes the ES-related estimator derived previously by the present authors~\cite{ezzell2025universal} in which $A = B = H_1$ for $H_1$ a portion of the Hamiltonian. 

All quantities, including the DDE ratio, are accessible in $O(1)$ time, but in the worst case, we must perform $q - 2$ equality checks via $\delta_{P_k}^{(p)}$. Hence, this is an $O(q)$ estimator, { a substantial improvement over numerical integration. In fact, it is even faster than evaluating $\avg{A(\tau) B}$ for a fixed $\tau$, which requires $O(q^2)$ effort. This is because the analytical result involves a DDE that is easy to evaluate given that we have already computed and stored the same quantity to evaluate the PMR-QMC weight as in} {\eqref{eq:W}}{ ---hence why it is an $O(1)$ quantity with respect to estimation as discussed in} {\cref{subsec:complexity-interlude}}.

Unlike other estimators we have considered thus far, we note that $\int_0^\beta \avg{A(\tau) B} \dtau$ estimators can have different complexities depending on $B$. For example, if we replace $B$ with a diagonal matrix $\wtB$, we find
\begin{equation}
 \label{eq:esimate-es-int-abdiag}
 \frac{- \delta_{P_k}^{(q)} \wtA_k(z)}{e^{-\beta [E_{z_0}, \ldots, E_{z_q}]}} \sum_{p=0}^q \wtB(z_p) e^{-\beta [E_{z_0}, \ldots, E_{z_q}, E_{z_p}]} \estimates \int_0^\beta \avg{A(\tau) \wtB} \dtau,
\end{equation}
which repeats the $E_{z_p}$ argument in the $e^{-\beta [\ldots]}$ summand. Computing the summand DDE now requires $O(q)$ effort with the addition of $E_{z_p}$, so the complexity is now $O(q^2)$. If we replaced $B = \wtB_l D_l P_l$ with $\wtB_{l_2} D_{l_2} P_{l_2} \wtB_{l_1} D_{l_1} P_{l_1},$ we would also have an $O(q^2)$ estimator from needing to remove an argument. On the other hand, if we replace the general $\wtB$ diagonal with $\Hdiag$, then
\begin{equation}
 \label{eq:esimate-es-int-ahdiag}
 \frac{-\delta_{P_k}^{(q)} \wtA_k(z)}{e^{-\beta [E_{z_0}, \ldots, E_{z_q}]}} \sum_{p=0}^q {E_{z_p}} e^{-\beta [E_{z_0}, \ldots, E_{z_q}, E_{z_p}]} \estimates \int_0^\beta \avg{\wtA(\tau) \Hdiag} \dtau,
\end{equation}
can be simplified into an $O(1)$ estimator using novel divided difference relations~\cite{ezzell2025universal}! The basic intuition is that the leading $E_{z_p}$---rather than $\wtB(z_p)$---relates to a derivative of the DDE with respect to $\beta$ which leads to nontrivial simplifications. 

\subsection{A generalized fidelity susceptibility integral}\label{Xsec43}\label{Xsec45-9.1}
\label{subsec:estimating-fsus}
The finite temperature fidelity susceptibility (FS)~\cite{albuquerque2010QuantumCriticalScaling, ezzell2025universal, schwandt2009QuantumMonteCarlo, wang2015FidelitySusceptibilityMade} is defined in {\eqref{eq:Fsus}}. For this reason, we denote the integrated susceptibility
\begin{equation}
 \int_0^{\beta/2} \tau \avg{A(\tau) B} \dtau,
\label{Xeqn106-126}
\end{equation}
the ``generalized FS integral.'' As with the generalized ES integral in {\cref{subsec:estimating-esus}}, we shall show that one does not need to perform a numerical integration. Instead, the present authors proved~\cite{ezzell2025universal},
\begin{equation}
 \int_0^{\beta/2} \tau e^{-\tau [x_{p+1}, \ldots, x_q]} e^{-(\beta-\tau) [x_0, \ldots, x_p]} \dtau = \sum_{r=0}^p e^{-\frac{\beta}{2}[x_0, \ldots, x_r]} \sum_{m = p + 1}^q e^{-\frac{\beta}{2} [x_r, \ldots, x_q, x_m]},
\label{Xeqn107-127}
\end{equation}
which we also show for completeness in {\cref{app:DDE-integral-proofs}}. By linearity, we obtain the corresponding estimator,
\begin{equation}
 \frac{ \delta_{P_k}^{(q)} \wtA_k(z) }{e^{-\beta [E_{z_0}, \ldots, E_{z_q}]}} \sum_{p=0}^q \delta_{P_l}^{(p)} \wtB_l(z_p) \sum_{r=0}^p e^{-\frac{\beta}{2}[E_{z_0}, \ldots, E_{z_r}]} \sum_{m = p + 1}^q e^{-\frac{\beta}{2} [E_{z_r}, \ldots, E_{z_q}, E_{z_m}]} \estimates \int_0^{\beta/2} \tau \avg{A(\tau) B} \dtau,
\label{Xeqn108-128}
\end{equation}
for $A = \wtA_k D_k P_k$ and $B = \wtB_l D_l P_l$. By direct algebra, we can reorder the sums (see again Ref.~\cite{ezzell2025universal}),
\begin{equation}
 \label{eq:esimate-fs-int-ab}
 \frac{ \delta_{P_k}^{(q)} \wtA_k(z) }{e^{-\beta [E_{z_0}, \ldots, E_{z_q}]}} \sum_{r=0}^q e^{-\frac{\beta}{2}[E_{z_0}, \ldots, E_{z_r}]} \sum_{\substack{p,m=r \\ p \leq m - 1}} \delta_{P_l}^{(p)} \wtB_l(z_p) e^{-\frac{\beta}{2} [E_{z_r}, \ldots, E_{z_q}, E_{z_m}]} \estimates \int_0^{\beta/2} \tau \avg{A(\tau) B} \dtau,
\end{equation}
whereupon it is apparent this is an $O(q^4)$ estimator, as the innermost DDE requires $O(q)$ effort.

As with the general ES integral, the general FS estimator and its corresponding complexity can differ depending on the form of $A$ and $B$. For example, if we consider $\wtA$ and $\wtB$ both diagonal, we find
\begin{equation}
 \label{eq:esimate-fs-int-abdiag}
 \frac{\wtA(z) }{e^{-\beta [E_{z_0}, \ldots, E_{z_q}]}} \sum_{r=0}^q e^{-\frac{\beta}{2}[E_{z_0}, \ldots, E_{z_r}]} \sum_{\substack{p,m=r \\ p \leq {m}}} \wtB(z_p) e^{-\frac{\beta}{2} [E_{z_r}, \ldots, E_{z_q}, {E_{z_p}}, E_{z_m}]} \estimates \int_0^{\beta/2} \tau \avg{\wtA (\tau) \wtB} \dtau,
\end{equation}
which is also an $O(q^4)$ estimator. Yet, if the diagonal $\wtB$ is $\Hdiag$ in particular, we find,
\begin{equation}
 \label{eq:esimate-fs-int-ahdiag}
 \frac{\wtA(z) }{e^{-\beta [E_{z_0}, \ldots, E_{z_q}]}} \sum_{r=0}^q e^{-\frac{\beta}{2}[E_{z_0}, \ldots, E_{z_r}]} \sum_{\substack{p,m=r \\ p \leq {m}}} {E_{z_p}} e^{-\frac{\beta}{2} [E_{z_r}, \ldots, E_{z_q}, {E_{z_p}}, E_{z_m}]}\estimates \int_0^{\beta/2} \tau \avg{\wtA (\tau) \wtB} \dtau,
\end{equation}
which can actually be reduced to a $O(q^3)$ estimator using novel divided difference relations derived in Ref.~\cite{zeng2025inequalities,ezzell2025universal}. Again, the change from $\wtB(z_p)$ to $E_{z_p}$ is crucial to relate this expression to derivatives of the DDE with respect to $\beta$. Although challenging, it is possible that this could be further improved, perhaps in other special cases.

{
\section{Numerical demonstration of our estimators}\label{Xsec44}\label{Xsec46-9.2}
\label{sec:numerical-demo}

Having thus derived estimators for arbitrary static operators and non-trivial dynamic generalizations therein, we now provide numerical evidence that they function correctly in practice. In support of these results, our code is open source~\cite{ezzell2025code} and user-friendly, as discussed in {\cref{app:code-details}}. Of note, it is natively compatible with arbitrary spin-1/2 Hamiltonians, and one can readily adapt our estimator code to PMR-QMC codes for general high spin models~\cite{babakhani2025quantum} and Bose-Hubbard models~\cite{akaturk2024quantum} as desired.

Evidence for the efficacy of PMR-QMC on Hamiltonian based observables has been well established~\cite{albash2017OffdiagonalExpansionQuantum,gupta2020PermutationMatrixRepresentation,barash2024QuantumMonteCarlo,akaturk2024quantum,ezzell2025universal,babakhani2025quantum}. For example, Refs.~\cite{albash2017OffdiagonalExpansionQuantum,gupta2020PermutationMatrixRepresentation,ezzell2025universal} directly show PMR-QMC agrees with results obtained from stochastic series expansion (SSE), with Refs.~\cite{albash2017OffdiagonalExpansionQuantum,gupta2020PermutationMatrixRepresentation} emphasizing a simulation time advantage for PMR-QMC. On the other hand, Refs.~\cite{ezzell2025universal,barash2024QuantumMonteCarlo} focus on the study of physical models of interest, with Ref.~\cite{ezzell2025universal} successfully studying critical phenomena.

Accordingly, we focus on proof-of-principle demonstrations that our method continues to work on non-trivial random and dynamic observables, which is the main appeal of our prior derivations. At the same time, our demonstrations are fully black-box in the sense that no model or observable specific changes to the source code are needed throughout~\cite{ezzell2025universal}. By simply modifying the Hamiltonian input file, one can also readily study any other spin-1/2 model (see {\cref{app:code-details}}).

\subsection{TFIM verification results}\label{Xsec45}\label{Xsec47-9.3}
\label{subsec:verify-results}

We begin with the spin-1/2 transverse-field Ising model (TFIM) for which we know all Pauli strings are natively canonical and hence have unbiased estimators. Specifically, we study the TFIM on a square lattice with open boundary conditions,
\begin{equation}
 \label{eq:tfim}
 H = -\sum_{\langle i, j \rangle }Z_i Z_j - \lambda \sum_{i=1}^{n} X_i,
\end{equation}
for $X_i,Z_i$ are the standard $X$ and $Z$ Pauli spin-1/2 matrices acting on the $i{\text{th}}$ site and $\avg{i, j}$ denotes only all the nearest neighbor connections on the square $n \times n$ lattice.

\begin{figure*}[htp!]
    \begin{subfigure}[t]{0.48\textwidth}
        \includegraphics[width=1\linewidth]{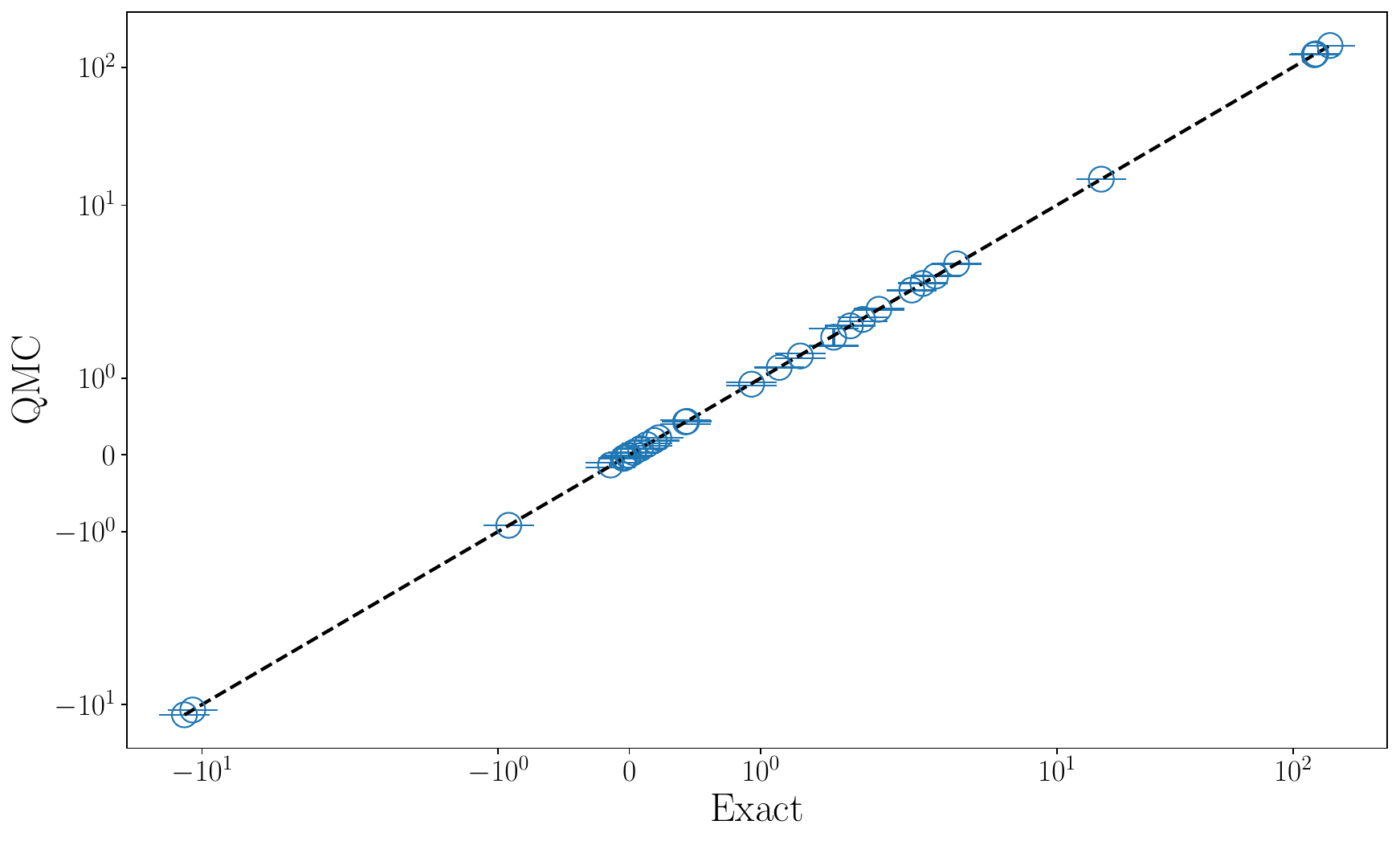}
        \caption{{  Exact and QMC values agree on a line ($3 \sigma$ error bars shown, but all QMC estimates agree within 2$\sigma$).}}
        \label{subfig:verify-standard}
    \end{subfigure}%
    \hfill 
    \begin{subfigure}[t]{0.48\textwidth}
        \includegraphics[width=1\linewidth]{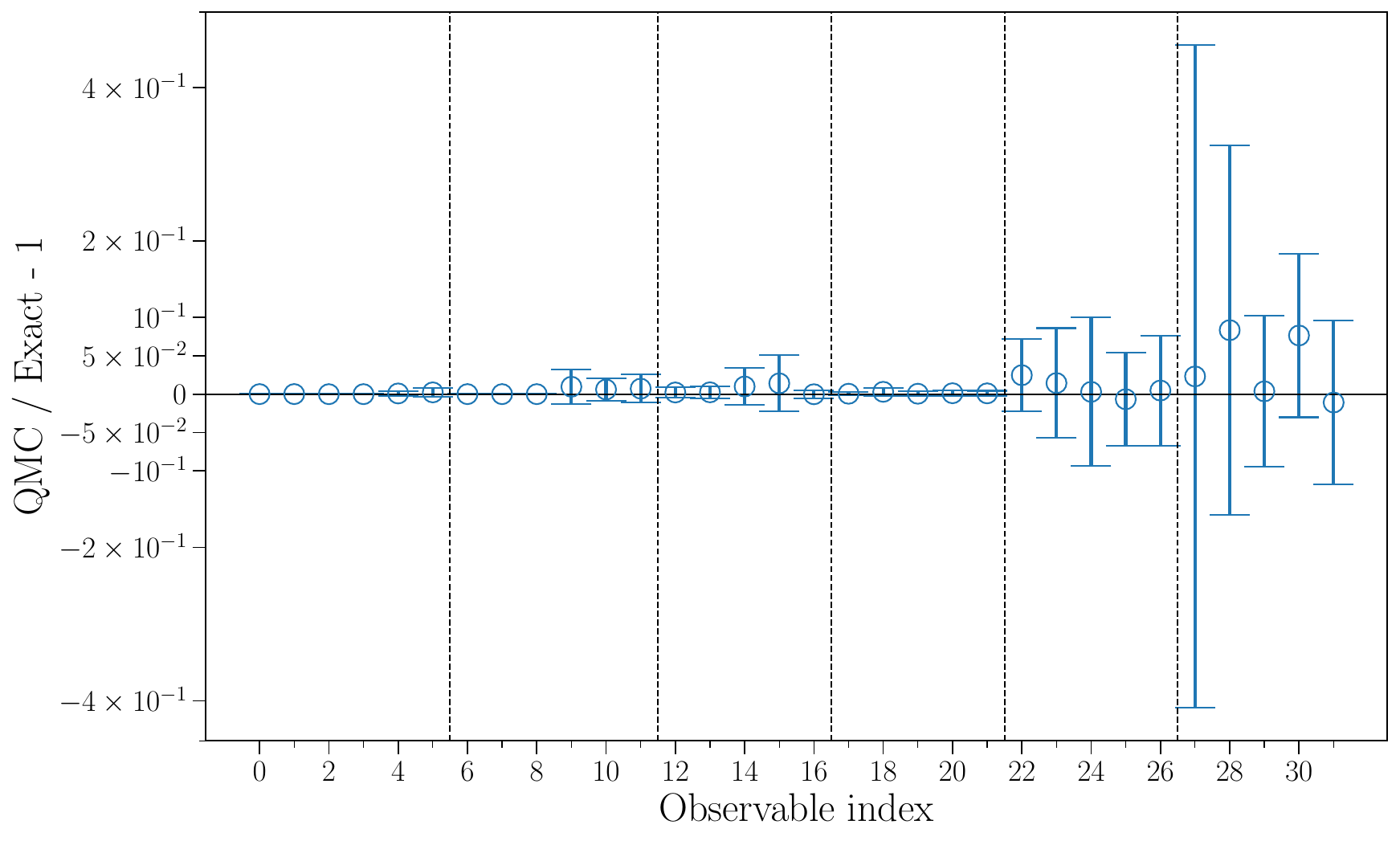}
        \caption{{  Ratio of QMC estimates to exact values are clearly all close to one up to $3\sigma$ error bars. }}
         \label{subfig:verify-custom}
    \end{subfigure}%
    \\
    \begin{subfigure}[t]{1\textwidth}
        \begin{tabular}{|c|c|}
            \hline  
            $0$ & $-\avg{H}$ \\
            \hline
            $1$ & $\avg{H^2}$ \\
            \hline 
            $2$ & $-\avg{D_0}$ \\
            \hline 
            $3$ & $\avg{D_0^2}$ \\
            \hline 
            $4$ & $- \avg{\Gamma}$ \\
            \hline 
            $5$ & $ \avg{\Gamma^2} $ \\
            \hline  
        \end{tabular}
        ~
        \begin{tabular}{|c|c|}
            \hline  
            $6$ & $\avg{D_0(\beta/2) D_0}$ \\
            \hline
            $7$ & $\int_0^\beta \avg{D_0(\tau) D_0} \dtau $ \\
            \hline 
            $8$ & $\int_0^{\beta/2} \tau \avg{D_0(\tau) D_0} \dtau $ \\
            \hline 
            $9$ & $\avg{\Gamma(\beta/2) \Gamma}$ \\
            \hline 
            ${10}$ & $\int_0^\beta \avg{\Gamma(\tau) \Gamma} \dtau $ \\
            \hline 
            ${11}$ & $\int_0^{\beta/2} \tau \avg{\Gamma(\tau) \Gamma} \dtau $ \\
            \hline  
        \end{tabular}
        ~
        \begin{tabular}{|c|c|}
            \hline  
            ${12}$ & $\chi_E^{D_0}$ \\
            \hline
            ${13}$ & $\chi_F^{D_0}$ \\
            \hline 
            ${14}$ & $\chi_E^{\Gamma}$ \\
            \hline 
            ${15}$ & $\chi_F^{\Gamma}$ \\
            \hline 
            ${16}$ & $C_v$ \\
            \hline  
        \end{tabular}
        ~
        \begin{tabular}{|c|c|}
            \hline  
            ${17}$ & $\avg{A}$ \\
            \hline
            ${18}$ & $\avg{A^2}$ \\
            \hline 
            ${19}$ & $  \avg{A(\beta/2) A}  $ \\
            \hline 
            ${20}$ & $ \int_0^\beta \avg{A(\tau) A} \dtau  $ \\
            \hline 
            ${21}$ & $ \int_0^{\beta/2} \tau \avg{A(\tau) A} \dtau  $ \\
            \hline  
        \end{tabular}
        ~
        \begin{tabular}{|c|c|}
            \hline  
            ${22}$ & $\avg{B}$ \\
            \hline
            ${23}$ & $\avg{B^2}$ \\
            \hline 
            ${24}$ & $  \avg{B(\beta/2) B}  $ \\
            \hline 
            ${25}$ & $ \int_0^\beta \avg{B(\tau) B} \dtau  $ \\
            \hline 
            ${26}$ & $ \int_0^{\beta/2} \tau \avg{B(\tau) B} \dtau  $ \\
            \hline  
        \end{tabular}
        ~
        \begin{tabular}{|c|c|}
            \hline  
            ${27}$ & $\text{Re}(\avg{AB})$ \\
            \hline
            ${28}$ & $\text{Im}(\avg{AB})$ \\
            \hline 
            ${29}$ & $  \avg{A(\beta/2) B}  $ \\
            \hline 
            ${30}$ & $ \int_0^\beta \avg{A(\tau) B} \dtau  $ \\
            \hline 
            ${31}$ & $ \int_0^{\beta/2} \tau \avg{A(\tau) B} \dtau  $ \\
            \hline  
        \end{tabular}
        ~
        \caption{A legend of $x$-axis labels. Observables $1$ to ${16}$ are { Hamiltonian-based observables defined in terms of $H,$ $D_0 \equiv \Hdiag,$ and $\Gamma \equiv \Hoffdiag = H - D_0$. Remaining ${17}$ to ${31}$ are custom observables defined in terms of $A = X_1 + Z_2 Z_3$ and $B \approx -0.77 X_3 X_9 + 0.16 Z_3 Z_6 Z_9$ $- 0.97 Y_1 X_6 Z_7$ (see
        \cref{eq:verify-A,eq:verify-B}).
        }
        }
    \end{subfigure}
    \caption{We demonstrate clear agreement between PMR-QMC estimates and { exact calculation (just direct numerical linear algebra as described in} \cref{app:code-details}) for a wide variety of observables. Calculations are performed for the $3\times 3$ square TFIM { with open boundary conditions} in Eq.~(\ref{eq:tfim}) for $\beta = 1.0, \lambda = 0.5$. QMC points and error bars represent the average and thrice the standard deviation, $3\sigma$, over 100 independent runs with different random seeds.}
    \label{fig:verify-qmc}
\end{figure*}

In support of the veracity of our method, we first estimate a variety of observables for a $3 \times 3$ instance of the square TFIM for $\beta = 1.0$ and $\lambda = 0.5$ with our QMC method and show they agree with ``exact (numerical) calculations'' as shown in {\cref{fig:verify-qmc}} for 32 different static and dynamic observables. { By ``exact calculations,'' we simply mean performing direct numerical linear algebra up to numerical precision with the Python packages Numpy}~\cite{harris2020array} { and SciPy}~\cite{virtanen2020scipy}. { For more details, we refer interested readers to} {\cref{app:code-details}}. The observables estimated in {\cref{fig:verify-qmc}} are defined in terms of,
\begin{align}
	\Hdiag &\equiv \diag{H} = - \sum_{\avg{i, j}} Z_i Z_j, \\
	\Hoffdiag &\equiv H - \diag{H} = -0.5 \sum_{i=1}^3 X_i, \\
	A &\equiv X_1 + Z_2 Z_3, \label{eq:verify-A} \\
	B &\equiv -0.773712 X_3 X_9 + 0.155294 Z_3 Z_6 Z_9 - 0.966529 Y_1 X_6 Z_7. \label{eq:verify-B}
\end{align}
The choice of these particular operators is not totally arbitrary and serves to illustrate that our different types of estimators work. 

First, the  \emph{Hamiltonian-based observables} $H$, $\Hdiag,$ and $\Hoffdiag$ have tailored estimators (see {\cref{tab:static-complexity}} for a reminder). Second, $X_1$ and $Z_2 Z_3$ are not Hamiltonian based but are obviously canonical in that they can be written clearly as a single term of the form $\Lambda_l D_l P_l$. Finally, any Pauli string can be written in canonical form for the TFIM (see {\cref{subsec:illustrative-examples}}), so $B$ is chosen to be a sum of random Pauli strings with random coefficients uniformly sampled from $[-1, 1]$. To avoid rare sampling issues (see {\cref{subsec:improved-estimator}}), we have restricted them to low-weight Paulis. As another technical remark, observables 27 and 28 show that; as expected, our method is capable of estimating non-Hermitian expectation values since it successfully estimates the real and imaginary parts of $\avg{AB}$. Of course, our derivations are for general operators, but for now, our code only supports inserting Pauli observables with real weights (see {\cref{app:code-details}}), so this is why we probe a non-Hermitian operator in such an indirect way. 

\subsection{TFIM proof-of-principle results}\label{Xsec46}\label{Xsec48-9.4}
\label{subsec:larger-results}

\begin{figure*}[htp!]
    \begin{subfigure}[t]{0.32\textwidth}
        \includegraphics[width=1\linewidth]{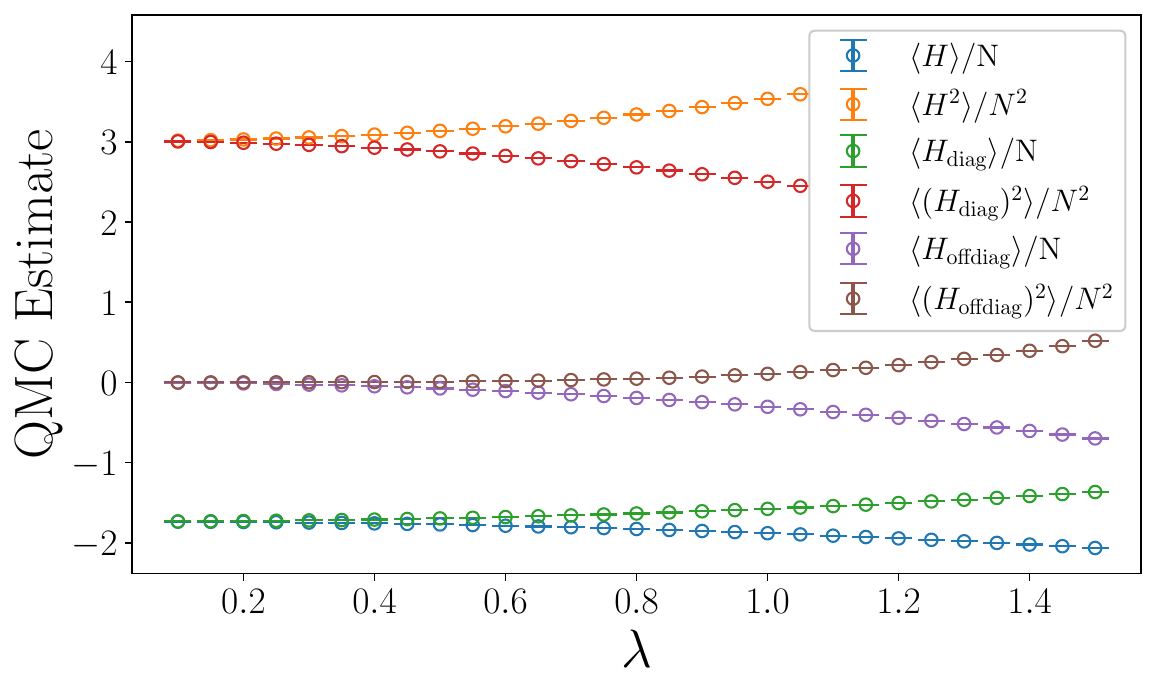}
        \caption{Static { Hamiltonian-based} observables.}
        \label{subfig:standard-static-obs}
    \end{subfigure}%
    ~
    \begin{subfigure}[t]{0.32\textwidth}
        \includegraphics[width=1\linewidth]{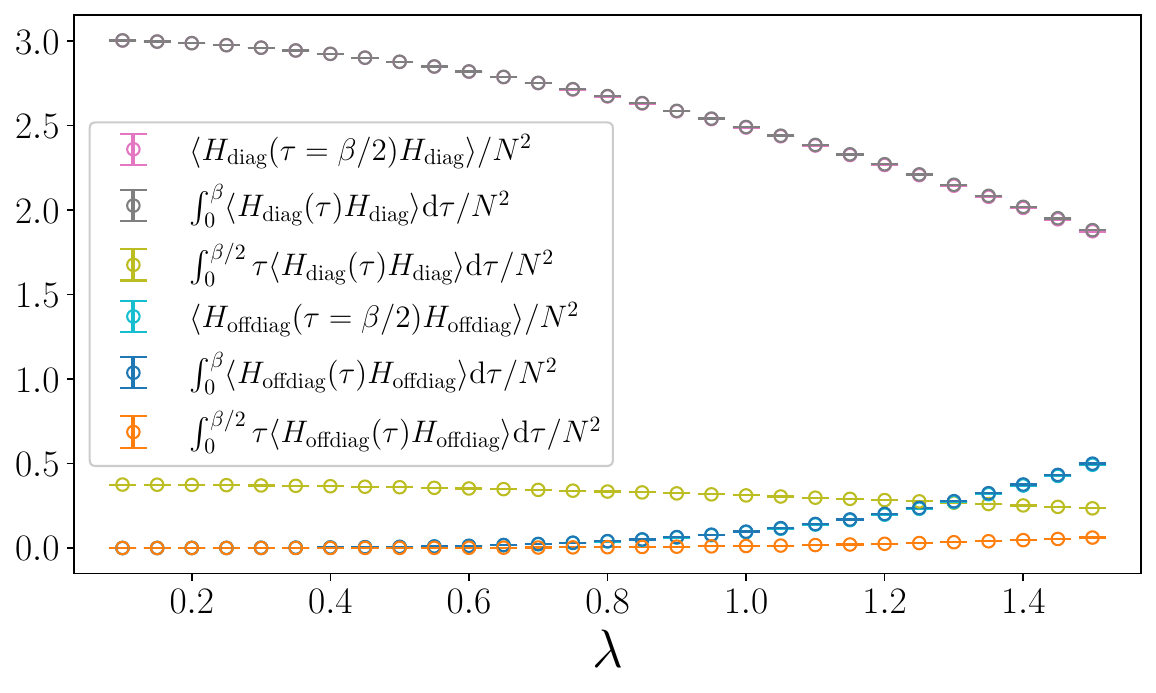}
        \caption{Dynamic { Hamiltonian-based} observables.}
        \label{subfig:standard-dynamic-obs}
    \end{subfigure}%
    ~
    \begin{subfigure}[t]{0.32\textwidth}
        \includegraphics[width=1\linewidth]{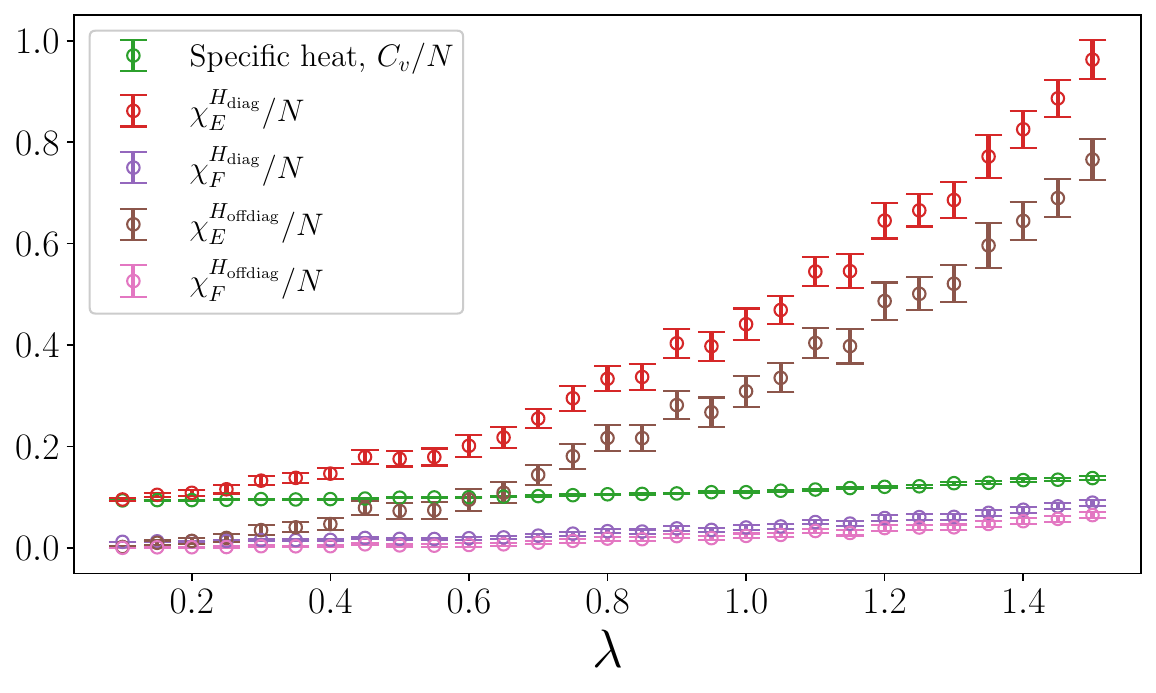}
        \caption{Derived { Hamiltonian-based} observables computed via jackknife binning analysis~\cite{berg2004introduction, barash2024QuantumMonteCarlo}.}
    \end{subfigure}%
    \caption{We estimate static, dynamic, and derived { Hamiltonian-based} observables for the $8\times 8$ square TFIM in \cref{eq:tfim} { with open boundary conditions} as a function of transverse field strength. Points and error bars represent the average and twice the standard deviation, $2\sigma$, over 100 independent runs with different random seeds.}
    \label{fig:standard-obs}
    \label{subfig:standard-derived-obs}
\end{figure*} 

Having demonstrated that our method correctly reproduces exactly computable values for a $3 \times 3$ TFIM instance in {\cref{subsec:verify-results}}, we now use PMR-QMC to estimate a similar set of 31 observables for an $8 \times 8$ instance of the TFIM for $\beta = 1.0$ and $\lambda \in [0.1, 1.5]$. The results are shown in {\cref{fig:standard-obs}} and \cref{fig:custom-obs}, respectively. As with the $3 \times 3$ example, the $8 \times 8$ plots explore observables defined in term of (the now $2^8 \times 2^8$ matrices) $H$, $\Hdiag$, $\Hoffdiag$,
\begin{align}
 A &\equiv X_1 + Z_2 Z_3, \label{eq:larger-A}
\end{align}
and a random observable
\begin{align}
 B &\equiv -0.241484 Z_{22} X_{31} Z_{49} + 0.784290 Y_{17} Z_{53} + 0.929765 Y_{62}, \label{eq:larger-B}.
\end{align}
The motivation for choosing these four observables is the same as discussed in {\cref{subsec:verify-results}}.

\begin{figure*}[hpt]
    \begin{subfigure}[t]{0.32\textwidth}
        \includegraphics[width=1\linewidth]{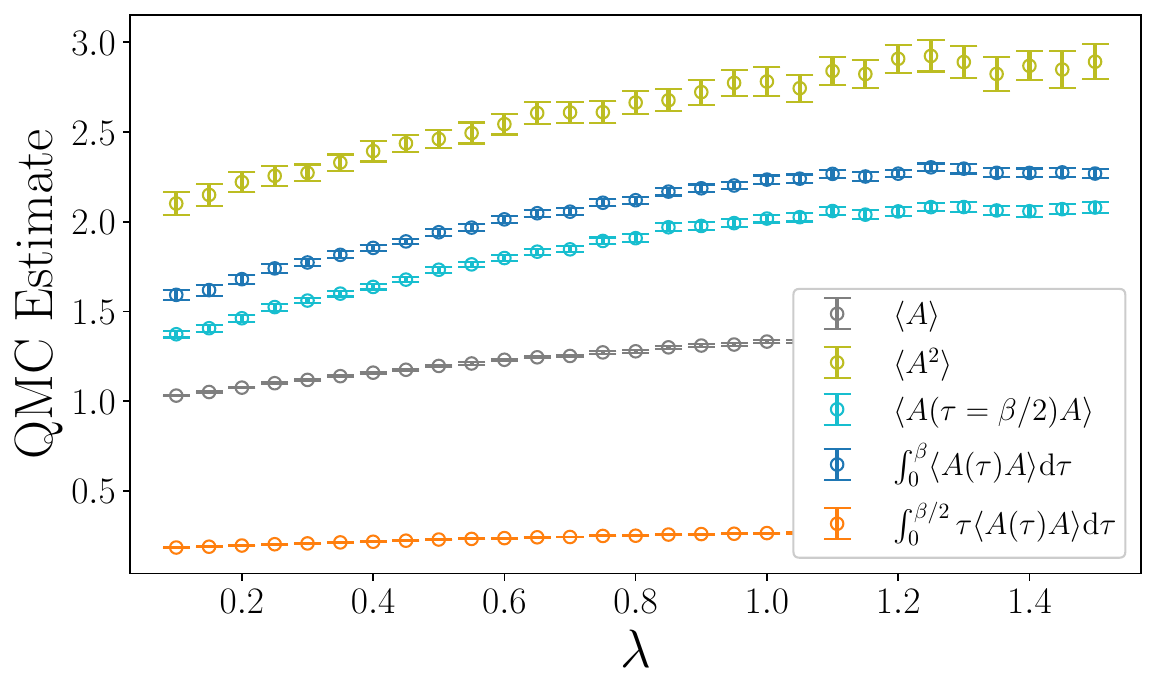}
        \caption{Observables using $A = X_1 + Z_2Z_3$.}
        \label{subfig:custom-A}
    \end{subfigure}%
    ~
    \begin{subfigure}[t]{0.32\textwidth}
        \includegraphics[width=1\linewidth]{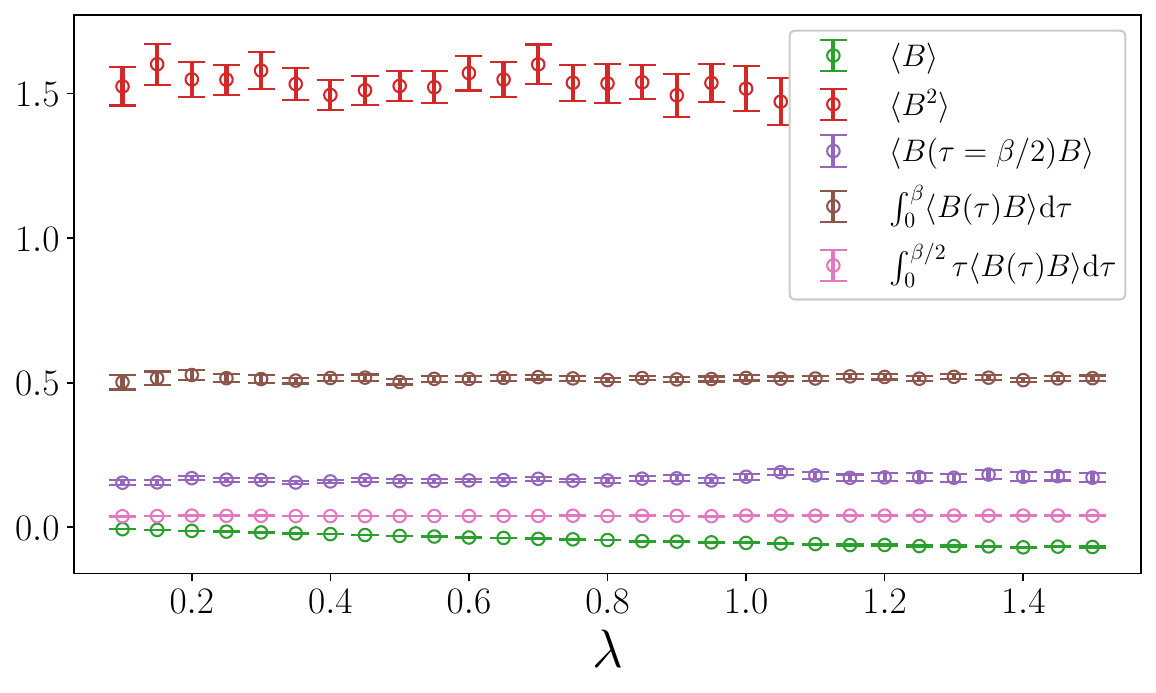}
        \caption{Observables using $B \approx 0.93 Y_{62}$ $-0.24 Z_{22}X_{31}Z_{49} + 0.78 Y_{17} Z_{53}$.}
        \label{subfig:custom-B}
    \end{subfigure}%
    ~
    \begin{subfigure}[t]{0.32\textwidth}
        \includegraphics[width=1\linewidth]{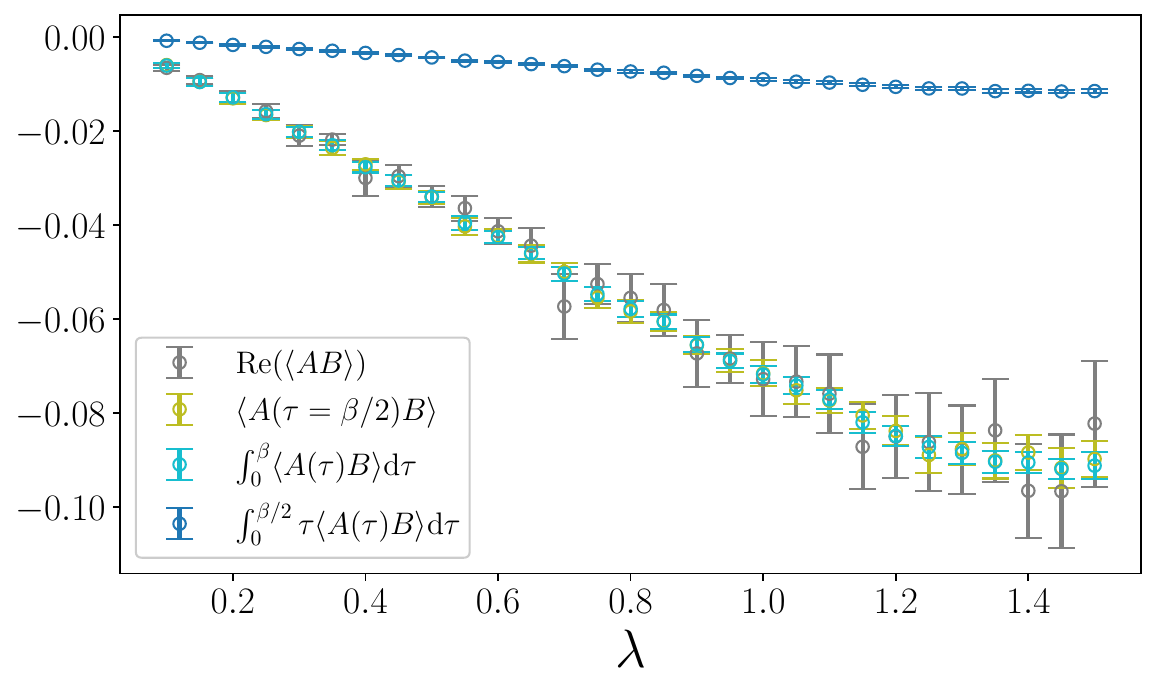}
        \caption{Quantities involving both $A$ in \cref{eq:larger-A} and $B$ in \cref{eq:larger-B}.}
        \label{subfig:custom-AB}
    \end{subfigure}%
    \caption{We estimate custom static and dynamic observables defined in terms of $A$ and $B$ given in \cref{eq:larger-A,eq:larger-B}, respectively for the $8\times 8$ square TFIM in \cref{eq:tfim} { with open boundary conditions} as a function of transverse field strength. Points and error bars represent the average and twice the standard deviation, $2\sigma$, over 100 independent runs with different random seeds.}
    \label{fig:custom-obs}
\end{figure*}

For a model of this size, direct numerical linear algebra verification is no longer tractable. However, many trends are consistent with general expectations. For example, $H = \Hdiag + \Hoffdiag$ with $\Hoffdiag \propto \lambda$ by definition. As anticipated, $\avg{\Hoffdiag}$ and $\avg{\Hoffdiag^2}$ converge to $0$ as $\lambda \rightarrow 0$. Consistent with this, $\avg{H} \approx \avg{\Hdiag}$ and $\avg{H^2} \approx \avg{\Hdiag^2}$ as $\lambda \rightarrow 0$ as well. {
For the random observables, to our knowledge, there is no appropriate modification to existing QMC codes that makes them capable of measuring such non-trivial observables. This is indeed part of the motivation for this work. } Nevertheless, some simple sanity checks can be made. For example, since $B$ is random, we do not expect a pure $B$ observables to depend on $\lambda$, which agrees with the flat trends in {\cref{fig:custom-obs}}b.

\subsection{Random model proof of principle results}\label{Xsec47}\label{Xsec49-10}

Having demonstrated the validity of our method, we now demonstrate its capability to study non-trivial random observables on a randomly rotated toy model. First, we define the simple two qubit model
\begin{equation}
 \label{eq:prl-model}
 H(\lambda) = Z_1 Z_2 + \lambda(Z_1 + Z_2) + 0.1 (X_1 + X_2),
\end{equation}
first introduced and studied in a quantum sensing experiment on an NMR system via the fidelity approach~\cite{zhang2008detection}. By design, this model exhibits an avoided level crossing near the critical points $\lambda \approx \pm 1$, and we fix $\lambda = 1.0$ in our results here. Now, we define the random observable,
\begin{equation}
 \label{eq:random-o-for-prl-model}
 O = 0.683403\, Z_{1}
- 0.643777\, Y_{1} Z_{2}
- 0.662378\, Z_{2}
+ 0.738353\, Y_{1} Y_{2}
\notag - 0.920660\, X_{1} X_{2}.
\end{equation}
As with the TFIM example above, we can readily estimate $\avg{O}, \avg{O^2}, \avg{O(\tau) O}$, $\int_0^\beta \avg{O(\tau) O} \dtau,$ and $\int_0^{\beta/2} \tau \avg{O(\tau) O} \dtau$, and verify they are all correct.

To further test the algorithm, we additionally rotate the Hamiltonian and the observable, i.e., we define
\begin{equation}
 \label{eq:random-H-and-O}
 H_U = U H(\lambda = 1) U^{\dagger}, \ \ \ O_U = U O U^{\dagger},
\end{equation}
for $U$ a linear combination of random anti-commuting Paulis with support over 100 qubits with normalized coefficients. The $U$ we use contains 59 unique Pauli terms which include $0.067559 \sum_{i=1}^{100} X_i$ and $-0.110868 X_1 X_2 X_3 Z_4$ which anti-commute, for example. Conjugation by this $U$ causes $H_U$ and $O_U$ to contain 179 and 295 terms, respectively. The complete specifications are provided in our open source code~\cite{ezzell2025code} for the exact $U$, $H_U$ and $O_U$ used in this example. Despite the large number of terms, we have chosen this $U$ such that $H_U$ does not have a sign problem, as discussed in Appendix C of Ref.~\cite{ezzell2025universal}.

As such, we expect our approach to be able to estimate $\avg{O_U}_{H_U} = \Tr[O_U e^{-\beta H_U}] / \Tr[e^{-\beta H_U}]$ and variations thereof and return the same answer as $\avg{O}_H = \Tr[O e^{-\beta H}] / \Tr[e^{-\beta H}]$. In {\cref{fig:comparison-of-O-and-OU}}, we verify that this holds true by comparing both $\avg{O}_H$ to $\avg{O_U}_{H_U}$ and the estimated values of
\begin{equation}
 \label{eq:f-and-h}
 f(\beta) = \int_0^{\beta/2} \tau \avg{O(\tau) O}_H \dtau \ \ \ \text{and} \ \ \ h(\beta) = \int_0^{\beta/2} \tau \avg{O_U (\tau) O_U }_{H_U} \dtau
\end{equation}
as a function of $\beta$. Up to error bars, all estimates agree with each other for both observables. Although not shown, similar agreement holds for other observables $\avg{O^2}$, $\avg{O(\tau) O}$, and $\int_0^\beta \avg{O(\tau) O} \dtau$.

\begin{figure}[htp]
    \centering
    \includegraphics[width=0.45\linewidth]{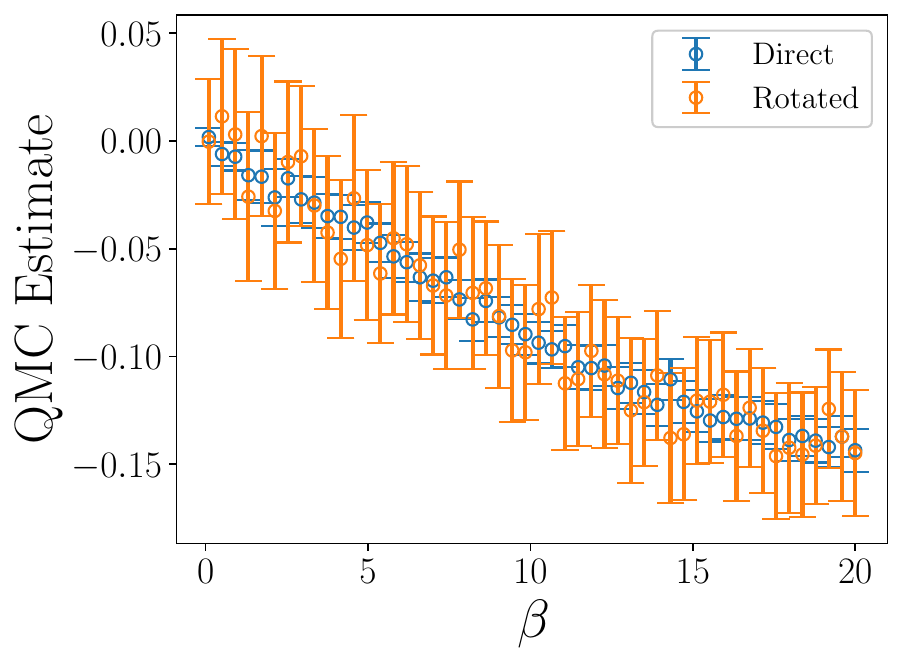}
    \includegraphics[width=0.45\linewidth]{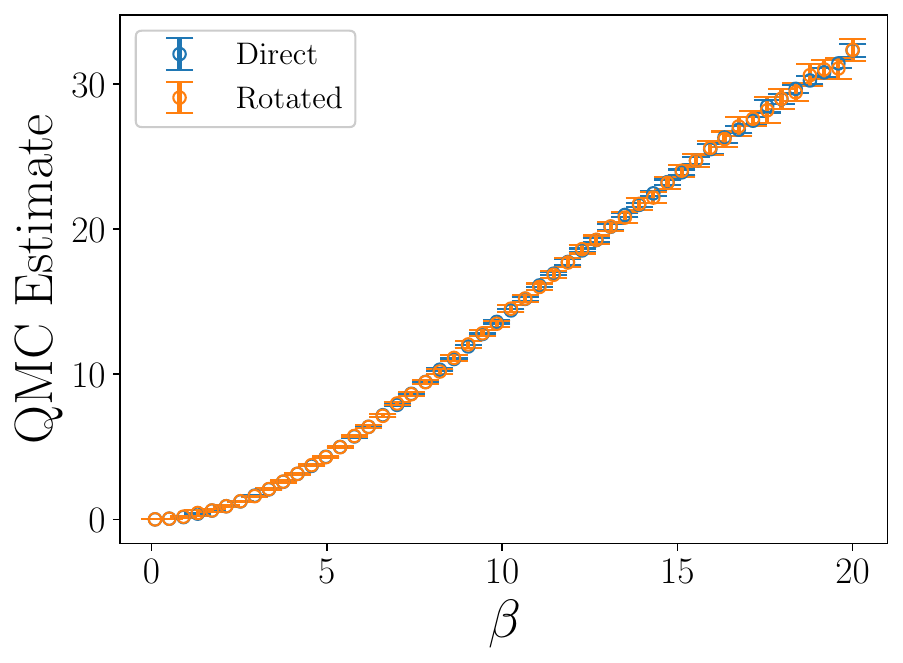}
    \caption{We verify that our method outputs the same result up to $3 \sigma$ error bars in the estimation of direct or rotated observables on $H$ and $H_U$ defined in \cref{eq:prl-model,eq:random-H-and-O}, respectively. (a) A comparison of $\avg{O}_H$ (direct) and $\avg{O_U}_{H_U}$ (rotated) as a function of $\beta$. (b) A comparison of $f(\beta)$ and $h(\beta)$ defined in \cref{eq:f-and-h} as a function of $\beta$.}
    \label{fig:comparison-of-O-and-OU}
\end{figure}

\subsection{Resource estimates of numerical experiments}\label{Xsec48}\label{Xsec50-A.1}

The numerical experiments for the two qubit and 100 qubit rotated model in {\cref{fig:comparison-of-O-and-OU}} were performed on a laptop with a 2.3GHz 8-core Intel Core i9 CPU and 64GB of RAM. Data collected for the TFIM used to make {\cref{fig:verify-qmc}}--\cref{fig:custom-obs} were performed on a University of Southern California high-performance computing cluster (HPC). As a QMC method, our approach is naturally suited for lazy parallel execution and exhibits near-perfect strong scaling as the number of cores increases (see, e.g., Fig. 8 of Ref.~\cite{barash2024QuantumMonteCarlo}), and all simulations in this work take advantage of parallel computations. For TFIM data taken on the HPC, 100 parallel computations were performed, but this is likely excessive, as we only used 6 threads for data taken on our laptop. Our algorithm is highly memory efficient~\cite{barash2024QuantumMonteCarlo, gupta2020CalculatingDividedDifferences}, so 200\,MB of memory per run was more than sufficient for all calculations done in this work.

\begin{figure}[htp]
    \centering
    \includegraphics[width=0.45\linewidth]{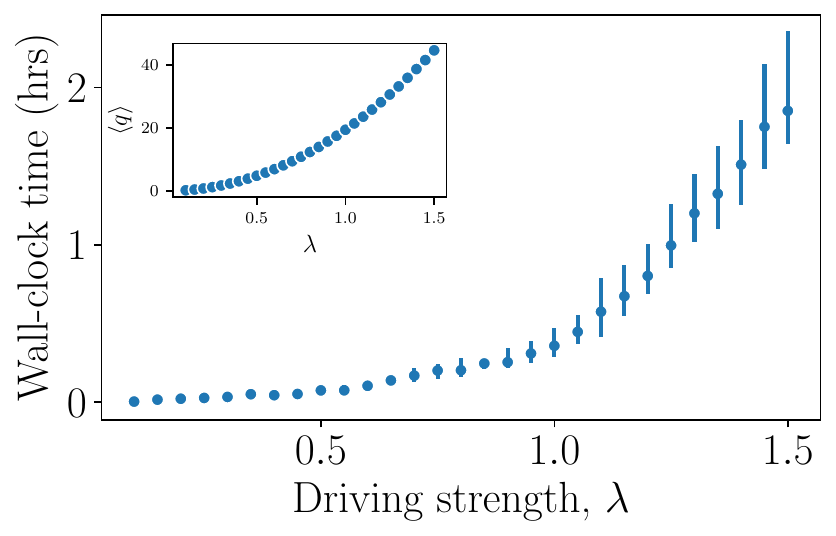}
    \includegraphics[width=0.45\linewidth]{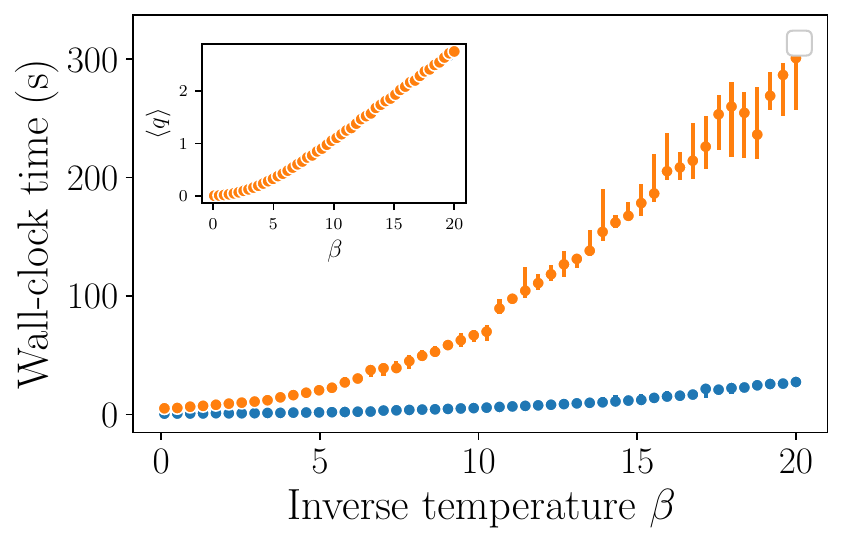}
    \caption{A brief summary of wall-clock times used to generate data in this work. (Left) Wall-clock time and $\avg{q}$ as a function of $\lambda$ for the $8 \times 8$ TFIM data to collect all measurement estimates plotted in \cref{fig:standard-obs,fig:custom-obs}. (Right)  Wall-clock time and $\avg{q}$ as a function of $\beta$ in the study of the direct and rotated models in \cref{fig:comparison-of-O-and-OU}. Error bars are 95\% intervals, i.e., the band in which 95\% of all empirical wall-clock or $\avg{q}$ values lie.}
    \label{fig:timing-estimates}
\end{figure} 

 As such, we focus on reporting wall-clock time estimates. In {\cref{fig:timing-estimates}}, we provide wall-clock timing estimates in the collection of the $8 \times 8$ TFIM data plotted in {\cref{fig:standard-obs}}, \cref{fig:custom-obs} and the direct versus rotated model comparison in {\cref{fig:comparison-of-O-and-OU}}. For the TFIM, wall-clock time increases with increasing $\lambda$ (scaling roughly as $\lambda^2$) as expected from prior studies~\cite{albash2017OffdiagonalExpansionQuantum, gupta2020PermutationMatrixRepresentation, ezzell2025universal}. Similarly, simulation time also scales as a modest power law with $\beta$ ($\sim 0.72 \beta^{1.00}$ for the direct model and $\sim 6.83 \beta^{1.11}$ for the rotated model) as expected generally~\cite{barash2024QuantumMonteCarlo,ezzell2025universal}. For a more detailed discussion of wall-clock times and scalings in general, we refer readers to numerous prior studies~\cite{albash2017OffdiagonalExpansionQuantum,gupta2020PermutationMatrixRepresentation,barash2024QuantumMonteCarlo,ezzell2025universal}---Refs.~\cite{albash2017OffdiagonalExpansionQuantum,gupta2020PermutationMatrixRepresentation} in particular did direct comparisons to both path integral QMC and SSE, finding that PMR-QMC is much faster in the study of disordered or glassy systems.

 In this work, however, we observe one novel feature in the resource estimates shown in the right plot of {\cref{fig:timing-estimates}}. Specifically, the rotated model takes more effort to simulate than the direct model, even though the average quantum dimension $\avg{q}$ remains the same. One direct reason for this is that $O_U$ contains substantially more terms ($K = 295$) compared to the direct $O$ (which has $K = 5$), and estimation complexity depends on $K$ as summarized in {\cref{tab:static-complexity}}.

 }

\section{Summary and conclusions}\label{Xsec49}\label{Xsec51-A.2}
\label{sec:conclusion}

In this work, we addressed the question: \emph{Can quantum Monte Carlo estimate the thermal expectation value of an arbitrary operator?} Within the permutation matrix representation (PMR) formulation of QMC, we showed that the answer is yes. To show this, we first developed a rigorous formulation of PMR based on a carefully constructed group basis with a regular natural group action. By expressing both the Hamiltonian and observables in this common basis and performing an off-diagonal series expansion, we derived a formal estimator for arbitrary operators within PMR-QMC. However, we showed that this formal estimator generically involves division by zero.

Using an illustrative example, we demonstrated that such formal divisions by zero can lead to biased estimates in practice. We then identified and defined a class of operators in \emph{canonical form}, for which unbiased estimators arise naturally that avoid any pathologies by construction. Combining group-theoretic arguments with this observation, we showed constructively that any operator admits an unbiased estimator. As a result, arbitrary static operators can be estimated within PMR-QMC.

We next extended our framework to important classes of dynamical operators, including imaginary-time correlation functions and integrated susceptibilities. Interestingly, the technical underpinning of this extension mostly follows from recently derived divided difference properties~\cite{ezzell2025universal,zeng2025inequalities}, which may be of independent interest. These dynamic estimators generalize previously known PMR-QMC estimators for energy and fidelity susceptibilities~\cite{ezzell2025universal} and enable the study of a broad class of dynamical and spectral quantities that are often inaccessible to conventional numerical approaches.

Finally, we demonstrated the practical power of our approach through numerical experiments on the transverse-field Ising model and on a highly non-local random toy model with up to one hundred spins. In both cases, we accurately estimated random non-local observables, highlighting the robustness and generality of the method. Our implementation is straightforward to use: a user specifies only the Hamiltonian and observable as text files of Pauli strings, from which ergodic, detailed-balance-preserving QMC updates are automatically constructed. Further, the estimators are written to be compatible with arbitrary observables. To facilitate reproducibility and interesting applications of our method, all data and code developed in this work is open source~\cite{ezzell2025code}.

At present, our implementation targets spin-1/2
Hamiltonians and observables. Since the underlying derivations are fully general, a natural extension is to implement these estimators in PMR-QMC frameworks compatible with higher-spin systems~\cite{babakhani2025quantum} or bosonic models~\cite{akaturk2024quantum}. Moreover, our numerical examples primarily serve as proofs of principle, so an important direction for future work is to exploit this generality to explore previously inaccessible physical phenomena. Finally, automating the construction of canonical operators would further streamline the use of our approach and broaden its applicability.

\section*{Data and code availability statement}
 Source code to reproduce all numerical experiments in this work is available in an open source GitHub library as a citable Zenodo repository~\cite{ezzell2025code}, \href{https://zenodo.org/records/17786518}{https://zenodo.org/records/17786518} as well as a GitHub repository, \href{https://github.com/naezzell/advmeaPMRQMC}{https://github.com/naezzell/advmeaPMRQMC}. Our library includes detailed user instructions as well as all data and scripts necessary to reproduce all plots in this work.

\section*{Declaration of competing interest}
The authors declare that they have no known competing financial interests or personal relationships that could have appeared to influence the work reported in this paper.

\section*{Acknowledgments}
{  We acknowledge the contributions of an anonymous reviewer whose ideas helped us show conclusively that all operators have an unbiased estimator in PMR-QMC. We also thank Lev Barash for many helpful discussions, assistance with our code, and comments on our manuscript. 
} IH acknowledges support by the Office of Advanced Scientific Computing Research of the U.S. Department of Energy under Contract No DE-SC0024389. 
N.E. acknowledges partial support from the Gold Family Fellowship, the Graduate Student Fellowship awarded by the USC Department of Physics and Astronomy in the Dornsife College of Letters, Arts, and Sciences, the U.S. Department of Energy (DOE)
Computational Science Graduate Fellowship under
Award No. DE-SC0020347 and the ARO MURI grant
W911NF-22-S-000 during parts of this work.
This project was also supported in part by NSF award \#2210374.
The authors acknowledge the Center for Advanced Research Computing (CARC) at the University of Southern California for providing the computing resources used in this work.

\appendix

\section{Code implementation details}\label{Xsec52-B.1}
\label{app:code-details}
Our code is open source~\cite{ezzell2025code} and builds upon open source PMR-QMC code to simulate arbitrary spin-1/2 Hamiltonians~\cite{barash2024QuantumMonteCarlo, barash2024PmrQmcCode}. As such, it is straightforward to use. We describe the procedure used to generate the results in {\cref{sec:numerical-demo}} below, including both the QMC and exact calculation details.

\subsection{Building and running a QMC simulation}\label{Xsec50}\label{Xsec53-B.2}
Initially, we prepared a file, \verb^H.txt^, which contains a human-readable Pauli description of the TFIM on a square lattice. For example,
\begin{verbatim}
-1.0 1 Z 2 Z
0.5 1 X
0.5 2 X
\end{verbatim}
encodes the 2 qubit TFIM, $H = - Z_1 Z_2 + 0.5 X_1 + 0.5 X_2$. We then edit a simple \verb^parameters.hpp^ file. This contains simulation parameters such as the number of Monte--Carlo updates, inverse temperature, and standard observables we wish to estimate. As an example, \verb^#define MEASURE_HDIAG_CORR^ is flag that instructs our program to estimate $\avg{\Hdiag(\tau) \Hdiag}$. To encode non-standard, custom observables, such as $A,B$ in {\eqref{eq:verify-A}} and \eqref{eq:verify-B}, one simply writes an \verb^A.txt^ and \verb^B.txt^ in the same format as \verb^H.txt^ described above.

One can then compile and execute a fixed \verb^C++^ program that reads in the above input files. When the simulation finishes, a simulation summary containing observable summary lines,
\begin{verbatim}
Total of observable #1: A
Total mean(O) = -0.833214286
Total std.dev.(O) = 0.00599454762
\end{verbatim}
alongside various meta-data such as allocated CPU time is printed to the console. Among this meta-data includes,
\begin{verbatim}
Total mean(sgn(W)) = 1
Total std.dev.(sgn(W)) = 0

Testing thermalization
Observable #1: A, mean of std.dev.(O) = 0.016684391,
 std.dev. of mean(O) = 0.0127192673: test passed
\end{verbatim}
which are the average sign of PMR-QMC weights and results of simple thermalization testing if at least 5 MPI cores are used, respectively. Furthermore, our code automatically computes derived quantities via standard jackknife binning analysis when relevant~\cite{berg2004introduction}. For example, if the standard observables $\avg{H^2}$ and $\avg{H}$ are both estimated, then our code automatically estimates the specific heat, $C_v = \beta^2(\avg{H^2} - \avg{H}^2)$, by default.

More details on how to use our code and reproduce the experiments and plots performed in this manuscript are contained in the README file of our open-source code~\cite{ezzell2025code}.

\subsection{Details of our exact calculations}\label{Xsec51}\label{Xsec54-A.2}

{ For the $3 \times 3$ TFIM results, we compare our QMC simulation results to ``exact (numerical) calculations'' for verification. By this, we refer to direct numerical linear algebra computations up to machine precision using the Python libraries Numpy}~\cite{harris2020array} { and SciPy}~\cite{virtanen2020scipy}. Specifically, we compute $e^{-\beta H}$ using the SciPy function \verb`scipy.linalg.expm` and then $\Tr[e^{-\beta H}]$ with the NumPy function \verb`numpy.linalg.trace`. We can thus compute and store the thermal state $\rho_\beta = e^{-\beta H} / \Tr[e^{-\beta H}]$ which can use to compute $\Tr[O \rho_\beta]$ for any static operator $O$ via matrix multiplication and the trace. For the dynamic observables, we performed numerical integration using SciPy's \verb`scipy.integrate.quad` and verified that the estimated integration error was small. Additional implementation details are provided in our open source code~\cite{ezzell2025code}, specifically in the \verb`utils/exact_calculations.py` file.

\section{Divided difference integral relation proofs}\label{Xsec55-B}
\label{app:DDE-integral-proofs}
We provide proofs of the claimed integral DDE relations from the main text.

\subsection{The convolution theorem or energy susceptibility integral}\label{Xsec52}\label{Xsec56-B.1}
We show the claim,
\begin{equation}
 \label{app-eq:convolution-theorem}
 \int_0^\beta e^{-\tau [x_{j+1}, \ldots, x_q]} e^{-(\beta-\tau) [x_0, \ldots, x_j]} \dtau = - e^{-\beta [x_0, \ldots, x_q]}.
\end{equation}
A concise proof given by the present authors in Ref.~\cite{zeng2025inequalities} shows this via the convolution property of the Laplace transform. We provide a slightly expanded version here for clarity. For convenience, we define the functions
\begin{align}
 f(t) &= e^{-t [x_{j+1}, \ldots, x_q]} \\
 g(t) &= e^{-t [x_0, \ldots, x_j]}
\end{align}
The convolution of these functions,
\begin{align}
 (f * g)(t) &= \int_0^t f(\tau) g(t - \tau) \dtau
\end{align}
is by construction the integral we want to evaluate for $t = \beta$. Let $\lap{f(t)}$ denote the Laplace transform of $f(t)$ from $t \rightarrow s$. By the convolution property of the Laplace transform and {\eqref{eq:laplace-of-dd}}, we find
\begin{align}
 \lap{(f * g)(t)} &= \lap{f(t)} \lap{g(t)} \\
 &= \left( \frac{ (-1)^{q-j-1} }{ \prod_{l=j+1}^q (s + x_l) } \right) \left( \frac{ (-1)^{j} }{ \prod_{m=0}^j (s + x_m) }\right) \\
 &= \frac{ (-1)^{q-1} }{ \prod_{l=0}^q (s + x_l) } \\
 &= \lap{ - e^{-t [x_0, \ldots, x_q]} }.
\end{align}
Taking the inverse Laplace transform of the first and final expression proves the claimed integral relation.

We note that a direct proof by series expanding both DDE via {\eqref{eq:dd-as-series-expansion}}, integrating term-by-term, regrouping, and simplifying is also possible.

\subsection{The fidelity susceptibility integral}\label{Xsec53}\label{Xsec57-B.2}
We show the claim,
\begin{align}
 \label{app-eq:fs-integral}
 \int_0^{\beta/2} \tau e^{-\tau [x_{j+1}, \ldots, x_q]} e^{-(\beta-\tau) [x_0, \ldots, x_j]} \dtau = \sum_{r=0}^j e^{-\frac{\beta}{2}[x_0, \ldots, x_r]} \sum_{m = j + 1}^q e^{-\frac{\beta}{2} [x_r, \ldots, x_q, x_m]},
\end{align}
by providing an expounded version of the proof first shown in Ref.~\cite{ezzell2025universal}. In order to show this, we first prove,
\begin{equation}
 \label{app-eq:repeated-sum-simplification}
 t e^{-t [x_0, \ldots, x_j]} = - \sum_{m=0}^j e^{-t [x_0, \ldots, x_j, x_m]},
\end{equation}
via the Laplace transform. Let $\lap{ f(t)}$ denote the Laplace transform from $t \rightarrow s$. From the frequency-domain property of the Laplace transform, {\eqref{eq:laplace-of-dd}}, and algebra,
\begin{align}
 \lap{t e^{-t [x_0, \ldots, x_j]} } &= -\partial_{s} \lap{ e^{-t [x_0, \ldots, x_j]} } \\
 &= -\partial_s \left( \frac{ (-1)^j }{ \prod_{k=0}^j (s + x_k) } \right) \\
 &= -\sum_{m=0}^j \frac{(-1)^{j+1}}{ (s + x_m) \prod_{k=0}^j (s + x_k)} \\
 &= \mathcal{L}\left\{ -\sum_{m=0}^j e^{-t [x_0, \ldots, x_j, x_m]} \right\}.
\end{align}
Taking the inverse Laplace transform of the first and last expression, we have thus shown {\eqref{app-eq:repeated-sum-simplification}}. We can now show {\eqref{app-eq:fs-integral}} in three steps,
\begin{align}
 &\int_0^{\beta/2} \tau e^{-\tau [x_{j+1}, \ldots, x_q]} e^{-(\beta-\tau) [x_0, \ldots, x_j]} \dtau = \sum_{r=0}^j e^{-(\beta/2 - \tau) [x_0, \ldots, x_r]}\notag\\&\quad \int_0^{\beta/2} \tau e^{-\tau [x_{j+1}, \ldots, x_q]} e^{-(\beta/2 - \tau) [x_r, \ldots, x_j]}\dtau \\
 &= \sum_{r=0}^j \sum_{m=j+1}^q e^{-(\beta/2 - \tau) [x_0, \ldots, x_r]} \int_0^{\beta/2} e^{-\tau [x_{j+1}, \ldots, x_q, x_m]} e^{-(\beta/2 - \tau) [x_r, \ldots, x_j]} \dtau \\
 &= \sum_{r=0}^j e^{-(\beta/2 - \tau) [x_0, \ldots, x_r]} \sum_{m=j+1}^q e^{-\frac{\beta}{2} [x_r, \ldots, x_q, x_m] },
\end{align}
where the first line follows from the Leibniz rule ({\eqref{eq:leibniz-rule}}), the second by {\eqref{app-eq:repeated-sum-simplification}}, and final line by {\eqref{app-eq:convolution-theorem}} for $t = \beta/2$.

We note that a more elaborate direct evaluation of the left-hand-side of {\eqref{app-eq:fs-integral}} is also possible. This direct approach begins by evaluating
\begin{equation}
 \int_0^{\beta/2} \tau e^{-\tau [x_{j+1}, \ldots, x_q]} e^{-(\beta/2-\tau) [x_0, \ldots, x_j]} \dtau
\label{Xeqn120-B.18}
\end{equation}
by expanding the DDEs via {\eqref{eq:dd-as-series-expansion}}, integrating term-by-term, applying divided difference tricks, regrouping, and simplifying. Given such an explicit integration, one can then apply the Leibniz rule to the $e^{-(\beta-\tau) [\ldots]}$ result as we did above to derive an alternative explicit, closed form solution to {\eqref{app-eq:fs-integral}}. However, this result is more complicated and more difficult to derive. We note that when attempting to simplify our result using divided difference relations, we obtained an elaborate proof of the relation,
\begin{equation}
 \int_0^\beta \tau \avg{O(\tau) O} \dtau = \frac{\beta}{2} \int_0^\beta \avg{O(\tau) O} \dtau.
\label{Xeqn121-B.19}
\end{equation}
This integral relation can alternatively be derived more straightforwardly using integration by parts or direct integration of matrix elements without invoking PMR-QMC or divided differences. This observation illustrates that the direct approach, while valid, is unnecessarily circuitous.

\bibliography{bibliography}

\end{document}